\documentclass[11pt]{amsart}

\usepackage[T1]{fontenc}
\usepackage{microtype}      % subtle typographic improvements
\usepackage{placeins}
\usepackage{amsmath,amssymb,amsfonts,amsthm}
\allowdisplaybreaks

\usepackage{mathtools}
\usepackage{mathrsfs}
\usepackage{bm}

\numberwithin{equation}{section}

\usepackage[margin=1in]{geometry}

\usepackage{graphicx}
\usepackage{enumitem}

\usepackage[numbers,sort&compress]{natbib}

\usepackage[pdftex,hyperfootnotes=false]{hyperref}
\hypersetup{
	colorlinks=true,     % black text, no colored links
	pdfstartview=FitH,
	bookmarksopen=true,
	bookmarksnumbered=true
}

\theoremstyle{plain}
\newtheorem{theorem}{Theorem}[section]
\newtheorem*{theorem*}{Theorem}
\newtheorem{maintheorem}{Theorem}

\newtheorem{maincorollary}[maintheorem]{Corollary}

\newtheorem{proposition}[theorem]{Proposition}
\newtheorem{lemma}[theorem]{Lemma}
\newtheorem{corollary}[theorem]{Corollary}

\theoremstyle{definition}
\newtheorem{definition}[theorem]{Definition}
\newtheorem{assumption}[theorem]{Assumption}

\theoremstyle{remark}
\newtheorem{remark}[theorem]{Remark}

\newcommand{\CC}{\mathbb{C}}

\newcommand{\N}{\mathbb{N}}

\newcommand{\univ}{\mathrm{univ}}

\DeclareMathOperator{\diag}{diag}

\DeclareMathOperator{\Span}{Span}

\title[Local Rank-One Logarithmic Instability]
{Local Rank-One Logarithmic Instability for the Mixed Hessian of the Dispersionless Toda $\tau$-Function}

\author{Oleg Alekseev}
\address{Department of Mathematics, HSE University (National Research University Higher School of Economics), Moscow, Russia}

\address{Chebyshev Laboratory, Department of Mathematics and Computer Science, Saint Petersburg State University, Saint Petersburg, Russia}

\date{}

\begin{document}

\begin{abstract}
We study a weighted renormalization of the mixed Hessian of the dispersionless Toda $\tau$-function associated with polynomial conformal maps. The starting point is an explicit logarithmic-kernel representation, which yields a decomposition of the Hessian into symmetry blocks and reduces the spectral analysis to the inverse-map generating function $U(x;\zeta)$ and the geometry of its dominant singularities.
Near a transversal subcritical approach to a simple analytic critical point, we identify a rank-one logarithmic spectral instability:
in each renormalized symmetry block, exactly one variational eigenvalue diverges logarithmically, whereas the remaining variational eigenvalues stay bounded.
The proof isolates the analytic mechanism behind this transition in the emergence of a dominant $s$-orbit of simple square-root branch points of the Taylor branch of the inverse map. We then apply the same framework to reduced Laplacian-growth trajectories and show that the same spectral transition occurs there under the same local continuation hypotheses. If, in addition, the reduced map remains univalent at the spectral crossing, then this transition occurs before geometric breakdown.
The result is local and conditional: it identifies the mechanism of the first instability and formulates an abstract criterion for extensions beyond the polynomial class.

\end{abstract}

\maketitle

%=====================================================================
%==================

\section{Introduction}
\label{sec:intro}

%==================

Logarithmic kernels attached to a conformal map organize several different
structures at once: coefficient theory \cite{Duren1983}, the Grunsky
formalism \cite{Pommerenke1992,Schiffer1981}, potential theory and Green
functions \cite{Bell1992}, and, in the integrable framework, the generating
identities of the dispersionless \(2\)D Toda hierarchy
\cite{WiegmannZabrodin2000,MarshakovWiegmannZabrodin2002,Teo2003,Teo2009,TakhtajanTeo2006}.
For an exterior conformal map \(f\) with inverse map \(w=w(z)\) near
infinity, one naturally encounters holomorphic kernels of Grunsky--Faber type,
as well as the mixed kernel
\[
K(z,\bar z')
:=
\log\!\left(1-\frac{1}{w(z)\,\overline{w(z')}}\right).
\]
The holomorphic kernels control the purely holomorphic coefficient theory of
the map and the associated Grunsky operator, whereas the mixed kernel couples
one holomorphic and one anti-holomorphic variable. It is this mixed kernel
that is the subject of the present paper.

There are two complementary reasons for focusing on \(K(z,\bar z')\).  The
first comes from potential theory. In the exterior Dirichlet problem, the Green function is naturally expressed
in terms of logarithmic kernels of this type, while the mixed kernel reflects
the coupling between the holomorphic and antiholomorphic boundary components
\cite{Bell1992,Jeong2007}.  The second
comes from integrability.  In the dispersionless \(2\)D Toda hierarchy, the
quasiclassical \(\tau\)-function is the generating potential for the same
Dirichlet problem and for Laplacian-growth dynamics
\cite{MineevWeinstein2000,WiegmannZabrodin2000,MarshakovWiegmannZabrodin2002,BOOKtau,TakasakiTakebe1995,Teo2003,Teo2009}.
The dispersionless Hirota relations identify the mixed second derivatives of
\(\log\tau\) with the coefficients of the mixed kernel.  Thus the kernel is
the generating object for the
mixed sector of the hierarchy.

More concretely, if \(r\) is the conformal radius of the exterior conformal map at infinity, then the
expansion of the mixed kernel at infinity has the form
\begin{equation}\label{eq:intro-mixed-kernel}
	\log\!\left(1-\frac{1}{w(z)\,\overline{w(z')}}\right)
	=
	-\sum_{m,n\ge1} H_{mn}\,
	\frac{(z/r)^{-m}}{m}\frac{(\bar z'/r)^{-n}}{n}.
\end{equation}
The coefficients satisfy
\(H_{mn}=r^{-(m+n)}\,\partial_{t_m}\partial_{\bar t_n}\log\tau\), so the matrix
\(H=(H_{mn})_{m,n\ge1}\) is precisely the scale-free mixed
\(t_m\bar t_n\)-Hessian of the dispersionless Toda free energy.  From the
integrable point of view, these coefficients measure the response of the free
energy to coupled deformations of the holomorphic and anti-holomorphic
harmonic moments.  From the operator-theoretic point of view, \(H\) is the
natural mixed susceptibility matrix of the problem.

This interpretation explains why the spectrum of \(H\) is meaningful. The
individual entries \(H_{mn}\) describe how a single holomorphic mode \(t_m\)
couples to a single anti-holomorphic mode \(\bar t_n\). The spectrum,
however, detects collective mixed directions. A large eigenvalue corresponds
to a distinguished linear combination of moment deformations producing a large
response of the Toda free energy, while a diverging eigenvalue signals a
singular mixed fluctuation direction. The problem studied in this paper is to
determine how analytic degeneration of the inverse conformal map is reflected
in the spectrum of the mixed Hessian generated by the dispersionless Hirota
equations.

The operator considered here should be distinguished from the classical
Grunsky operator \cite{Schiffer1981,Shen2010,JohanssonViklund2023}. Both arise from logarithmic kernels attached to the
same conformal map, but they belong to different analytic frameworks.  The
Grunsky operator is holomorphic and belongs to the coefficient theory of
univalent maps.  The present paper instead studies the mixed kernel
\eqref{eq:intro-mixed-kernel} and the mixed \(t_m\bar t_n\)-Hessian it
generates.  The comparison with Grunsky theory is therefore geometric rather
than spectral: the object considered here is the mixed Hessian intrinsic to
the dispersionless Toda/Dirichlet formalism.

In this paper we work on a finite polynomial leaf of exterior conformal maps
\begin{equation}\label{eq:intro-map}
	f(w)=rw+\sum_{n=1}^{N} a_n w^{1-s_n},
	\qquad
	2\le s_1<\cdots<s_N,
\end{equation}
with reduced parameters \(\zeta_n:=a_n/r\).  Writing \(x=r/z\) and
introducing
\(U(x;\zeta):=\frac{z}{r}\frac{1}{w(z)}\), one finds that the Taylor branch of
the inverse map near infinity is determined by the algebraic equation
\begin{equation}\label{eq:intro-inv}
	U(x;\zeta)
	=
	1+\sum_{n=1}^{N}\zeta_n\,x^{s_n}U(x;\zeta)^{s_n},
\end{equation}
see Proposition~\ref{prop:inv-eq}. We denote by \(\rho_*(\zeta)\) the radius
of analyticity of this branch. The paper studies the spectral behavior of the
mixed Hessian as the inverse branch \(U(\cdot;\zeta)\) approaches a simple
dominant singularity on \(|x|=\rho_*(\zeta)\).

A basic feature of the polynomial class is symmetry. If
\(s:=\gcd(s_1,\dots,s_N)\), then the mixed Hessian splits into \(s\)
independent blocks according to residue classes modulo \(s\); see
Proposition~\ref{prop:block-decomposition-scale}. The spectral analysis is
therefore naturally blockwise. Within this framework, a simple dominant
square-root singularity of the inverse map gives rise to a distinguished
singular direction in each symmetry sector. The paper identifies the
local mechanism responsible for this phenomenon on polynomial leaves.

The explicit one-harmonic family \(f(w)=rw+aw^{1-s}\), analyzed in
\cite{Alekseev2025-1}, provides the solvable case behind the present work.
In that family the inverse map is globally explicit, the analytic threshold
\(\zeta_c\) can be located exactly, each symmetry block develops precisely one
logarithmically diverging eigenvalue there, and the same analysis shows that
spectral criticality may occur strictly before geometric loss of univalence.
The present paper identifies a local mechanism by which the spectral phenomenon observed in that explicit case extends to general polynomial leaves under an explicit continuation hypothesis.

The main result of the paper is a local universality theorem for the symmetry
blocks of the mixed Hessian in a fixed weighted realization. The theorem concerns subcritical paths approaching a simple point of
 \(\mathcal C\). At such a point, the first analytic obstruction occurs at
radius one, a unique dominant
$s$-orbit has been fixed, and the uniform
continuation hypotheses of Definition~\ref{def:uniform-delta} are satisfied. Within this framework, the
paper proves a rank-one logarithmic instability theorem for the renormalized
symmetry blocks, together with a Laplacian-growth corollary along trajectories
for which the same local assumptions hold. A compact-limit refinement
for the remainder and an abstract \(N=\infty\) criterion are considered
separately.

\subsection{Main results}
\label{subsec:main-results}

All spectral statements in this subsection are formulated for the fixed
weighted block realization introduced in
Section~\ref{subsec:weighted-space}. The theorem assumes a
simple critical point of the inverse-map equation, a unique dominant
\(s\)-orbit of simple square-root branch points along the chosen subcritical
path, and the uniform continuation assumptions of
Definition~\ref{def:uniform-delta}. Under these hypotheses, each symmetry
block exhibits the same logarithmic rank-one instability.

\begin{maintheorem}[Conditional rank-one logarithmic instability in a fixed weighted realization]
\label{thm:A}
Let \(\varepsilon \mapsto \zeta(\varepsilon)\) be a \(C^{1}\) path approaching a
simple critical point \(\zeta_{c}\in\mathcal C\) from the subcritical side, and
assume that the path is transversal in the sense of
Definition~\ref{def:transversal-path}. Suppose moreover that, for all
sufficiently small \(\varepsilon>0\), the Taylor branch
\(U(\cdot;\zeta(\varepsilon))\) has a unique dominant \(s\)-orbit of simple
square-root branch points and satisfies the uniform continuation assumptions of
Definition~\ref{def:uniform-delta}.

Fix an admissible weighted realization as in
Section~\ref{subsec:weighted-space}, that is, choose \(\beta>0\) together
with an admissible exponential weight \(\alpha\) for the subcritical block
realization. Then, in each symmetry block, the renormalized Gram block
\(\widetilde G^{(q)}(\varepsilon)\) acting on \(\ell^2(\mathbb N_0)\) develops
exactly one variational eigenvalue that diverges on the logarithmic scale,
while all higher variational eigenvalues remain bounded.

More precisely, if \(\mu_{k}^{(q)}(\varepsilon)\) denotes the \(k\)-th min--max
level of the \(q\)-th renormalized block, then
\[
\mu_{1}^{(q)}(\varepsilon)
=
\Gamma^{(q)}\,L(\varepsilon)+O(1),
\qquad
L(\varepsilon)
=
\log\frac{1}{\rho_{*}(\zeta(\varepsilon))-1}+O(1),
\]
where \(\Gamma^{(q)}>0\) depends on the fixed weighted realization, whereas
\(\mu_{k}^{(q)}(\varepsilon)\) remains bounded for every \(k\ge 2\). This is
proved in Theorem~\ref{thm:universality}.
\end{maintheorem}

At the operator level, the proof proceeds by extracting a singular rank-one
term from the block Gram representation of the Hessian. More precisely,
Proposition~\ref{prop:rank-one-extraction} shows that each weighted block
admits a decomposition into a logarithmically divergent rank-one part and a
remainder which stays uniformly bounded. Lemma~\ref{lem:compact-remainder}
strengthens this by proving that, after subtraction of the singular rank-one
term, the remainder converges in operator norm to a compact limit. Thus the
instability is genuinely one-dimensional at leading order: one direction
diverges logarithmically, whereas the rest of the spectrum remains controlled.

\begin{maincorollary}[Spectral instability along Laplacian growth]
\label{cor:B}
Let \(\zeta(T)\) be a reduced Laplacian-growth trajectory on a fixed
polynomial leaf, and suppose that it meets the restricted analytic critical
locus \(\mathcal C\) transversally at time \(T_c\) while the hypotheses of
Theorem~\ref{thm:A} hold along the left-hand approach \(T\uparrow T_c\). Then
each symmetry block of the renormalized mixed Hessian develops exactly one
logarithmically diverging variational eigenvalue as \(T\uparrow T_c\), while
all higher min--max levels remain bounded. This is the content of
Theorem~\ref{thm:lg-spectral}, viewed as the Laplacian-growth realization of
the local mechanism described in Theorem~\ref{thm:A}.
\end{maincorollary}

A further consequence concerns the possible separation between spectral
criticality and geometric breakdown. Proposition~\ref{prop:zc-before-zuniv}
shows that if the reduced map is still univalent at the spectral critical time
\(T_c\), then the spectral transition occurs strictly before the geometric
threshold at which univalence is lost:
\[
T_c<T_{\mathrm{univ}}.
\]
This is a genuine assumption: for the one-harmonic family studied in
\cite{Alekseev2025-1}, it can be verified on an explicit parameter range. Whenever such a
separation occurs on a finite polynomial leaf, the same local rank-one
mechanism governs the spectral transition.

The proof of Theorem~\ref{thm:A} has three steps. First, one analyzes the
characteristic system associated with \eqref{eq:intro-inv}. At a simple
critical point this yields a dominant square-root singularity of the inverse
map. Second, one transfers this local branching structure to uniform
asymptotics for the coefficients of the powers \(U(x;\zeta)^p\). Third, one
inserts these asymptotics into the blockwise Gram representation of the mixed
Hessian and extracts the singular contribution at the operator level.
Corollary~\ref{cor:B} and Proposition~\ref{prop:zc-before-zuniv} reinterpret the same local instability along a
distinguished dynamical trajectory and compare it with geometric loss of
univalence.

Two points about the scope of the theorem are worth emphasizing. First, the
analysis is local: the paper does not provide a general criterion ensuring the
dominance and continuation hypotheses on an arbitrary polynomial leaf.
Second, the operator statement is formulated after a fixed weighted
renormalization. The qualitative rank-one conclusion is therefore intrinsic at
the level of the extracted singular mechanism, whereas the concrete block
realization and the coefficient \(\Gamma^{(q)}\) depend on the chosen
admissible weights.

A secondary theme of the paper is the possible extension of this mechanism
beyond finite polynomial leaves. Section~\ref{subsec:infty-extension}
formulates an abstract \(N=\infty\) criterion on the Gram side, isolating the
operator-theoretic mechanism needed to extend the rank-one extraction to
infinite-dimensional leaves. The explicit pole and logarithmic families
discussed in Section~\ref{sec:infty-explicit-leaves} should be read in this
spirit.

\subsection{Organization of the paper}

\label{subsec:organization}

Section~\ref{sec:framework} introduces the polynomial conformal-map class,
derives the inverse-map equation, identifies the mixed kernel with the mixed
Toda Hessian, and proves the symmetry-block and Gram decompositions.
Section~\ref{sec:alg-criticality} analyzes the characteristic system of
\eqref{eq:intro-inv}, identifies simple square-root branching, and derives the
coefficient asymptotics associated with a dominant orbit.
Section~\ref{sec:universal-instability} contains the operator-theoretic core
of the paper: the weighted block realization, the rank-one extraction, the
proof of the universality theorem, and the abstract \(N=\infty\) criterion.
Section~\ref{sec:laplacian-growth} interprets the local theory along
Laplacian-growth trajectories and clarifies the separation between spectral
criticality and geometric loss of univalence.
Section~\ref{sec:infty-explicit-leaves} discusses explicit infinite-dimensional
pole and logarithmic leaves as illustrative characteristic cases.
Section~\ref{sec:conclusion} summarizes the picture and outlines directions
for further work.  The appendices collect the uniform transfer estimates and
operator bounds used in the proof of the main theorem.

% =======================

\section{Framework: Conformal Maps and the Hessian}
\label{sec:framework}

% =======================

\subsection{Polynomial conformal maps and reduced parameters}

We consider conformal maps from the exterior disk to the complement of a bounded simply connected
domain,
\begin{equation}\label{eq:framework-map}
	f(w)=r w+\sum_{n=1}^{N} a_n\,w^{1-s_n},
	\qquad |w|>1,
\end{equation}
where $r>0$ is the conformal radius at infinity, the translational degree of freedom has been
removed (no constant term), and
\[
	2\le s_1<\cdots<s_N.
\]
Throughout, the exponents $\{s_n\}_{n=1}^N$ are fixed (a polynomial leaf), while the coefficients
vary.

It is convenient to work with the scale--free parameters
\begin{equation}\label{eq:reduced-parameters}
	\zeta_n:=\frac{a_n}{r},\qquad n=1,\dots,N,
\end{equation}
and we regard
\[
	\zeta:=(\zeta_1,\dots,\zeta_N)\in\CC^N
\]
as the parameter vector. This removes the trivial scaling dependence carried by $r$. In particular,
all scale--invariant objects introduced below depend on $\zeta$ only.

We also set
\begin{equation}\label{eq:sdef}
	s:=\gcd(s_1,\dots,s_N).
\end{equation}
The integer $s$ is the symmetry index: it will govern an $s$--fold rotational invariance of the
inverse--map equation and hence the block decomposition of the mixed Hessian.

\begin{remark}[Infinite--dimensional families]\label{rem:infty-families}
	Although most of the paper works with $N<\infty$ (so that the inverse--map equation is algebraic),
	the universality mechanism uses only analytic continuation of the Taylor branch $U(\cdot;\zeta)$,
	the dominant square--root singularity scenario, and uniform control of the remainder in an analytic
	Gram representation. Subsection~\ref{subsec:infty-extension} presents the corresponding extension to
	$N=\infty$ under explicit analytic hypotheses.
\end{remark}

\subsection{Inverse map and normalized generating function}

Let $w(z)$ denote the inverse of $f$, defined for $|z|$ sufficiently large and normalized by
$w(z)\sim z/r$ as $z\to\infty$. Introduce the dimensionless coordinate
\[
	x:=\frac{r}{z},
\]
and define the normalized inverse--map generating function
\begin{equation}\label{eq:U-def}
	U(x;\zeta):=\frac{z}{r}\,\frac{1}{w(z)}.
\end{equation}
By construction, $U(x;\zeta)$ is analytic near $x=0$ and satisfies $U(0;\zeta)=1$.

Introduce the polynomial
\begin{equation}\label{eq:F-def}
	F(y,x;\zeta):=y-1-\sum_{n=1}^N \zeta_n\,x^{s_n}\,y^{s_n}.
\end{equation}
Algebraic criticality will be characterized in terms of the branch--point system for $F$
(Section~\ref{sec:alg-criticality}).

\begin{proposition}[Inverse map equation]\label{prop:inv-eq}
	The function $U(x;\zeta)$ is characterized as the unique germ analytic at $x=0$ with $U(0;\zeta)=1$
	solving
	\begin{equation}\label{eq:inv-eq}
		U(x;\zeta)=1+\sum_{n=1}^N \zeta_n\,x^{s_n}\,U(x;\zeta)^{s_n}.
	\end{equation}
	Equivalently, $F(U(x;\zeta),x;\zeta)=0$ with $F$ given by \eqref{eq:F-def}.
\end{proposition}

\begin{proof}
	Substituting $w(z)=\frac{z}{rU(x;\zeta)}$ into $z=f(w)$ gives
	\[
		z=\frac{z}{U}+\sum_{n=1}^{N}a_n\left(\frac{z}{rU}\right)^{1-s_n}.
	\]
	Dividing by $z$ and using $x=r/z$ and $\zeta_n=a_n/r$ yields
	\[
		1=\frac{1}{U}+\sum_{n=1}^{N}\zeta_n\,x^{s_n} U^{s_n-1}.
	\]
	Multiplying by $U$ yields \eqref{eq:inv-eq}. At $(x,U)=(0,1)$ the derivative of the left-hand side
	of \eqref{eq:inv-eq} with respect to $U$ equals $1$, hence the implicit function theorem yields a
	unique analytic solution $U(\cdot;\zeta)$ near $x=0$ for each fixed $\zeta$.
\end{proof}

\begin{lemma}[$s$--symmetry]\label{lem:s-symmetry}
	Let $s=\gcd(s_1,\dots,s_N)$. The inverse--map equation \eqref{eq:inv-eq} is invariant under
	$x\mapsto e^{2\pi i/s}x$. In particular, the Taylor branch satisfies
	\begin{equation}\label{eq:s-sym}
		U(e^{2\pi i/s}x;\zeta)=U(x;\zeta).
	\end{equation}
	Consequently, for each $p\ge1$ the expansion of $U(x;\zeta)^p$ involves only powers $x^{ms}$.
\end{lemma}

\begin{proof}
	Since each exponent $s_n$ is divisible by $s$, one has $(e^{2\pi i/s}x)^{s_n}=x^{s_n}$, hence the
	right-hand side of \eqref{eq:inv-eq} is unchanged under $x\mapsto e^{2\pi i/s}x$.
	The Taylor branch is uniquely determined by analyticity at $x=0$ and the normalization $U(0;\zeta)=1$,
	so \eqref{eq:s-sym} follows. The final claim is an immediate consequence of \eqref{eq:s-sym}.
\end{proof}

\begin{definition}[Radius of analyticity]\label{def:rho-star}
	Let $\rho_*(\zeta)$ denote the radius of the largest disk centered at $x=0$ on which the Taylor
	branch of \eqref{eq:inv-eq} is analytic:
	\begin{equation}\label{eq:rho-star}
		\rho_*(\zeta):=\sup\{\rho>0:\ U(\cdot;\zeta)\ \text{analytic on }|x|<\rho\}.
	\end{equation}
\end{definition}

For each fixed parameter value \(\zeta\), the phrase \emph{Taylor branch}
means the unique germ at \(x=0\) selected by \(U(0;\zeta)=1\). Whenever this
germ is analytically continued along a path, the resulting continuation is
called the \emph{chosen Taylor sheet} along that path. Statements about
boundary singularities, dominant orbits, or characteristic values always refer
to this continuation of the distinguished germ.

\subsection{Kernel representation of the mixed Hessian}

In the standard dispersionless Toda normalization, $\mathcal F:=\log\tau$ denotes the free energy and
$(t_m)_{m\ge1}$, $(\bar t_n)_{n\ge1}$ are the exterior harmonic moments (with the convention of
\cite{WiegmannZabrodin2000,BOOKtau,MarshakovWiegmannZabrodin2002}). In the present paper we use
only the mixed second derivatives $\partial_{t_m}\partial_{\bar t_n}\mathcal F$ and the kernel identity below.
Equivalent forms, in particular, the generating-operator identity
\[
D(z)\,\bar D(\bar z')\mathcal F=-\log\!\left(1-\frac{1}{w(z)\overline{w(z')}}\right), \quad D(z):=\sum_{m\ge1}\frac{z^{-m}}{m}\,\frac{\partial}{\partial t_m},
\]
can be found, for example, in
\cite[Eq.~(42)]{BOOKtau} and \cite[\S4]{MarshakovWiegmannZabrodin2002}. We state the
coefficient form explicitly because it is the entry point for the scale-free reduction used
throughout the paper.

\begin{proposition}[Kernel identity for the mixed Toda Hessian]
	\label{prop:kernel-hessian}
	Let $f$ be an exterior conformal map and let $w(z)\sim z/r$ be its inverse
map near $\infty$. In the standard dispersionless Toda normalization of
\cite{WiegmannZabrodin2000,BOOKtau,MarshakovWiegmannZabrodin2002}, one then has, for
	$m,n\ge1$,
	\begin{equation}\label{eq:hessian-kernel}
		\frac{\partial^2\log\tau}{\partial t_m\,\partial\bar t_n}
		=
		-mn\,[z^{-m}][\bar z'^{-n}]\,
		\log\!\left(1-\frac{1}{w(z)\overline{w(z')}}\right),
	\end{equation}
	where $z$ and $z'$ are independent complex variables and $[z^{-m}]$ denotes coefficient extraction
	in the Laurent expansion at infinity.
\end{proposition}

We use \eqref{eq:hessian-kernel} only as a recalled identity; for derivations and convention choices see
\cite{WiegmannZabrodin2000,BOOKtau,MarshakovWiegmannZabrodin2002}. In particular,
\cite[Eq.~(42)]{BOOKtau} gives the generating-operator form, from which \eqref{eq:hessian-kernel}
follows by coefficient extraction at infinity. The right-hand side of \eqref{eq:hessian-kernel} depends on the conformal radius $r$ through the
normalization $w(z)\sim z/r$ at infinity. For the spectral questions studied here this dependence is
inessential: it is absorbed by the scale-free variable $x=r/z$ introduced above. Since
$w(z)^{-1}=x\,U(x;\zeta)$, we introduce the scale-free kernel
\begin{equation}\label{eq:kernel-scaled}
	K(x,\bar x';\zeta)
	:=
	\log\Bigl(1-x\bar x'\,U(x;\zeta)\,\overline{U(x';\zeta)}\Bigr).
\end{equation}
This expression depends only on $(x,\bar x',\zeta)$, with no residual $r$-dependence.

\begin{definition}[Scale--invariant mixed Hessian]\label{def:scale-inv-hessian}
	The \emph{scale--invariant mixed Hessian} is the infinite Hermitian matrix
	\begin{equation}\label{eq:H-rescaled}
		H_{mn}(\zeta):=-mn [x^m][\bar x'^n] K(x,\bar x';\zeta),
		\qquad m,n\ge1.
	\end{equation}
	Equivalently, by \eqref{eq:hessian-kernel} and the change of variables $x=r/z$, one has
	\begin{equation}\label{eq:H-vs-tau}
		H_{mn}
		=
		r^{-(m+n)}\,
		\frac{\partial^2\log\tau}{\partial t_m\,\partial\bar t_n}.
	\end{equation}
	Thus $H=(H_{mn})$ is the natural scale-invariant form of the mixed Toda Hessian.
\end{definition}

	The scale-invariant Hessian $H=(H_{mn})_{m,n\ge1}$ is naturally a positive semidefinite
	quadratic form on finitely supported sequences. Spectral statements require a compatible Hilbert
	realization, and the standard $\ell^2(\N)$ topology is not stable at analytic criticality. The
	appropriate framework is the weighted Gram realization introduced later in
	Definition~\ref{def:conjugated-gram}. All spectral statements in
	Sections~\ref{sec:universal-instability}--\ref{sec:laplacian-growth} refer to these renormalized operators.

\subsection{Symmetry and block decomposition}
The $s$--symmetry \eqref{eq:s-sym} implies that, for each $p\ge1$,
\begin{equation}\label{eq:Sp-expansion}
	U(x;\zeta)^p=\sum_{m\ge0} R_p(m;\zeta)\,x^{ms}.
\end{equation}
The coefficients $R_p(m;\zeta)$ are the basic ingredients in the Gram representations of the Hessian
blocks.

\begin{remark}[One--mode case and Raney numbers]\label{rem:raney-one-mode}
	In the highly symmetric one--mode situation $N=1$ with exponent $s_1=s$ and parameter $\zeta=\zeta_1$,
	the inverse--map equation \eqref{eq:inv-eq} reduces to $U=1+\zeta\,x^{s}U^{s}$. In this case the
	coefficients $R_p(m;\zeta)$ factorize as
	\begin{equation}\label{eq:Rp-raney}
		U(x;\zeta)^p=\sum_{m=0}^{\infty} R_{s,p}(m)\,\zeta^{m}\,x^{sm},
		\qquad\text{i.e.}\qquad
		R_p(m;\zeta)=R_{s,p}(m)\,\zeta^{m},
	\end{equation}
	where $R_{s,p}(m)$ are the Raney (Fuss--Catalan) numbers associated with the parameters $(s,p)$.
	We refer to \cite{Alekseev2025-1} for this specialization and its role in the one--parameter spectral
	analysis, and to the classical literature on Fuss--Catalan/Raney numbers for explicit closed forms and interpretation~\cite{Raney1960,GrahamKnuthPatashnik1994}.
\end{remark}

\begin{proposition}[Block decomposition]\label{prop:block-decomposition-scale}
	The scale--invariant Hessian satisfies
	\[
		H_{mn}=0
		\qquad\text{unless}\qquad
		m\equiv n\pmod{s}.
	\]
	Equivalently,
	\[
		H=\bigoplus_{q=1}^{s} H^{(q)},
	\]
	where $H^{(q)}$ acts on indices $m\equiv q\pmod{s}$.
\end{proposition}

\begin{proof}
Expanding the logarithm in \eqref{eq:kernel-scaled} yields
\[
		K(x,\bar x';\zeta)
		=
		-\sum_{p\ge1}\frac{1}{p}\,(x\bar x')^{p}\,U(x;\zeta)^p\,\overline{U(x';\zeta)^p}.
	\]
	Using \eqref{eq:Sp-expansion} we obtain
	\[
		K(x,\bar x';\zeta)
		=
		-\sum_{p\ge1}\frac{1}{p}
		\sum_{m,n\ge0}
		R_p(m;\zeta)\,\overline{R_p(n;\zeta)}\,
		x^{p+ms}\,\bar x'^{\,p+ns}.
	\]
	Thus every monomial in this expansion has exponents congruent to $p$ modulo $s$ in both variables,
	and therefore the coefficient of $x^m\bar x'^{\,n}$ vanishes unless $m\equiv n\pmod{s}$. Because of \eqref{eq:H-rescaled}, the same congruence condition holds for $H$.
\end{proof}

\subsection{Gram representation}

For each residue class $q\in\{1,\dots,s\}$ define the index sequence
\[
	p_j:=q+js,\qquad j\ge0,
\]
and identify the block $\ell^2(\{p_j:j\ge0\})$ with $\ell^2(\N_0)$ via $j\leftrightarrow p_j$.
We write $H^{(q)}_{j_1j_2}:=H_{p_{j_1},p_{j_2}}$ for the corresponding block matrix elements.

\begin{proposition}[Algebraic Gram representation]\label{prop:gram-scale}
	For each $q\in\{1,\dots,s\}$ the block $H^{(q)}$ admits the exact Gram representation
	\begin{equation}\label{eq:Gram-alg}
		H^{(q)}_{j_1 j_2}
		=
		\sum_{p\equiv q\!\!\!\!\pmod{s}}
		v^{(p)}_{j_1}\,\overline{v^{(p)}_{j_2}},
	\end{equation}
	where
	\begin{equation}\label{eq:Gram-alg-vectors}
		v^{(p)}_j
		:=
		\begin{cases}
			\dfrac{p_j}{\sqrt p}\,R_p\!\left(\dfrac{p_j-p}{s};\zeta\right),
			   & \text{if } p_j\ge p \text{ (equivalently } \dfrac{p_j-p}{s}\in\N_0), \\[1.2ex]
			0, & \text{otherwise.}
		\end{cases}
	\end{equation}
\end{proposition}

\begin{proof}
	From the logarithmic expansion used in the proof of Proposition~\ref{prop:block-decomposition-scale},
	the coefficient of $x^{p+ks}\bar x'^{\,p+\ell s}$ in $K$ equals
	$-\frac{1}{p}R_p(k;\zeta)\overline{R_p(\ell;\zeta)}$. Therefore,
	\[
		H_{p+ks,\;p+\ell s}
		=
		\frac{(p+ks)(p+\ell s)}{p}\,
		R_p(k;\zeta)\,\overline{R_p(\ell;\zeta)}.
	\]
	Fix $q\in\{1,\dots,s\}$ and set $p_j=q+js$. Restricting to $m=p_{j_1}$ and $n=p_{j_2}$, and summing
	over those $p\equiv q\pmod{s}$ for which $p\le \min\{p_{j_1},p_{j_2}\}$, yields \eqref{eq:Gram-alg}
	with vectors \eqref{eq:Gram-alg-vectors}. For fixed $j_1,j_2$ only finitely many $p$ satisfy
	$p\le \min\{p_{j_1},p_{j_2}\}$, hence the sum is finite entrywise.
\end{proof}

\begin{remark}[Algebraic versus tail Gram representations]
	\label{rem:alg-vs-tail-gram}
	Proposition~\ref{prop:gram-scale} gives the mode-indexed representation
	\(H^{(q)}=VV^*\): for fixed matrix indices only finitely many logarithmic
	modes contribute. For the critical analysis one instead uses
	Proposition~\ref{prop:analytic-gram}, where the tail index \(m\) is the
	natural Hilbert-space index and the dual Gram block has the form \(G^{(q)}=V^*V\).
	By Remark~\ref{rem:HG-spectrum}, these dual realizations share the same nonzero
	spectrum. The dominant asymptotic
	$R_p(m;\zeta)\sim A_p(\zeta)m^{-3/2}\rho_*(\zeta)^{-ms}$ then produces the
	borderline sum $\sum_{m\ge0}\eta^m/(m+1)=L(\varepsilon)$ responsible for the
	logarithmic divergence.
\end{remark}

% ======================

\section{Algebraic Criticality and the Onset of Analytic Breakdown}
\label{sec:alg-criticality}

% =======================

The Gram representations of the Hessian reduce spectral questions to asymptotics of Taylor
coefficients of powers of the inverse--map generating function $U(\cdot;\zeta)$.  The latter are
governed by the \emph{first obstruction to analytic continuation} of the Taylor branch
$x\mapsto U(x;\zeta)$, which for algebraic functions is typically a square--root branch point.
This section presents the corresponding algebraic mechanism and the resulting coefficient asymptotics
used in Section~\ref{sec:universal-instability}.

\subsection{Characteristic points and simple square--root branching}
Recall the polynomial $F(y,x;\zeta)=y-1-\sum_{n=1}^{N}\zeta_n\,x^{s_n}y^{s_n}$
from \eqref{eq:F-def},
so that the Taylor branch is characterized by
\begin{equation}\label{eq:branch-eq}
	F(U(x;\zeta),x;\zeta)=0,
	\qquad
	U(0;\zeta)=1.
\end{equation}

\begin{proposition}[Branch points and the characteristic system]\label{prop:xi-star-system}
	Fix $\zeta$ and consider the Taylor branch $x\mapsto U(x;\zeta)$ determined by \eqref{eq:branch-eq}.
	Assume that this branch admits analytic continuation in a slit neighborhood of a point $x_*\in\CC$
	with $0<|x_*|<\infty$, and that $x_*$ is a singular point on the Taylor sheet. Let $\lambda$ denote
	the corresponding finite branch value attained by that continuation at $x_*$. Then $(x_*,\lambda)$
	satisfies the \emph{characteristic system}
	\begin{equation}\label{eq:char-system-F}
		F(\lambda,x_*;\zeta)=0,
		\qquad
		\partial_yF(\lambda,x_*;\zeta)=0.
	\end{equation}
	Equivalently,
	\begin{equation}\label{eq:char-system}
		\lambda
		=
		1+\sum_{n=1}^N \zeta_n\,x_*^{s_n} \lambda^{s_n},
		\qquad
		1
		=
		\sum_{n=1}^N s_n\,\zeta_n\,x_*^{s_n} \lambda^{s_n-1}.
	\end{equation}

	Conversely, suppose that the Taylor branch admits analytic continuation to a
	punctured neighborhood of $x_*$ with branch value $\lambda$, that
	$(x_*,\lambda)$ satisfies \eqref{eq:char-system-F}, and that the
	nondegeneracy conditions
	\begin{equation}\label{eq:Fx-nondeg}
		\partial_xF(\lambda,x_*;\zeta)\neq0,
	\end{equation}
	and
	\begin{equation}\label{eq:Fyy-nondeg}
		\partial_y^2F(\lambda,x_*;\zeta)\neq0
	\end{equation}
	hold. Then $x_*$ is a square--root branch point of the Taylor sheet: in a neighborhood of $x_*$ one has a
	local Puiseux expansion
	\[
		U(x;\zeta)=\lambda + c\,\sqrt{x-x_*}+O(x-x_*),
		\qquad c\neq0,
	\]
	for a suitable choice of the square--root branch.
\end{proposition}

\begin{proof}
	Deferred to Appendix~\ref{app:xi-star-system}.
\end{proof}

The later local theory uses branch points of the Taylor sheet with finite
branch value \(\lambda\) on the continued sheet. This is precisely the case
relevant for the dominant-orbit analysis carried out below.

\subsection{Dominant orbit and critical locus}

The $s$--symmetry (Lemma \ref{lem:s-symmetry}) implies that singularities occur in $s$--orbits.  For the spectral problem we
consider the situation where the analytic boundary is governed by a single dominant orbit.

\begin{definition}[Dominant $s$--orbit at the analytic boundary]\label{def:dominant-orbit}
	Fix $\zeta$ and assume the Taylor branch is analytic in $|x|<\rho_*(\zeta)$, where $\rho_*(\zeta)$ is
	as in Definition~\ref{def:rho-star}. Let $s:=\gcd(s_1,\dots,s_N)$ so that
	$U(e^{2\pi i/s}x;\zeta)=U(x;\zeta)$. We say that the analytic boundary is \emph{dominant modulo the $s$--symmetry} if the chosen Taylor sheet has
	exactly $s$ singular points on the circle $|x|=\rho_*(\zeta)$, forming a single orbit
	\[
		\bigl\{e^{2\pi i j/s}\,x_*(\zeta)\bigr\}_{j=0}^{s-1},
		\qquad |x_*(\zeta)|=\rho_*(\zeta),
	\]
	and no other singularities on that circle. We refer to $x_*(\zeta)$ as a chosen representative of the
	dominant orbit.
\end{definition}

\begin{definition}[Analytic critical locus]
\label{def:alg-critical-locus}
We define the analytic critical locus \(\mathcal C\subset\CC^N\) by
\[
\mathcal C
:=
\bigl\{
\zeta\in\CC^N:
\rho_*(\zeta)=1
\ \text{and}\
\zeta \text{ satisfies Definition~\ref{def:dominant-orbit}}
\bigr\}.
\]
Equivalently, \(\mathcal C\) consists of those parameters for which the first
obstruction to analytic continuation of \(U(\cdot;\zeta)\) occurs on the unit
circle \(|x|=1\) and is represented by a unique dominant \(s\)-orbit in the
sense of Definition~\ref{def:dominant-orbit}.
\end{definition}

In particular, \(\mathcal C\) need not coincide with the full radius-one locus
\(\{\rho_*=1\}\). It singles out the part of that locus for which a unique
dominant \(s\)-orbit has already been identified, which is the case used in the
spectral theorem below.

\begin{definition}[Simple critical point]\label{def:simple-critical}
	A parameter $\zeta_c\in\mathcal C$ is called a \emph{simple critical point}
	if, for one representative $x_{*,c}$ of the dominant
	orbit and the corresponding finite branch value $\lambda_c$ on the chosen Taylor
	sheet, one has
	\[
		F(\lambda_c,x_{*,c};\zeta_c)=0,
		\qquad
		\partial_yF(\lambda_c,x_{*,c};\zeta_c)=0,
		\qquad
		\partial_y^2F(\lambda_c,x_{*,c};\zeta_c)\neq0.
	\]
\end{definition}

For the polynomial $F(y,x;\zeta)=y-1-\sum_{n=1}^N \zeta_n x^{s_n}y^{s_n}$,
the characteristic equations imply
\[
	\partial_xF(\lambda,x_*;\zeta)=-\sum_{n=1}^N s_n\zeta_n x_*^{s_n-1}\lambda^{s_n}=-\frac{\lambda}{x_*}.
\]
Moreover, $F(\lambda_c,x_{*,c};\zeta_c)=0$ forces $\lambda_c\neq0$, since
$F(0,x;\zeta)=-1$. Hence the fold condition
$\partial_xF(\lambda_c,x_{*,c};\zeta_c)\neq0$ is automatic at a critical
point of the Taylor sheet, and only the nonvanishing of
$\partial_y^2F(\lambda_c,x_{*,c};\zeta_c)$ needs to be imposed.

\subsection{Distance to criticality and branch-point parameters}

Fix a norm on $\CC^N$ (all norms are equivalent).  In a neighborhood of a simple critical point $\zeta_c\in\mathcal C$
and within the dominant regime of Definition~\ref{def:dominant-orbit}, the critical locus is characterized by the analytic--radius condition
$\rho_*(\zeta)=1$.  We therefore measure proximity to criticality by
\begin{equation}\label{eq:eps-def}
	\varepsilon:=\rho_*(\zeta)-1.
\end{equation}
All asymptotic statements below are understood in the limit $\varepsilon\downarrow0^+$.

\begin{definition}[Subcritical, critical, supercritical]
\label{def:subcritical-critical-supercritical}
Within the dominant regime of Definition~\ref{def:dominant-orbit}, we call a
parameter \(\zeta\) \emph{subcritical} if \(\rho_*(\zeta)>1\),
\emph{critical} if \(\rho_*(\zeta)=1\), and \emph{supercritical} if
\(\rho_*(\zeta)<1\).
\end{definition}

\begin{definition}[Transversal approach]\label{def:transversal-path}
	Let $t\mapsto\zeta(t)$ be a $C^1$ path with $\zeta(0)=\zeta_c\in\mathcal C$.
	Assume that, near $t=0$, the path stays in the dominant regime of
	Definition~\ref{def:dominant-orbit} and that the representative dominant
	branch point continues $C^1$-smoothly along the path, so that
	$\rho_*(\zeta(t))$ is differentiable there. By
	Lemma~\ref{lem:local-branch-data}, this holds whenever the continued point
	remains a representative of the unique dominant orbit.
	We say that $\zeta(t)$ satisfies the \emph{scalar transversality condition} at
	$\zeta_c$ if
	\[
	\frac{d}{dt}\rho_*\bigl(\zeta(t)\bigr)\Big|_{t=0}\neq0.
	\]
	Then, on either side of the crossing where $\rho_*(\zeta(t))>1$, one may use
	\[
		\varepsilon:=\rho_*\bigl(\zeta(t)\bigr)-1
	\]
	as a local parameter. After this reparametrization the path may be written as
	$\varepsilon\mapsto\zeta(\varepsilon)$ with
	$\rho_*(\zeta(\varepsilon))=1+\varepsilon$.
\end{definition}

Throughout the paper, ``transversal'' refers only to the scalar
nondegeneracy condition of Definition~\ref{def:transversal-path}, namely the
nonvanishing derivative of the analyticity radius once that derivative is
well defined along the chosen dominant continuation.

\begin{lemma}[Local continuation of the critical branch parameters]
\label{lem:local-branch-data}
Let \(\zeta_c\in\mathcal C\) be a simple critical point, and choose a
representative critical point \(x_{*,c}\) together with the corresponding
branch value \(\lambda_c\) on the chosen Taylor sheet. Then there exist a
neighborhood \(\mathcal O\) of \(\zeta_c\) and unique \(C^1\) maps
\[
\zeta\mapsto x_*(\zeta),\qquad
\zeta\mapsto \lambda(\zeta),\qquad
\zeta\in\mathcal O,
\]
such that \((x_*(\zeta_c),\lambda(\zeta_c))=(x_{*,c},\lambda_c)\) and
\[
F(\lambda(\zeta),x_*(\zeta);\zeta)=0,
\qquad
\partial_yF(\lambda(\zeta),x_*(\zeta);\zeta)=0.
\]
If, in addition, the continued point \(x_*(\zeta)\) remains a representative
of the unique dominant orbit for \(\zeta\in\mathcal O\), then
\[
\rho_*(\zeta)=|x_*(\zeta)|
\]
on \(\mathcal O\), and hence \(\rho_*\) is \(C^1\) on \(\mathcal O\).
\end{lemma}

\begin{proof}
	Apply the implicit function theorem to the map
	\[
		(y,x,\zeta)\mapsto \bigl(F(y,x;\zeta),\partial_yF(y,x;\zeta)\bigr)
	\]
	in the variables $(y,x)$ at $(\lambda_c,x_{*,c};\zeta_c)$. Its Jacobian in
	$(y,x)$ equals
	\[
		\begin{pmatrix}
			\partial_yF & \partial_xF \\
			\partial_y^2F & \partial_x\partial_yF
		\end{pmatrix},
	\]
	whose determinant at the critical point is
	$-\partial_xF(\lambda_c,x_{*,c};\zeta_c)\,\partial_y^2F(\lambda_c,x_{*,c};\zeta_c)\neq0$
	because $\zeta_c$ is simple. This yields unique $C^1$ continuations
	$\lambda(\zeta)$ and $x_*(\zeta)$ near $\zeta_c$.

	This proves the first part of the statement, namely the local continuation of
	the characteristic pair. The additional identity
	$\rho_*(\zeta)=|x_*(\zeta)|$ does not follow from the implicit function theorem by itself. It also requires that the continued point remain a representative of the unique dominant orbit. Under that extra hypothesis one may shrink
	$\mathcal O$ so that the same representative branch is followed throughout,
	and then $\rho_*(\zeta)=|x_*(\zeta)|$ is $C^1$ because $x_*(\zeta)\neq0$.
\end{proof}

Thus the continuation of \((x_*,\lambda)\) is automatic near a simple critical
point, but the additional dominance and uniform \(\Delta\)-analytic control
required later are not. Those are precisely the extra hypotheses encoded in
Definition~\ref{def:uniform-delta}.
Along such a path we write
\[
	\rho_*(\varepsilon):=\rho_*(\zeta(\varepsilon)),\qquad x_*(\varepsilon):=x_*(\zeta(\varepsilon)),
\]
where $x_*(\zeta)$ is the $C^1$ representative branch from
Lemma~\ref{lem:local-branch-data}, chosen within the dominant orbit of
Definition~\ref{def:dominant-orbit}.
The corresponding branch values on the Taylor sheet will be used as auxiliary variables:
\begin{equation}\label{eq:lambda-def}
	\lambda(\varepsilon)
	:=
	\lim_{x\to x_*(\varepsilon)}
	U(x;\zeta(\varepsilon)),
\end{equation}
where $x$ is on the chosen Taylor sheet.

\subsection{Transfer asymptotics from a dominant simple orbit}

We next translate the square--root singularity at the dominant branch point
(Proposition~\ref{prop:xi-star-system}) into coefficient asymptotics for powers of the inverse map.
The required estimates follow from standard singularity analysis (transfer theorems) in
$\Delta$--domains; we use \cite[Ch.~VI]{FlajoletSedgewick2009} as a convenient reference.
These results extract the large-order asymptotics
of Taylor coefficients from the local singular expansion at the dominant singularity
and provide the coefficient estimates needed for the spectral analysis below.

\begin{definition}[$\Delta$--domains and $\Delta$--analyticity]\label{def:delta-domain}
	Fix a point $\xi_*\in\CC\setminus\{0\}$ and an opening angle $\phi\in(0,\pi/2)$. A \emph{$\Delta$--domain at $\xi_*$} is a domain
	of the form
	\[
		\Delta(\xi_*,\phi)
		:=
		\Bigl\{x\in\CC:\ |x|<|\xi_*|+\eta,\ x\neq\xi_*,\ |\arg(\xi_*-x)|>\phi\Bigr\},
	\]
	for some thickness parameter $\eta>0$. A function $G$ is \emph{$\Delta$--analytic at $\xi_*$} if it admits an analytic
	continuation to $\Delta(\xi_*,\phi)$ for some $\phi\in(0,\pi/2)$.
\end{definition}

For each fixed $\zeta$, the Taylor sheet $U(\cdot;\zeta)$ is an algebraic function of $x$ (it solves
\eqref{eq:F-def}). Hence its analytic continuation has only algebraic branch points and admits local
Puiseux expansions. In particular, at isolated dominant singularities the standard $\Delta$--analytic
setup holds, and the simple square--root local form from Proposition~\ref{prop:xi-star-system} is the
generic case.

For the operator estimates below, we require the \(\Delta\)-domains and the
local Puiseux expansions at the moving dominant singularity to be controlled
uniformly in \(\varepsilon\). In the present paper we isolate this requirement
as the standing hypothesis of Definition~\ref{def:uniform-delta}. In concrete
finite-dimensional leaves it may be verified by a separate continuation
analysis of the Taylor sheet near the dominant orbit, but we do not prove a
general continuation theorem here.

\begin{proposition}[Coefficient transfer from the dominant orbit]
\label{prop:transfer-dominant}
Assume Definition~\ref{def:dominant-orbit} and that the Taylor sheet has a dominant $s$--orbit of
simple square--root branch points as in Proposition~\ref{prop:xi-star-system}, with $\Delta$--analytic
continuation near the orbit.
Then for each $p\ge1$, writing
\[
	U(x;\zeta)^p=\sum_{m\ge0}R_p(m;\zeta)\,x^{ms},
\]
one has the transfer asymptotics
\[
	R_p(m;\zeta)=A_p(\zeta)\,x_*(\zeta)^{-ms}\,m^{-3/2}\bigl(1+O(m^{-1})\bigr),
	\qquad m\to\infty,
\]
with an amplitude bound
\[
|A_p(\zeta)|\le C\,p\,|\lambda|^{p-1}
\]
(locally in $\zeta$), where $\lambda$ is the corresponding branch value at $x_*$ on the chosen
Taylor sheet.
\end{proposition}

\begin{proof}
	This is Lemma~\ref{lem:coeff-asymp-dominant}; see Appendix~\ref{app:transfer-dominant}.
\end{proof}

\begin{corollary}[Critical limit of the transfer coefficients]
	\label{cor:transfer-critical-data}
	Let $\zeta(\varepsilon)\to\zeta_c\in\mathcal C$ be a $C^1$ path satisfying the scalar
	transversality condition of Definition~\ref{def:transversal-path}, through the dominant
	subcritical regime toward a simple critical point. Let $\kappa(\varepsilon)$ denote the square-root
	coefficient in the Taylor-sheet Puiseux expansion at the representative dominant point
	$x_*(\varepsilon)$, for a fixed choice of square-root branch. Then the branch-point parameters and
	transfer amplitudes admit limits
	\[
		x_*(\varepsilon)\to x_{*,c},
		\qquad
		\lambda(\varepsilon)\to\lambda_c,
		\qquad
		A_p(\varepsilon)\to A_p^{(c)},
	\]
	with
	\[
		A_p^{(c)}
		=
		-\frac{s^{-1/2}}{2\sqrt{\pi}}\;p\,\kappa_c\,\lambda_c^{\,p-1}.
	\]
	Moreover, for $\varepsilon$ small,
	\[
		|A_p(\varepsilon)|\le C\,p\,|\lambda(\varepsilon)|^{p-1}.
	\]
\end{corollary}

\begin{proof}
	This is Lemma~\ref{lem:critical-coeff-asymp}; see Appendix~\ref{app:transfer-critical}.
\end{proof}

	The factor $s^{-1/2}$ in the formula for $A_p^{(c)}$ comes from the change of
	variable $z=x^s$. Indeed, if $z_*(\varepsilon)=x_*(\varepsilon)^s$, then
	\[
		1-\frac{z}{z_*(\varepsilon)}
		=
		1-\Bigl(\frac{x}{x_*(\varepsilon)}\Bigr)^s
		=
		s\Bigl(1-\frac{x}{x_*(\varepsilon)}\Bigr)
		+
		O\!\left(\Bigl(1-\frac{x}{x_*(\varepsilon)}\Bigr)^2\right),
	\]
	so the coefficient of $\sqrt{1-z/z_*(\varepsilon)}$ is
	$s^{-1/2}\kappa(\varepsilon)$, up to the fixed choice of square-root branch.

\subsection{Role in the spectral problem}
\label{subsec:role-spectral}

Proposition~\ref{prop:xi-star-system} identifies the first obstruction to analytic continuation of the
Taylor sheet: at a simple characteristic point the continuation develops a nondegenerate fold, hence a
square--root branch point on the analytic boundary $|x|=\rho_*(\zeta)$.  Under the dominant $s$--orbit
hypothesis (Definition~\ref{def:dominant-orbit}), this singularity governs the large--mode behavior of
the coefficients $R_p(m;\zeta)$ in \eqref{eq:Sp-expansion} via standard transfer theorems
(summarized in Proposition~\ref{prop:transfer-dominant}). Along a critical approach, the branch-point parameters and transfer amplitudes admit the limits presented in Corollary~\ref{cor:transfer-critical-data}.

These coefficient asymptotics underlie the spectral analysis: in each symmetry block
$H^{(q)}$, they control the leading part of the analytic Gram vectors (Proposition~\ref{prop:analytic-gram})
and hence determine which directions in $\ell^2$ can become large as $\rho_*(\zeta)\downarrow1$.
The quantitative implementation of this mechanism is carried out in
Section~\ref{sec:universal-instability}.

% =======================

\section{Conditional local universality of logarithmic instability}
\label{sec:universal-instability}

% =======================

This section establishes the operator-level rank-one extraction and its
spectral consequence, namely Theorem~\ref{thm:universality}. Section~\ref{sec:alg-criticality} provides the analytic framework for the spectral problem: along a transversal subcritical
approach to a simple critical point, the inverse map
\(U(\,\cdot\,;\zeta(\varepsilon))\) has a unique dominant \(s\)-orbit of
simple square-root branch points, and the corresponding coefficient
asymptotics are uniform. The task here is to convert that local branch-point parameters
into a precise operator statement for the renormalized symmetry blocks of the mixed Hessian.

The proof of the finite-\(N\) result has three steps. First, we introduce the
logarithmic scale \(L(\varepsilon)\) associated with the moving dominant
orbit. Second, we perform a fixed diagonal renormalization that places the
block Gram operators in a fixed weighted Hilbert-space framework up to criticality.
Third, we insert the uniform transfer asymptotics into the Gram
representation and extract a single rank-one logarithmic term, with the
remainder bounded in operator norm. The section concludes with an axiomatic
\(N=\infty\) extension.

\begin{assumption}[Standing hypotheses for the critical approach]
\label{ass:standing-hypotheses}
	Throughout Section~\ref{sec:universal-instability} we work along a \(C^1\) path
	\(\varepsilon\mapsto\zeta(\varepsilon)\), \(\varepsilon\downarrow0\), such that:
	\begin{enumerate}[label=\textup{(H\arabic*)},leftmargin=2.4em]
		\item \emph{Subcritical simple limit.}
		\(\zeta(\varepsilon)\to\zeta_c\in\mathcal C\), where \(\zeta_c\) is a
		simple critical point, and \(\zeta(\varepsilon)\) is subcritical for all
		sufficiently small \(\varepsilon>0\).

		\item \emph{Scalar transversality / parametrization.}
		The path satisfies the nondegeneracy condition of
		Definition~\ref{def:transversal-path}. Equivalently, after
		reparametrization we may and do assume
		\[
			\rho_*(\zeta(\varepsilon))=1+\varepsilon,
			\qquad 0<\varepsilon<\varepsilon_0.
		\]

		\item \emph{Dominant regime.}
		For all sufficiently small \(\varepsilon>0\), the Taylor branch lies in the
		dominant regime of Definition~\ref{def:dominant-orbit}, with a unique
		dominant \(s\)-orbit of simple branch points.

		\item \emph{Uniform square-root continuation.}
		The path satisfies Definition~\ref{def:uniform-delta}.
	\end{enumerate}
\end{assumption}

Items \textup{(H1)}--\textup{(H2)} describe the local algebraic setup established in
Section~\ref{sec:alg-criticality}: a simple critical point together with a
transversal subcritical parametrization. By contrast, \textup{(H3)}--\textup{(H4)}
contain the additional analytic assumptions. These require one to single out a
unique dominant orbit, keep it separated from the other singularities, and
control the continuation of the Taylor branch uniformly near the dominant orbit.
Under these hypotheses, the operator-theoretic conclusion proved below is a
local one. What remains open here is a global continuation result guaranteeing these assumptions on an arbitrary polynomial leaf.

Since the main theorem is conditional, we present the minimal guide needed to
apply it on a concrete polynomial leaf: (i) Locate a simple critical point by
solving the characteristic system from Section~\ref{sec:alg-criticality} and
checking the nondegeneracy conditions in Definition~\ref{def:simple-critical}; (ii) Choose a transversal subcritical approach by
verifying the scalar transversality condition of Definition~\ref{def:transversal-path},
so that \(\rho_*(\zeta(\varepsilon))\downarrow1\) monotonically along the chosen path; (iii) Establish dominant-orbit separation by checking that the Taylor branch has a unique dominant \(s\)-orbit of simple branch points and that no competing singularity of the same modulus enters the continuation region; (iv) Verify the uniform continuation properties stated in
Definition~\ref{def:uniform-delta}. The first part of this verification is local and algebraic, while the second is a genuinely uniform analytic criterion. Only after these checks can one pass to the operator statement and extract the logarithmic instability.

\subsection{Standing critical geometry and the logarithmic scale}
\label{subsec:standing-eta}

Fix $s\ge1$ and let $\varepsilon\mapsto\zeta(\varepsilon)$,
$\varepsilon>0$, be a path of subcritical parameters approaching a
simple critical point $\zeta_c\in\mathcal C$ as
$\varepsilon\downarrow0$, in the sense of
Section~\ref{sec:alg-criticality} and
Corollary~\ref{cor:transfer-critical-data}.
In particular, there is a unique dominant orbit
$\mathcal O_*(\varepsilon)$
(Definition~\ref{def:dominant-orbit}) with representative point
$x_*(\varepsilon)$ on the Taylor sheet and modulus
$\rho_*(\varepsilon)=|x_*(\varepsilon)|>1$
such that $\rho_*(\varepsilon)\downarrow1$ as $\varepsilon\downarrow0$.
For any $\rho\in(1,\rho_*(\varepsilon))$ set
\begin{equation}\label{eq:eta-zeta-rho}
	\eta(\zeta(\varepsilon);\rho)
	:=\bigl(\rho\,\rho_*(\varepsilon)\bigr)^{-2s}\in(0,1).
\end{equation}
We adopt the midpoint cutoff
\begin{equation}\label{eq:rho-choice-midpoint}
	\rho(\varepsilon):=\frac{1+\rho_*(\varepsilon)}{2},
\end{equation}
and abbreviate
\begin{equation}\label{eq:eta-eps-abbrev}
	\eta(\varepsilon):=\eta(\zeta(\varepsilon);\rho(\varepsilon)),
	\qquad
	L(\varepsilon):=\sum_{m\ge0}\frac{\eta(\varepsilon)^m}{m+1}.
\end{equation}

\begin{lemma}[Asymptotics of $L$ under the midpoint cutoff]
	\label{lem:Lambda-asymp}
	As $\varepsilon\downarrow0$,
	\begin{equation}\label{eq:Lambda-asymp-statement}
		1-\eta(\varepsilon)
		=3s\bigl(\rho_*(\varepsilon)-1\bigr)
		+O\bigl((\rho_*(\varepsilon)-1)^2\bigr),
		\qquad
		L(\varepsilon)
		=\log\!\frac{1}{\rho_*(\varepsilon)-1}+O(1).
	\end{equation}
	In particular, if the path is parametrized so that
	$\rho_*(\varepsilon)-1=\varepsilon+o(\varepsilon)$
	(as in Section \ref{sec:alg-criticality}), then
	$L(\varepsilon)=\log(1/\varepsilon)+O(1)$.
\end{lemma}

\begin{proof}
	Write $\delta:=\rho_*(\varepsilon)-1>0$. Since
	$\rho(\varepsilon)=\tfrac12(1+\rho_*(\varepsilon))$,
	\[
		\rho(\varepsilon)\rho_*(\varepsilon)
		=\frac{(1+\rho_*(\varepsilon))\rho_*(\varepsilon)}{2}
		=1+\tfrac{3}{2}\delta+O(\delta^2).
	\]
	Hence
	$\eta(\varepsilon)=(\rho(\varepsilon)\rho_*(\varepsilon))^{-2s}
		=1-3s\,\delta+O(\delta^2)$.
	The identity
	$L(\varepsilon)=-\eta(\varepsilon)^{-1}\log(1-\eta(\varepsilon))$
	then gives
	$L(\varepsilon)
		=\log\bigl(1/(1-\eta(\varepsilon))\bigr)+O(1)
		=\log(1/\delta)+O(1)$.
\end{proof}

The appendix also uses a closely related logarithmic scale defined directly
from the dominant modulus. It differs from \(L(\varepsilon)\) by only a bounded
term as \(\varepsilon\downarrow0\).

\subsection{Block factorization and weighted renormalization}\label{subsec:weighted-space}

Fix $\beta>0$.  Along the critical approach we also need a uniform bound on the Taylor branch on the midpoint circle.
With $\rho(\varepsilon)$ from \eqref{eq:rho-choice-midpoint} set
\[
	M(\varepsilon):=\sup_{|x|=\rho(\varepsilon)}|U(x;\zeta(\varepsilon))|.
\]
\begin{lemma}[Cauchy bounds on the midpoint circle]
	\label{lem:cauchy-Rp}
	Let $R_p(m;\varepsilon)$ be defined as in \eqref{eq:Sp-expansion} and let $M(\varepsilon)$ be the
	supremum of $|U(\cdot;\zeta(\varepsilon))|$ on the midpoint circle $|x|=\rho(\varepsilon)$.
	Then for all $p\ge1$ and $m\ge0$,
	\[
		|R_p(m;\varepsilon)|\le M(\varepsilon)^p\,\rho(\varepsilon)^{-ms}.
	\]
	In particular, there exist $\varepsilon_0>0$ and a constant $M_0<\infty$ such that
	$M(\varepsilon)\le M_0$ for all $0<\varepsilon<\varepsilon_0$.
\end{lemma}
\begin{proof}
	See Appendix~\ref{app:gram-bounds}.
\end{proof}

Since $\lambda(\varepsilon)\to\lambda_c$ (Corollary~\ref{cor:transfer-critical-data}), there exists
$\varepsilon_0>0$ and constants $M_0,L_0<\infty$ such that
$M(\varepsilon)\le M_0$ and $|\lambda(\varepsilon)|\le L_0$ for $0<\varepsilon<\varepsilon_0$.
Fix $\alpha>\max\{M_0,L_0\}$.
The inequality $\alpha>M_0$ is used for the midpoint-circle contribution and the
finite initial segment of the Gram kernel, while $\alpha>L_0$ is used for the
critical spike envelope coming from the amplitudes $A_p(\varepsilon)$. Thus a
single choice $\alpha>\max\{M_0,L_0\}$ controls both mechanisms uniformly.

For each congruence class $q\in\{1,\dots,s\}$, set $p_j:=q+js$ and
define the diagonal weight operator
\begin{equation}\label{eq:weights}
	w_j:=p_j^{\frac32+\beta}\,\alpha^{p_j},
	\qquad
	\mathcal W:=\diag(w_0,w_1,\dots).
\end{equation}

Define the weighted column Hilbert space
\begin{equation}\label{eq:weighted-Hilbert}
	\mathcal H_\beta
	:= \Bigl\{x=(x_j)_{j\ge0}:\
	\|x\|_{\mathcal H_\beta}^2
	:=\sum_{j\ge0}|x_j|^2 w_j^2<\infty\Bigr\},
	\qquad
	\langle x,y\rangle_{\mathcal H_\beta}
	:=\sum_{j\ge0}\overline{x_j}\,y_j\,w_j^2.
\end{equation}
Thus $\mathcal H_\beta$ is equipped with the standard complex Hilbert-space
convention, linear in the second slot. Then
$\mathcal W:\mathcal H_\beta\to\ell^2(\N_0)$, $(\mathcal W x)_j=w_jx_j$, is unitary, and we identify
$\mathcal H_\beta$ with $\ell^2(\N_0)$ via $\mathcal W$.

The polynomial factor $p_j^{3/2}$ compensates the $m^{-3/2}$ decay
from the transfer asymptotics (Lemma~\ref{lem:uniform-transfer}), and
the exponential factor $\alpha^{p_j}$ absorbs the growth
$|\lambda(\varepsilon)|^{p_j}$ in the amplitudes
$A_{p_j}(\varepsilon)$. Together they ensure that the conjugated Gram
vectors lie in $\ell^2(\N_0)$ uniformly in $\varepsilon$.

The admissible choice of weights is therefore part of the operator framework.  For two fixed admissible choices of \((\alpha,\beta)\), the quantitative coefficient \(\Gamma^{(q)}\) and the concrete realization of the renormalized block may change.  What remains invariant in the argument below is the qualitative conclusion that the extracted singular part has rank one and grows on the logarithmic scale, whereas the higher variational levels stay bounded.

\begin{definition}[Weighted renormalization]\label{def:conjugated-gram}
	Fix $q\in\{1,\dots,s\}$ and write $p_j=q+js$. For finitely supported sequences
	$c=(c_k)_{k\ge0}$ define the synthesis map $V^{(q)}(\zeta)c$ by
	\[
		\bigl(V^{(q)}(\zeta)c\bigr)_j
		:=
		\sum_{k\ge0} v^{(p_k)}_j(\zeta)\,c_k,
		\qquad j\ge0,
	\]
	where $v^{(p_k)}(\zeta)$ are the Gram vectors from \eqref{eq:Gram-alg-vectors}. Since
	$v^{(p_k)}_j(\zeta)=0$ for $j<k$, the sum is finite for each fixed $j$. With this convention, the identities
	\begin{equation}\label{eq:Hq-VV}
		G^{(q)}(\zeta)=V^{(q)}(\zeta)^*V^{(q)}(\zeta),
		\qquad
		H^{(q)}(\zeta)=V^{(q)}(\zeta)V^{(q)}(\zeta)^*
	\end{equation}
	are understood first on finitely supported sequences, equivalently at the level of quadratic forms.
	Let $\mathcal W=\diag(w_0,w_1,\dots)$ be the diagonal weight operator from \eqref{eq:weights}.
	Define the renormalized synthesis operator
	\[
		\widetilde V^{(q)}(\zeta):=V^{(q)}(\zeta)\,\mathcal W^{-1},
	\]
	initially on finitely supported sequences, and the corresponding renormalized Gram and Hessian forms
	\begin{equation}\label{eq:GHtilde}
		\widetilde G^{(q)}(\zeta):=\widetilde V^{(q)}(\zeta)^*\,\widetilde V^{(q)}(\zeta),
		\qquad
		\widetilde H^{(q)}(\zeta):=\widetilde V^{(q)}(\zeta)\,\widetilde V^{(q)}(\zeta)^*.
	\end{equation}
	After Proposition~\ref{prop:subcrit-compact}, these forms extend to bounded positive operators on
	$\ell^2(\N_0)$ in the subcritical regime.
\end{definition}

\begin{remark}[Hessian/Gram spectral equivalence]
	\label{rem:HG-spectrum}
	From~\eqref{eq:GHtilde} it follows that \(\widetilde H^{(q)}(\varepsilon)\) and
	\(\widetilde G^{(q)}(\varepsilon)\) share the same nonzero spectrum with
	multiplicities. In particular, in compact regimes
	(e.g.\ Proposition~\ref{prop:subcrit-compact}), their nonzero eigenvalues
	coincide in nonincreasing order, and the min--max levels
	\(\mu_k^{(q)}(\varepsilon)\) may be computed from either block.
\end{remark}

	For vectors \(u,v\) in a complex Hilbert space we write
	\[
		(u\otimes v)\xi:=\langle v,\xi\rangle\,u.
	\]
	In particular, \(u\otimes u\) is a positive rank--one operator.

Now we discuss subcritical compactness of the renormalized operators, which is a consequence of the choice of weights and the uniform bound on the Taylor branch on the midpoint circle.

\begin{proposition}[Subcritical compactness]\label{prop:subcrit-compact}
	Fix $q\in\{1,\dots,s\}$ and $\beta>0$, and let $\alpha$ be chosen as above.
	For $0<\varepsilon<\varepsilon_0$ the renormalized synthesis map
	$\widetilde V^{(q)}(\varepsilon)$ extends to a Hilbert--Schmidt operator on $\ell^2(\N_0)$.
	In particular, the quadratic forms \eqref{eq:GHtilde} extend uniquely to positive trace-class operators
	$\widetilde G^{(q)}(\varepsilon)$ and $\widetilde H^{(q)}(\varepsilon)$ on $\ell^2(\N_0)$. Hence they are compact and bounded.
\end{proposition}

\begin{proof}
	Work on the midpoint circle $|x|=\rho(\varepsilon)$ from \eqref{eq:rho-choice-midpoint}.
	By definition of $M(\varepsilon)$ and Cauchy's estimate (Lemma~\ref{lem:cauchy-Rp}),
	for all $p\ge1$ and $m\ge0$,
	\[
		|R_p(m;\varepsilon)|
		\le M(\varepsilon)^p\,\rho(\varepsilon)^{-sm}
		\le M_0^p\,\rho(\varepsilon)^{-sm}.
	\]
	Using the explicit Gram vectors \eqref{eq:Gram-alg-vectors} and the bound $M_0<\alpha$,
	the squared $\ell^2$--norm of the $j$-th renormalized column of $\widetilde V^{(q)}$ satisfies
	\[
		\bigl\|\widetilde v^{(p_j)}(\varepsilon)\bigr\|_{\ell^2}^2
		\lesssim
		\frac{M_0^{2p_j}}{w_j^2}
		\sum_{m\ge0}\frac{(p_j+sm)^2}{p_j}\,\rho(\varepsilon)^{-2sm}.
	\]
	For fixed $\varepsilon>0$ the geometric sum satisfies
	\[
		\sum_{m\ge0}\frac{(p_j+sm)^2}{p_j}\,\rho(\varepsilon)^{-2sm}
		\le C_\varepsilon(1+p_j^2),
	\]
	with $C_\varepsilon<\infty$. Therefore
	\[
		\bigl\|\widetilde v^{(p_j)}(\varepsilon)\bigr\|_{\ell^2}^2
		\lesssim
	C_\varepsilon\,p_j^{-1-2\beta}\,\Bigl(\frac{M_0}{\alpha}\Bigr)^{2p_j}.
	\]
	The right-hand side is summable in $j$ since $M_0/\alpha<1$ and $p_j\to\infty$.
	Hence $\widetilde V^{(q)}(\varepsilon)$ is Hilbert--Schmidt, and the trace-class/compactness conclusions follow.
\end{proof}

\begin{remark}[Min--max levels]\label{rem:rayleigh}
	We measure the spectrum of bounded positive self-adjoint blocks by the
	Courant--Fischer min--max levels $\mu_k(A)$ from
	Definition~\ref{def:variational-eigs}. When the block is compact (for
	example in the subcritical regime of
	Proposition~\ref{prop:subcrit-compact}), these min--max levels coincide
	with the usual eigenvalues in nonincreasing order. For the renormalized
	Hessian and Gram blocks the same nonzero levels are obtained whenever both
	operators are defined; see Remark~\ref{rem:HG-spectrum}.
\end{remark}

\begin{definition}[Variational eigenvalues]\label{def:variational-eigs}
	Let $A$ be a bounded self-adjoint operator on a Hilbert space $\mathcal H$.
	For $k\ge1$ define the Courant--Fischer levels
	\[
		\mu_k(A)
		:=\inf_{\substack{L\subset \mathcal H\\ \dim L=k-1}}\
		\sup_{\substack{u\perp L\\ \|u\|=1}}
		\langle u,Au\rangle.
	\]
	When $A$ is compact and positive, $(\mu_k(A))_{k\ge1}$ coincides with
	the usual eigenvalues of $A$ in nonincreasing order, counted with
	multiplicity. In general, $(\mu_k(A))_{k\ge1}$ presents the ordered
	min--max levels of the quadratic form.
	We apply this to $A=\widetilde G^{(q)}(\varepsilon)$ and write
	$\mu_k^{(q)}(\varepsilon):=\mu_k(\widetilde G^{(q)}(\varepsilon))$.
	For each fixed subcritical $\varepsilon>0$, the same sequence is obtained
	from $\widetilde H^{(q)}(\varepsilon)$, since in that regime the blocks are
	compact and $\widetilde H^{(q)}(\varepsilon)$ and
	$\widetilde G^{(q)}(\varepsilon)$ have the same nonzero eigenvalues; see
	Proposition~\ref{prop:subcrit-compact} and Remark~\ref{rem:HG-spectrum}.
\end{definition}

We next consider the tail Gram structure that will be exploited in the
rank--one extraction.
In view of Definition~\ref{def:conjugated-gram}, the renormalized Hessian
$\widetilde H^{(q)}$ and the renormalized Gram operator $\widetilde G^{(q)}$
share the same nonzero spectrum.  We therefore analyze $\widetilde G^{(q)}$
via an explicit entry formula for $G^{(q)}$ in which the Gram index is the
tail parameter~$m$.

\begin{proposition}[Gram operator entry formula]
	\label{prop:analytic-gram}
	Assume $\rho_*(\zeta)>1$. For $j_1\le j_2$ set $\Delta:=j_2-j_1$. Then
	\begin{equation}\label{eq:gram-entry}
		G^{(q)}_{j_1j_2}(\zeta)
		=
		\sum_{m\ge0}
		\frac{(p_{j_2}+sm)^2}{\sqrt{p_{j_1}p_{j_2}}}\,
		\overline{R_{p_{j_1}}(m+\Delta;\zeta)}\,R_{p_{j_2}}(m;\zeta),
	\end{equation}
	and $G^{(q)}_{j_1j_2}(\zeta)
		=\overline{G^{(q)}_{j_2j_1}(\zeta)}$ for $j_1>j_2$.
\end{proposition}

\begin{proof}
	Fix any $\rho$ with $1<\rho<\rho_*(\zeta)$. Since $U(\cdot;\zeta)$ is analytic on $|x|<\rho_*(\zeta)$,
	Cauchy's estimate gives $|R_p(m;\zeta)|\le C_{p,\rho}\,\rho^{-sm}$ for each fixed $p$.
	Hence the series in \eqref{eq:gram-entry} converges absolutely. This justifies
	the reindexing below and the termwise computation of the inner product.
	By definition,
	$G^{(q)}_{j_1j_2}(\zeta)
		=\bigl\langle
		\bm v^{(p_{j_1})}(\zeta),\,
		\bm v^{(p_{j_2})}(\zeta)
		\bigr\rangle_{\ell^2}$,
	where $\bm v^{(p)}(\zeta)$ is the Gram vector from
	Proposition~\ref{prop:gram-scale}.
	Since $v^{(p)}_i=0$ unless $p_i\ge p$, the supports of
	$\bm v^{(p_{j_1})}$ and $\bm v^{(p_{j_2})}$ overlap at indices
	$i\ge j_2$ (assuming $j_1\le j_2$).
	Writing $i=j_2+m$ with $m\ge0$ and $\Delta:=j_2-j_1$, we have
	\[
		v^{(p_{j_1})}_i
		=\frac{p_i}{\sqrt{p_{j_1}}}\,R_{p_{j_1}}(m+\Delta;\zeta),
		\qquad
		v^{(p_{j_2})}_i
		=\frac{p_i}{\sqrt{p_{j_2}}}\,R_{p_{j_2}}(m;\zeta).
	\]
	Using $p_i=p_{j_2}+sm$ and summing over $m\ge0$ yields
	\eqref{eq:gram-entry}.
\end{proof}

\subsection{Uniform transfer on the dominant orbit}\label{subsec:standing-transfer}

\begin{definition}[Uniform continuation hypothesis near the dominant orbit]\label{def:uniform-delta}
	Let $\zeta(\varepsilon)$ be a $C^1$ path with $\zeta(0)=\zeta_c\in\mathcal C$ a simple critical point
	(Definition~\ref{def:simple-critical}). Assume that for $0\le\varepsilon\le\varepsilon_1$ the Taylor branch
	$U(\,\cdot\,;\zeta(\varepsilon))$ lies in the dominant regime of Definition~\ref{def:dominant-orbit} with a unique
	dominant $s$--orbit of simple branch points.
	Write $s=\gcd(s_1,\dots,s_N)$, set $z:=x^{s}$, and, for notational convenience, denote by
	$\widehat U(z;\varepsilon)$ the same Taylor sheet expressed in the variable $z$.
	Let $x_*(\varepsilon)$ be a representative dominant singularity and set $z_*(\varepsilon):=x_*(\varepsilon)^s$.
	Let $\lambda(\varepsilon)$ denote the corresponding branch value of the chosen Taylor sheet at $z_*(\varepsilon)$, that is,
	\[
		\lambda(\varepsilon):=
		\lim_{z\to z_*(\varepsilon)}\widehat U(z;\varepsilon),
	\]
	where $z$ is on the chosen Taylor sheet.
	We say that the path satisfies a \emph{uniform square-root continuation
condition} if there exist constants \(\varepsilon_0>0\),
\(\phi\in(0,\pi/2)\), \(\eta_0>0\), and \(B_0,C_0,c_0>0\) such that, for every
\(0<\varepsilon<\varepsilon_0\), the following properties hold:
	\begin{enumerate}[label=(U\arabic*),leftmargin=2.5em]
		    \item\label{U1} \emph{Uniform indented-domain continuation on an indented disk.}
		      Set
		      \[
			      R(\varepsilon):=|z_*(\varepsilon)|+\eta_0/2,
			      \qquad
			      \Omega_\varepsilon:=\{\,|z|<R(\varepsilon)\,\}\cap \Delta(z_*(\varepsilon),\phi),
		      \]
			    where \(\Delta(z_*(\varepsilon),\phi)\) is the indented domain from
		      Definition~\ref{def:delta-domain}. The Taylor sheet
		      \(\widehat U(\cdot;\varepsilon)\) admits analytic continuation to
		      \(\Omega_\varepsilon\). Let \(\Gamma_\varepsilon:=\partial\Omega_\varepsilon\),
		      oriented positively. Then \(\Gamma_\varepsilon\) is a simple closed contour
		      winding once around \(0\), the closed disk
		      \[
			      \overline{D\bigl(0,\rho(\varepsilon)^s\bigr)}
			      \subset \Omega_\varepsilon,
			      \qquad
			      \rho(\varepsilon):=\frac{1+\rho_*(\varepsilon)}{2},
		      \]
			  and \(z_*(\varepsilon)\notin \Gamma_\varepsilon\).
		      Moreover,
		      \[
			      M:=\sup_{0<\varepsilon<\varepsilon_0}\ \sup_{z\in\Gamma_\varepsilon}\ |\widehat U(z;\varepsilon)|<\infty.
		      \]
		\item\label{U2} \emph{Uniform Puiseux expansion to order $3/2$.}
		        For all $z\in\Delta(z_*(\varepsilon),\phi)$ with $|\delta_z|<\eta_0$, where
		      \[
			     \delta_z:=1-\frac{z}{z_*(\varepsilon)},
		      \]
		      one has a Puiseux expansion on the Taylor sheet of the form
		      \begin{equation}\label{eq:puiseux-uniform}
			      \widehat U(z;\varepsilon)
			      =
			      \lambda(\varepsilon)
			      +
				     s^{-1/2}\kappa(\varepsilon)\,\delta_z^{1/2}
			      +
				     c(\varepsilon)\,\delta_z
			      +
				    d(\varepsilon)\,\delta_z^{3/2}
			      +
			      E(z;\varepsilon),
		      \end{equation}
			    where $\delta_z^{1/2}$ is taken with $\arg\delta_z\in(-\pi,\pi)$, and the coefficients and remainder satisfy
		      \[
			      |\lambda(\varepsilon)|\ge c_0,\qquad |\kappa(\varepsilon)|\ge c_0,\qquad
			      |\lambda(\varepsilon)|+|\kappa(\varepsilon)|+|c(\varepsilon)|+|d(\varepsilon)|\le C_0,
		      \]
		      and $|E(z;\varepsilon)|\le C_0\,|\delta_z|^{2}$.
			  In addition, for each $p\ge1$, let $\mathcal R_p(\cdot;\varepsilon)$ denote the regular part in
		      the local expansion
		      \[
			      \widehat U(z;\varepsilon)^p
			      =
			      \lambda(\varepsilon)^p
			      +
			      b_p(\varepsilon)\,\delta_z^{1/2}
			      +
			      d_p(\varepsilon)\,\delta_z^{3/2}
			      +
			      \mathcal R_p(z;\varepsilon)
		      \]
		      on $\Omega_\varepsilon\cap\{|\delta_z|<\eta_0\}$, where $b_p(\varepsilon)$ and
		      $d_p(\varepsilon)$ are the corresponding Puiseux coefficients. We assume that each
		      $\mathcal R_p(\cdot;\varepsilon)$ extends holomorphically to the full disk
		      \[
			      |z|<R(\varepsilon)=|z_*(\varepsilon)|+\eta_0/2,
		      \]
		      and satisfies the uniform boundary bound
		      \[
			      \sup_{0<\varepsilon<\varepsilon_0}\ \sup_{|z|=R(\varepsilon)}
			      |\mathcal R_p(z;\varepsilon)|
			      \le B_0^{\,p}
			      \qquad (p\ge1).
		      \]
		      Here $\kappa(\varepsilon)$ is the square-root coefficient in the
		      original $x$-variable, so that
		      \[
			      U(x;\zeta(\varepsilon))
			      =
			      \lambda(\varepsilon)
			      +
			      \kappa(\varepsilon)\sqrt{1-x/x_*(\varepsilon)}
			      +
			      O\!\left(1-x/x_*(\varepsilon)\right).
		      \]
	\end{enumerate}
\end{definition}

The uniform continuation hypothesis of Definition~\ref{def:uniform-delta} plays
a central role in the proof of Theorem~\ref{thm:universality}.
Clause \textup{(U1)} gives the global continuation region and the contour bound needed for
Cauchy estimates on the midpoint circle. The first part of \textup{(U2)} gives the local
Puiseux expansion, with coefficients controlled uniformly in \(\varepsilon\). The second part of
\textup{(U2)}, namely the holomorphic continuation and boundary estimate for the regular part
\(\mathcal R_p\), is stronger than standard \(\Delta\)-analyticity near a simple square-root
branch point. It is not a consequence of the local branch-point analysis in
Proposition~\ref{prop:xi-star-system} or Lemma~\ref{lem:local-branch-data}. Rather, it is an
additional uniform assumption that must be verified on a concrete family. In the present paper it is
used only in Lemma~\ref{lem:uniform-transfer} to control the cutoff error uniformly in both \(p\) and
\(\varepsilon\).

\begin{lemma}[Uniform transfer on the moving dominant orbit]
	\label{lem:uniform-transfer}
	Assume that along the transversal path $\zeta(\varepsilon)$ the dominant singularities of the Taylor branch
	form a unique dominant $s$--orbit of simple square--root branch points as in Definition~\ref{def:dominant-orbit},
	and assume the uniform square--root continuation condition of Definition~\ref{def:uniform-delta}.
	Then there exist constants $M_{\mathrm{tr}},C>0$ and transfer amplitudes $A_p(\varepsilon)$ ($p\ge1$) such that,
	for all sufficiently small $\varepsilon$, all $p\ge1$, and all $m\ge M_{\mathrm{tr}}(1+p^2)$,
	\begin{equation}
		\label{eq:uniform-transfer}
		R_p(m;\varepsilon)
		=
		A_p(\varepsilon)\, m^{-3/2}\,x_*(\varepsilon)^{-ms}
		\Bigl(1+O\bigl((1+p^2)/m\bigr)\Bigr),
	\end{equation}
	uniformly in $p$ and $m$.
	Moreover,
	\begin{equation}
		\label{eq:Ap-growth}
		|A_p(\varepsilon)|\le C\,p\,|\lambda(\varepsilon)|^{p-1},
	\end{equation}
	where $\lambda(\varepsilon)$ is the corresponding branch value at
$x_*(\varepsilon)$ on the chosen Taylor sheet.
\end{lemma}

\begin{proof}
	Deferred to Appendix~\ref{app:gram-bounds}.
\end{proof}

\subsection{Universality theorem}\label{subsec:universality-thm}

Fix \(\beta>0\), choose \(\alpha\) as in
Section~\ref{subsec:weighted-space}, and recall from \eqref{eq:eta-eps-abbrev} that
\[
	L(\varepsilon):=\sum_{m\ge0}\frac{\eta(\varepsilon)^m}{m+1},\qquad L(\varepsilon)=\log \frac{1}{\rho_*(\varepsilon)-1}+O(1)
\]
by Lemma~\ref{lem:Lambda-asymp}. For each \(q\in\{1,\dots,s\}\), let
\[
	\widetilde G^{(q)}(\varepsilon)
	:=\widetilde G^{(q)}(\zeta(\varepsilon))
\]
denote the renormalized Gram block from
Definition~\ref{def:conjugated-gram}.

\begin{proposition}[Rank--one extraction in a symmetry block]
	\label{prop:rank-one-extraction}
	Assume the standing hypotheses of
	Assumption~\ref{ass:standing-hypotheses}. Fix
	$q\in\{1,\dots,s\}$. Then there exist $\varepsilon_0>0$, vectors
	$\widetilde d^{(q)}(\varepsilon)\in\ell^2(\N_0)$, and bounded operators
	$\widetilde C^{(q)}(\varepsilon)$ such that for
	$0<\varepsilon<\varepsilon_0$,
	\begin{equation}\label{eq:Gq-tilde-decomposition}
		\widetilde G^{(q)}(\varepsilon)
		=
		L(\varepsilon)\,
		\widetilde d^{(q)}(\varepsilon)\otimes
		\widetilde d^{(q)}(\varepsilon)
		+
		\widetilde C^{(q)}(\varepsilon),
		\qquad
		\sup_{0<\varepsilon<\varepsilon_0}
		\|\widetilde C^{(q)}(\varepsilon)\|<\infty.
	\end{equation}
	Moreover,
	\[
		\widetilde d^{(q)}(\varepsilon)\to \widetilde d^{(q)}_c
		\quad\text{in }\ell^2(\N_0),
		\qquad
		\widetilde d^{(q)}_c\neq0.
	\]
\end{proposition}

The proof is deferred to Appendix~\ref{app:gram-bounds}.

At this point the roles of the assumptions become clear. Lemma~\ref{lem:local-branch-data} provides the local continuation of the characteristic parameters, whereas Assumption~\ref{ass:standing-hypotheses} collects the additional dominance and uniform
$\Delta$-control required to extract the logarithmic tail.

\begin{theorem}[Conditional local rank--one logarithmic instability]
	\label{thm:universality}
	Assume the standing hypotheses of
	Assumption~\ref{ass:standing-hypotheses}. Fix \(\beta>0\), choose
	\(\alpha>\max\{M_0,L_0\}\) as in
	Section~\ref{subsec:weighted-space}, and let
	\(\widetilde G^{(q)}(\varepsilon)\) be the renormalized Gram block of
	Definition~\ref{def:conjugated-gram}. Fix
	$q\in\{1,\dots,s\}$, and let
	\[
		\mu_k^{(q)}(\varepsilon):=\mu_k\bigl(\widetilde G^{(q)}(\varepsilon)\bigr),
		\qquad k\ge1.
	\]
	Let
	\[
		\Gamma^{(q)}:=\|\widetilde d_c^{(q)}\|^2>0,
	\]
	where \(\widetilde d_c^{(q)}\) is the limit vector from
	Proposition~\ref{prop:rank-one-extraction}. Then, as
	\(\varepsilon\downarrow0\),
	\begin{equation}\label{eq:mu-sing}
		\mu^{(q)}_1(\varepsilon)
		=
		\Gamma^{(q)}\,L(\varepsilon)+O(1),
		\qquad
		\frac{\mu^{(q)}_1(\varepsilon)}{L(\varepsilon)}
		\longrightarrow \Gamma^{(q)},
	\end{equation}
	and for every \(k\ge2\),
	\[
	\sup_{0<\varepsilon<\varepsilon_0}\mu_k^{(q)}(\varepsilon)<\infty.
	\]
	In particular, in each symmetry block exactly one variational eigenvalue
	diverges logarithmically, while all higher variational eigenvalues remain
	bounded.
\end{theorem}

	By Remark~\ref{rem:HG-spectrum}, for each fixed subcritical
	\(\varepsilon>0\) the same variational sequence may equivalently be
	obtained from the renormalized Hessian block
	\(\widetilde H^{(q)}(\varepsilon)\). The structural decomposition in
	Proposition~\ref{prop:rank-one-extraction} concerns the \emph{renormalized}
	Gram block \(\widetilde G^{(q)}\) obtained from
	Definition~\ref{def:conjugated-gram} using the weights
	\eqref{eq:weights}. In particular, the limiting spike profile
	\(\widetilde d_c^{(q)}\) and the coefficient \(\Gamma^{(q)}\) depend on
	the choice of \((\alpha,\beta)\). The theorem asserts that, for any such
	fixed weighted renormalization, the divergent part is rank one and occurs
	on the logarithmic scale \(L(\varepsilon)\). Changing the local choice of
	square-root branch multiplies \(\widetilde d^{(q)}(\varepsilon)\) by
	a unimodular phase and does not affect \(\Gamma^{(q)}\).

\begin{proof}[Proof of Theorem~\ref{thm:universality}]
	Fix \(q\in\{1,\dots,s\}\). By
	Proposition~\ref{prop:rank-one-extraction}, we have
	\[
		\widetilde G^{(q)}(\varepsilon)
		=
		L(\varepsilon)\,
		\widetilde d^{(q)}(\varepsilon)\otimes
		\widetilde d^{(q)}(\varepsilon)
		+
		\widetilde C^{(q)}(\varepsilon),
	\]
	with $\widetilde d^{(q)}(\varepsilon)\to \widetilde d_c^{(q)}$
		in $\ell^2(\N_0)$, $\sup_{0<\varepsilon<\varepsilon_0}\|\widetilde C^{(q)}(\varepsilon)\|<\infty$. Since \(\widetilde d_c^{(q)}\neq0\), one has
	\(\widetilde d^{(q)}(\varepsilon)\neq0\) for all sufficiently small
	\(\varepsilon>0\).

	Write $A(\varepsilon):=\widetilde G^{(q)}(\varepsilon)$, $d(\varepsilon):=\widetilde d^{(q)}(\varepsilon)$, and $C(\varepsilon):=\widetilde C^{(q)}(\varepsilon)$.
	By Definition~\ref{def:variational-eigs}, the first min--max level is
	\[
		\mu_1^{(q)}(\varepsilon)
		=
		\sup_{\|u\|=1}\langle u,A(\varepsilon)u\rangle .
	\]

	Evaluating the Rayleigh quotient at
	\(u=d(\varepsilon)/\|d(\varepsilon)\|\), we obtain
	\[
		\mu_1^{(q)}(\varepsilon)
		\ge
		\left\langle \frac{d(\varepsilon)}{\|d(\varepsilon)\|},
		A(\varepsilon)\frac{d(\varepsilon)}{\|d(\varepsilon)\|}\right\rangle
		=
		L(\varepsilon)\|d(\varepsilon)\|^2
		+
		\left\langle \frac{d(\varepsilon)}{\|d(\varepsilon)\|},
		C(\varepsilon)\frac{d(\varepsilon)}{\|d(\varepsilon)\|}\right\rangle
		\ge
		L(\varepsilon)\|d(\varepsilon)\|^2-\|C(\varepsilon)\|.
	\]
	On the other hand, for every unit vector \(u\),
	\[
		\langle u,A(\varepsilon)u\rangle
		=
		L(\varepsilon)|\langle u,d(\varepsilon)\rangle|^2
		+
		\langle u,C(\varepsilon)u\rangle
		\le
		L(\varepsilon)\|d(\varepsilon)\|^2+\|C(\varepsilon)\|,
	\]
	since \(|\langle u,d(\varepsilon)\rangle|^2\le \|d(\varepsilon)\|^2\).
	Taking the supremum over \(\|u\|=1\) yields
	\[
		\mu_1^{(q)}(\varepsilon)
		=
		L(\varepsilon)\|d(\varepsilon)\|^2+O(1).
	\]
	Because \(d(\varepsilon)\to \widetilde d_c^{(q)}\) in \(\ell^2(\N_0)\), we have
	\[
		\|d(\varepsilon)\|^2\to \|\widetilde d_c^{(q)}\|^2=\Gamma^{(q)},
	\]
	which proves \eqref{eq:mu-sing}.

Now let \(k\ge2\). By Definition~\ref{def:variational-eigs},
\[
	\mu_k^{(q)}(\varepsilon)
	=
	\inf_{\dim M=k-1}
	\sup_{\substack{u\perp M\\ \|u\|=1}}
	\langle u,A(\varepsilon)u\rangle .
\]
Choose any \((k-1)\)-dimensional subspace \(M_k(\varepsilon)\subset \ell^2(\N_0)\)
such that $\Span\{d(\varepsilon)\}\subset M_k(\varepsilon)$. If \(u\perp M_k(\varepsilon)\) and \(\|u\|=1\), then in particular
\(u\perp d(\varepsilon)\), so $\langle u,A(\varepsilon)u\rangle
	=
	\langle u,C(\varepsilon)u\rangle$. Hence
\[
	\sup_{\substack{u\perp M_k(\varepsilon)\\ \|u\|=1}}
	\langle u,A(\varepsilon)u\rangle
	\le \|C(\varepsilon)\|.
\]
Therefore, $\mu_k^{(q)}(\varepsilon)\le \|C(\varepsilon)\|$, for
	$k\ge2$. Since \(\sup_{0<\varepsilon<\varepsilon_0}\|C(\varepsilon)\|<\infty\), we conclude that
\[
	\sup_{0<\varepsilon<\varepsilon_0}\mu_k^{(q)}(\varepsilon)<\infty,
	\qquad k\ge2.
\]
\end{proof}

The boundedness statement in Proposition~\ref{prop:rank-one-extraction} is
sufficient for the variational eigenvalue bounds \eqref{eq:mu-sing}. The
following lemma states a stronger structural fact proved in the appendix.

\begin{lemma}[Compact remainder and norm limit]
	\label{lem:compact-remainder}
	Under the hypotheses of Theorem~\ref{thm:universality}, for each
	\(q\in\{1,\dots,s\}\) the remainder operators
	\(\widetilde C^{(q)}(\varepsilon)\) in
	\eqref{eq:Gq-tilde-decomposition} are Hilbert--Schmidt, hence compact, for
	all sufficiently small \(\varepsilon>0\). Moreover there exists a
	Hilbert--Schmidt operator \(\widetilde C^{(q)}_*\) such that
	\[
		\|\widetilde C^{(q)}(\varepsilon)-\widetilde C^{(q)}_*\|
		\xrightarrow[\varepsilon\downarrow0]{}0.
	\]
\end{lemma}

\begin{proof}
	Deferred to Appendix~\ref{app:compact-remainder}.
\end{proof}

\begin{figure}[t]
	\centering
	\includegraphics[width=\textwidth]{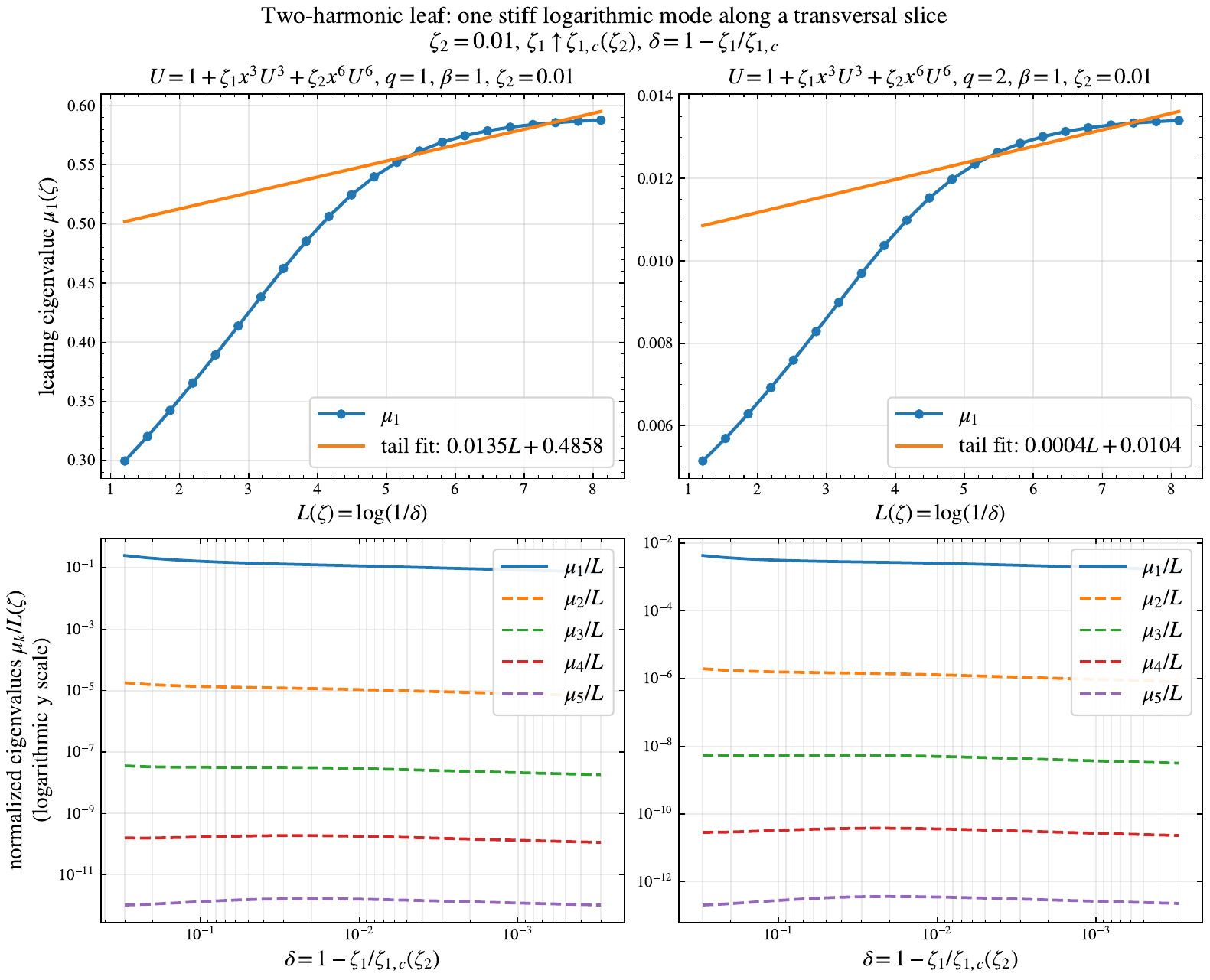}
	\caption{Numerical behavior consistent with Theorem~\ref{thm:universality} on the
	two--harmonic leaf
	\(
	U=1+\zeta_1 x^3U^3+\zeta_2 x^6U^6
	\).
	We fix \(\zeta_2=0.01\) and approach the critical locus along the
	transversal slice \(\zeta_1\uparrow \zeta_{1,c}(\zeta_2)\), with
	\(\delta:=1-\zeta_1/\zeta_{1,c}(\zeta_2)\).
	The spectra are computed from the $(J+1)\times(J+1)$ principal truncation of the renormalized Gram blocks \(\widetilde G^{(q)}\)
	with \(J=70\), \(\alpha=2\), and \(\beta=1\).
	The figure is heuristic: it illustrates the predicted scaling at this fixed truncation.
	Top row: the leading eigenvalue \(\mu_1^{(q)}\) for the representative blocks
	\(q=1,2\), plotted against \(\log(1/\delta)\), together with linear tail fits.
	Bottom row: the normalized eigenvalues \(\mu_k^{(q)}/\log(1/\delta)\),
	\(k=1,\dots,5\), on a logarithmic vertical scale.
	The plots show one logarithmically diverging stiff mode, while the normalized
	soft modes tend to zero.}
	\label{fig:twoharm-numerics}
\end{figure}

Figure~\ref{fig:twoharm-numerics} shows numerical behavior consistent with the same spectral mechanism on a genuinely
two--parameter leaf. We consider
\[
	U(x;\zeta_1,\zeta_2)=1+\zeta_1 x^3U(x;\zeta_1,\zeta_2)^3+\zeta_2 x^6U(x;\zeta_1,\zeta_2)^6,
\]
fix \(\zeta_2=0.01\), and approach the critical locus along the transversal
slice \(\zeta_1\uparrow \zeta_{1,c}(\zeta_2)\).  Writing
\[
	\delta:=1-\frac{\zeta_1}{\zeta_{1,c}(\zeta_2)},
\]
one has \(\varepsilon\asymp\delta\) along this slice, hence
\[
	L(\varepsilon)=\log(1/\delta)+O(1).
\]

The upper plots show the leading variational eigenvalue \(\mu_1^{(q)}\) for the
representative blocks \(q=1,2\), together with linear tail fits in \(\log(1/\delta)\),
which is the behavior predicted by Theorem~\ref{thm:universality} if the
hypotheses of Definition~\ref{def:uniform-delta} hold along this slice. Since
\(s=3\) on this leaf, there is a third symmetry block \(q=3\). We omit it from
the figure only to avoid overloading the comparison. The lower plots show the
normalized eigenvalues \(\mu_k^{(q)}/\log(1/\delta)\) on a logarithmic scale.
The leading ratio approaches a nonzero constant, while the soft modes tend to
zero, displaying the expected rank--one dominance.

\subsection{Extension to \texorpdfstring{\(N=\infty\)}{N=infinity}}
\label{subsec:infty-extension}

The proof of Theorem~\ref{thm:universality} uses only the coefficient
asymptotics along the dominant orbit and the boundedness of the weighted Gram
remainder. This leads to an abstract infinite-dimensional criterion, without providing an automatic extension to arbitrary infinite-dimensional leaves. The remainder of
this subsection therefore formulates an abstract criterion: it isolates the
operator-theoretic conditions that would have to be checked on a concrete
\(N=\infty\) family before the same rank-one mechanism could be concluded. To avoid
confusion with the finite-mode quantity \(L(\varepsilon)\) from
Section~\ref{subsec:standing-eta}, we write \(L_*(\varepsilon)\) here for the
logarithmic tail built directly from the true dominant modulus
\(\rho_*(\varepsilon)\).

\begin{assumption}[Abstract \(N=\infty\) hypotheses]\label{ass:infty-rank-one}
Fix a symmetry block \(q\in\{1,\dots,s\}\), and write \(p_j:=q+js\) and
\(w_j:=p_j^{3/2+\beta}\alpha^{p_j}\) with \(\beta>0\) and \(\alpha>1\).
Assume that there exists a subcritical parameter interval
\(0<\varepsilon<\varepsilon_0\) together with quantities
\[
\rho_*(\varepsilon)>1,\qquad
z_*(\varepsilon)=\rho_*(\varepsilon)^s e^{i\phi(\varepsilon)},\qquad
\eta_*(\varepsilon)=\rho_*(\varepsilon)^{-2s},\qquad
L_*(\varepsilon):=\sum_{m\ge0}\frac{\eta_*(\varepsilon)^m}{m+1},
\]
and coefficient arrays \(R_p(m;\varepsilon)\) (\(p\equiv q \!\!\!\pmod s\))
such that \(\rho_*(\varepsilon)\downarrow1\) as \(\varepsilon\downarrow0\), equivalently
\(\eta_*(\varepsilon)\uparrow1\), so in particular \(L_*(\varepsilon)\to\infty\):

\begin{enumerate}[label=\textup{(A\arabic*)},leftmargin=2.8em]
\item\label{it:infty-A1}
\emph{Midpoint-circle coefficient bound.} There exists
\(M_0<\infty\) such that, with
\(\rho(\varepsilon):=\frac{1+\rho_*(\varepsilon)}{2}\),
\[
|R_p(m;\varepsilon)|\le M_0^p\,\rho(\varepsilon)^{-ms}
\qquad (m\ge0,\ p\equiv q \!\!\!\pmod s,\ 0<\varepsilon<\varepsilon_0).
\]

\item\label{it:infty-A2}
\emph{Uniform transfer along a single dominant orbit.} There exist
constants \(M_{\mathrm{tr}},C,L_0>0\) and amplitudes \(A_p(\varepsilon)\) such
that
\[
|A_p(\varepsilon)|\le C\,p\,L_0^{\,p-1}
\]
for all admissible \(p\), and
\[
R_p(m;\varepsilon)
=
A_p(\varepsilon)m^{-3/2}z_*(\varepsilon)^{-m}
+
O\left(
|A_p(\varepsilon)|(1+p^2)m^{-5/2}|z_*(\varepsilon)|^{-m}
\right)
\]
uniformly for \(m\ge M_{\mathrm{tr}}(1+p^2)\) and
\(0<\varepsilon<\varepsilon_0\).

\item\label{it:infty-A3}
\emph{Weighted Gram realization.} One has \(\alpha>\max\{M_0,L_0\}\), and for each
\(\varepsilon\) there exists a bounded synthesis operator
\(\widetilde V^{(q)}(\varepsilon):\ell^2(\mathbb N_0)\to \ell^2(\mathbb N_0)\)
whose Gram block
\(\widetilde G^{(q)}_0(\varepsilon):=\widetilde V^{(q)}(\varepsilon)^*
\widetilde V^{(q)}(\varepsilon)\) has entries
\[
(\widetilde G^{(q)}_0(\varepsilon))_{j_1j_2}
=
\frac{1}{w_{j_1}w_{j_2}}
\sum_{m\ge0}
\frac{(p_{j_2}+sm)^2}{\sqrt{p_{j_1}p_{j_2}}}\,
\overline{R_{p_{j_1}}(m+\Delta;\varepsilon)}\,
R_{p_{j_2}}(m;\varepsilon),
\qquad
\Delta:=j_2-j_1,
\]
for \(j_1\le j_2\), and by Hermitian symmetry otherwise.

\item\label{it:infty-A4}
\emph{Critical spike convergence.} The vectors
\[
\widetilde d^{(q)}(\varepsilon)(j)
:=
e^{-ij\phi(\varepsilon)}
\frac{s}{\sqrt{p_j}}\,
\frac{\overline{A_{p_j}(\varepsilon)}}{w_j},
\qquad j\ge0,
\]
converge in \(\ell^2(\mathbb N_0)\) to a nonzero limit
\(\widetilde d^{(q)}_c\neq0\) as \(\varepsilon\downarrow0\).
\end{enumerate}
\end{assumption}

\begin{proposition}[Abstract \(N=\infty\) rank-one criterion]
\label{prop:universality-infty}
Fix a symmetry block \(q\in\{1,\dots,s\}\), and assume
Assumption~\ref{ass:infty-rank-one}\textup{\ref{it:infty-A1}--\ref{it:infty-A3}}.
Define the spike profile
\[
\widetilde d^{(q)}(\varepsilon)(j)
:=
 e^{-ij\phi(\varepsilon)}
 \frac{s}{\sqrt{p_j}}\,
 \frac{\overline{A_{p_j}(\varepsilon)}}{w_j},
 \qquad j\ge0.
\]
Then there exist bounded operators
\(\widetilde C^{(q)}_0(\varepsilon)\) on \(\ell^2(\mathbb N_0)\) such that
\[
\widetilde G^{(q)}_0(\varepsilon)
=
L_*(\varepsilon)\,
\widetilde d^{(q)}(\varepsilon)\otimes
\widetilde d^{(q)}(\varepsilon)
+
\widetilde C^{(q)}_0(\varepsilon),
\qquad
\sup_{0<\varepsilon<\varepsilon_0}
\|\widetilde C^{(q)}_0(\varepsilon)\|
<\infty .
\]
If, in addition,
Assumption~\ref{ass:infty-rank-one}\textup{\ref{it:infty-A4}} holds, then
\[
\mu_1\!\left(\widetilde G^{(q)}_0(\varepsilon)\right)
=
\|\widetilde d^{(q)}_c\|^2\,L_*(\varepsilon)+O(1),
\]
and, for every \(k\ge2\),
\[
\sup_{0<\varepsilon<\varepsilon_0}
\mu_k\!\left(\widetilde G^{(q)}_0(\varepsilon)\right)
<\infty .
\]
\end{proposition}

\begin{proof}
Fix \(q\), write \(p_j=q+js\), and let
\(\widetilde V^{(q)}(\varepsilon)\) be the synthesis operator from
Assumption~\ref{ass:infty-rank-one}\textup{\ref{it:infty-A3}}. Set
\[
\theta(\varepsilon):=\rho_*(\varepsilon)^{-s}\in(0,1),
\qquad
\eta_*(\varepsilon):=|z_*(\varepsilon)|^{-2}=\rho_*(\varepsilon)^{-2s},
\]
and let \(D(\varepsilon)\) denote the diagonal operator with diagonal
\(\widetilde d^{(q)}(\varepsilon)(j)\). We follow the same decomposition as in
Proposition~\ref{prop:rank-one-extraction}, with
Assumption~\ref{ass:infty-rank-one}\textup{\ref{it:infty-A1}--\ref{it:infty-A2}}
in place of Lemma~\ref{lem:cauchy-Rp} and Lemma~\ref{lem:uniform-transfer}.

Assume first that \(j_2\ge j_1\), set \(\Delta:=j_2-j_1\),
\(P:=\max\{p_{j_1},p_{j_2}\}\), and
\(M:=\lceil M_{\mathrm{tr}}(1+P^2)\rceil\). By
Assumption~\ref{ass:infty-rank-one}\textup{\ref{it:infty-A2}}, for
\(m\ge M\),
\[
R_{p_{j_\nu}}(m+\delta_\nu;\varepsilon)
=
A_{p_{j_\nu}}(\varepsilon)(m+\delta_\nu)^{-3/2}
z_*(\varepsilon)^{-(m+\delta_\nu)}
+
\mathcal E_{\nu,m}(\varepsilon),
\]
with \(\delta_1=\Delta\), \(\delta_2=0\), and
\[
|\mathcal E_{\nu,m}(\varepsilon)|
\le
C\,|A_{p_{j_\nu}}(\varepsilon)|(1+p_{j_\nu}^2)(m+1)^{-5/2}
|z_*(\varepsilon)|^{-(m+\delta_\nu)}.
\]
Substituting into the Gram formula from
Assumption~\ref{ass:infty-rank-one}\textup{\ref{it:infty-A3}}, the product of
the two principal parts gives
\[
\widetilde d^{(q)}_{j_1}(\varepsilon)
\overline{\widetilde d^{(q)}_{j_2}(\varepsilon)}
\theta(\varepsilon)^\Delta
a_{m,\Delta}(P)\eta_*(\varepsilon)^m,
\]
where
\[
a_{m,\Delta}(P):=
\frac{(p_{j_2}+sm)^2}{s^2\,m^{3/2}(m+\Delta)^{3/2}}.
\]
Exactly as in the proof of Proposition~\ref{prop:rank-one-extraction},
\[
a_{m,\Delta}(P)=\frac{1}{m+1}+O\!\left(\frac{1+P^2}{(m+1)^2}\right)
\]
uniformly for \(m\ge M\). Therefore the tail contribution has the form
\begin{equation}\label{eq:infty-tail-structure}
(\widetilde G^{(q)}_{0,j_1j_2}(\varepsilon))_{\mathrm{tail}}
=
\widetilde d^{(q)}_{j_1}(\varepsilon)
\overline{\widetilde d^{(q)}_{j_2}(\varepsilon)}
\theta(\varepsilon)^\Delta
\sum_{m\ge M}\frac{\eta_*(\varepsilon)^m}{m+1} +
\mathcal R^{\infty}_{j_1j_2}(\varepsilon),
\end{equation}
where, after enlarging the constant if necessary,
\[
|\mathcal R^{\infty}_{j_1j_2}(\varepsilon)|
\le
C(1+p_{j_1}^2)(1+p_{j_2}^2)
|\widetilde d^{(q)}_{j_1}(\varepsilon)|
|\widetilde d^{(q)}_{j_2}(\varepsilon)|.
\]
Since \(\alpha>L_0\) and
\(|A_p(\varepsilon)|\le C pL_0^{p-1}\), one has the uniform spike envelope
\[
|\widetilde d^{(q)}(\varepsilon)(j)|
\le
C\,p_j^{-1-\beta}\Bigl(\frac{L_0}{\alpha}\Bigr)^{p_j},
\]
so the right-hand side is a Hilbert--Schmidt envelope on
\(\ell^2(\mathbb N_0)\).

Next, Assumption~\ref{ass:infty-rank-one}\textup{\ref{it:infty-A1}} and the
inequality \(\alpha>M_0\) imply, exactly as in the proof of
Proposition~\ref{prop:rank-one-extraction}, that the portion of the sum with
\(0\le m<M\) defines a uniformly Hilbert--Schmidt operator. It therefore
remains to estimate the omitted tail in \eqref{eq:infty-tail-structure},
namely
\[
\widetilde d^{(q)}_{j_1}(\varepsilon)
\overline{\widetilde d^{(q)}_{j_2}(\varepsilon)}
\theta(\varepsilon)^\Delta
\sum_{m=0}^{M-1}\frac{\eta_*(\varepsilon)^m}{m+1},
\]
is also uniformly Hilbert--Schmidt, since
\(\sum_{m=0}^{M-1}(m+1)^{-1}\le C(1+\log(1+P))\) and the extra logarithmic
factor is absorbed by the exponentially decaying spike envelope. Hence there
exists a uniformly bounded operator \(\widetilde B^{(q)}_0(\varepsilon)\) such
that
\[
\widetilde G^{(q)}_0(\varepsilon)
=
L_*(\varepsilon)
D(\varepsilon)T_{\theta(\varepsilon)}D(\varepsilon)^*
+
\widetilde B^{(q)}_0(\varepsilon),
\qquad
\sup_{0<\varepsilon<\varepsilon_0}
\|\widetilde B^{(q)}_0(\varepsilon)\|<\infty,
\]
where \((T_\theta)_{j_1j_2}=\theta^{|j_1-j_2|}\).

Finally, the same Hilbert--Schmidt estimate as in proof of
Proposition~\ref{prop:rank-one-extraction} gives
\[
L_*(\varepsilon)
\bigl\|D(\varepsilon)(T_{\theta(\varepsilon)}-I)D(\varepsilon)^*\bigr\|_{\mathrm{HS}}
\longrightarrow 0,
\]
since \(1-\theta(\varepsilon)=O(\rho_*(\varepsilon)-1)\) and
\(L_*(\varepsilon)=-\eta_*(\varepsilon)^{-1}\log(1-\eta_*(\varepsilon))
=\log\!\frac{1}{\rho_*(\varepsilon)-1}+O(1)\) as
\(\eta_*(\varepsilon)=\rho_*(\varepsilon)^{-2s}\uparrow1\), hence
\(L_*(\varepsilon)(1-\theta(\varepsilon))\to0\). Thus
\[
\widetilde G^{(q)}_0(\varepsilon)
=
L_*(\varepsilon)
\widetilde d^{(q)}(\varepsilon)\otimes
\widetilde d^{(q)}(\varepsilon)
+
\widetilde C^{(q)}_0(\varepsilon),
\qquad
\sup_{0<\varepsilon<\varepsilon_0}
\|\widetilde C^{(q)}_0(\varepsilon)\|<\infty.
\]
The case \(j_1>j_2\) follows by Hermitian symmetry.

If Assumption~\ref{ass:infty-rank-one}\textup{\ref{it:infty-A4}} also holds,
the min--max argument used in the proof of
Theorem~\ref{thm:universality} applies verbatim to the positive operators
\(\widetilde G^{(q)}_0(\varepsilon)\), yielding
\[
\mu_1\bigl(\widetilde G^{(q)}_0(\varepsilon)\bigr)
=
\|\widetilde d_c^{(q)}\|^2\,L_*(\varepsilon)+O(1),
\qquad
\sup_{0<\varepsilon<\varepsilon_0}
\mu_k\bigl(\widetilde G^{(q)}_0(\varepsilon)\bigr)<\infty
\quad (k\ge2).
\]
\end{proof}

Proposition~\ref{prop:universality-infty} is a conditional
operator-theoretic criterion for infinite-dimensional reductions arising
naturally in Laplacian growth. For orientation,
Section~\ref{sec:infty-explicit-leaves} cleanly separates theorem-level conditions
from characteristic diagnostics. On the single-pole leaf,
Proposition~\ref{prop:single-pole-rho-star} proves the exact identity
\(\rho_*(\zeta)=\rho_{\mathrm{char}}(\zeta)\). On the single-log leaf,
Proposition~\ref{prop:single-log-rho-star} proves the weaker but intrinsic
same-sheet statement that every finite boundary singularity of the continued
Taylor branch is a critical-value obstruction on the logarithmic sheet reached
from the origin. Outside those two propositions, the explicit examples are used
only to test the hypotheses and to organize the characteristic geometry. In
particular, no closed formula for \(\rho_{\mathrm{char}}\) is promoted to a
closed theorem for \(\rho_*\) on logarithmic leaves without an additional
sheet-selection argument.

% =======================

\section{Interpretation in Laplacian growth}
\label{sec:laplacian-growth}

% =======================

The universality mechanism of Theorem~\ref{thm:universality} becomes dynamical
along classical Laplacian growth with injection at infinity: the
Polubarinova--Galin evolution produces a distinguished real--analytic
trajectory $T\mapsto \zeta(T)$ in the reduced parameters~\eqref{eq:reduced-parameters}.
We first recall this trajectory on a fixed polynomial leaf and define the
spectral critical time $T_c$ (Section \ref{subsec:lg-trajectories}). We then apply
Theorem~\ref{thm:universality} along that trajectory, under the same local
hypotheses as in Section~\ref{sec:universal-instability}, to obtain the
logarithmic blow--up of the top variational eigenvalue as $T\uparrow T_c$
(Section \ref{subsec:spectral-lg}), and we discuss how this spectral threshold
compares to geometric breakdown of univalence (Section \ref{subsec:spectral-vs-geom}). Finally, to place the polynomial
ansatz in a broader Laplacian--growth landscape, we consider two explicit
$N=\infty$ leaves (logarithms and poles) that provide explicit characteristic
formulas and principal-sheet numerical illustrations for the sheet-dependent
quantity \(\rho_{\mathrm{char}}(\zeta)\) (Section~\ref{sec:infty-explicit-leaves}).
That section contains two rigorous results for the single-pole and single-log
leaves, together with characteristic diagnostics illustrating the corresponding
sheet structure.

\subsection{Laplacian-growth trajectories on a polynomial leaf}
\label{subsec:lg-trajectories}

Let $D(T)\subset\CC$ be a Laplacian growth domain driven by injection at
infinity, and let
\[
	f(w;T)\colon \{|w|>1\}\to \CC\setminus\overline{D(T)}
\]
be the exterior conformal map normalized by $f(w;T)=r(T)\,w+O(1)$ as
$w\to\infty$, with conformal radius $r(T)>0$.
The evolution is governed by the Polubarinova--Galin equation
\begin{equation}\label{eq:PG}
	\Re\!\left(
	\overline{w\,\partial_w f(w;T)}\;
	\partial_T f(w;T)\right)=1,
	\qquad |w|=1;
\end{equation}
see \cite{Galin1945,PolubarinovaKochina1945}.
We now restrict attention to a Laplacian-growth reduction for which, on a
time interval of interest, the conformal map stays inside the polynomial ansatz
with fixed exponents $2\le s_1<\cdots<s_N$,
\begin{equation}\label{eq:poly-submanifold-lg}
	f(w;T)=r(T)\,w+\sum_{n=1}^{N}a_n(T)\,w^{1-s_n}.
\end{equation}
Equivalently, we assume that the reduction yields a real-analytic family of coefficients solving the Polubarinova--Galin equation on some interval,
and that the corresponding conformal maps remain univalent on $|w|>1$ throughout that interval.

Let $[0,T_{\mathrm{univ}})$ denote the maximal interval on which the chosen
reduction exists and is univalent on $|w|>1$. The reduced parameters
\[
	\zeta_n(T):=\frac{a_n(T)}{r(T)},\qquad n=1,\dots,N,
\]
then define a real--analytic curve $T\mapsto\zeta(T)\in\CC^N$ for
$T\in[0,T_{\mathrm{univ}})$. Since the scale--invariant Hessian $H(\zeta)$
depends only on $\zeta$, any result from
Sections~\ref{sec:framework}--\ref{sec:universal-instability} whose local
hypotheses are satisfied near the trajectory may be applied along this curve.
In particular, once the hypotheses of Theorem~\ref{thm:universality} are
verified near a critical time, the universality mechanism becomes dynamical.
What comes from the Laplacian-growth dynamics itself is the real--analytic
trajectory on the chosen polynomial leaf. The theorem below applies the local finite-mode analysis along that reduced trajectory, rather than deriving the hypotheses of Section~\ref{sec:universal-instability} from the dynamics.

\begin{definition}[Spectral critical time]\label{def:Tc}
	The \emph{spectral critical time} of the trajectory $\zeta(T)$ is
	\begin{equation}\label{eq:Tc-def}
		T_c:=\inf\bigl\{T>0:\;\zeta(T)\in\mathcal C\bigr\},
	\end{equation}
	where $\mathcal C$ is the restricted critical locus
	from Definition~\ref{def:alg-critical-locus}. Thus $T_c$ records the first
	intersection with the part of the radius-one locus where the simple-dominant
	regime needed for the spectral theorem is available, not merely the full set
	$\{\rho_*=1\}$. If the trajectory does not meet~$\mathcal C$, we set
	$T_c:=\infty$.
\end{definition}

\begin{remark}[Moment interpretation]\label{rem:LG-moments}
	For injection at infinity, Richardson's theorem~\cite{Richardson1972}
	implies that the higher harmonic moments (Toda times $t_k$, $k\ge1$)
	are conserved, while the area moment $t_0$ grows linearly in~$T$; see
	also \cite{MineevWeinstein2000}.
	Thus $T\mapsto f(\cdot;T)$ is the distinguished $t_0$--trajectory
	in the integrable hierarchy. The argument relies only on the real-analyticity of $\zeta(T)$ and on transversality at the critical crossing.
\end{remark}

\subsection{Spectral instability along Laplacian growth}
\label{subsec:spectral-lg}

\begin{lemma}[Dominant modulus as a transversal coordinate]
	\label{lem:lg-transversal}
	Let $\zeta(T)$ be a $C^1$ trajectory intersecting the critical locus $\mathcal C$
	at a simple critical point $\zeta_c=\zeta(T_c)$. Assume moreover that $T_c$ is the first intersection time, so that $\rho_*(\zeta(T))>1$ for $T<T_c$ sufficiently close to $T_c$.
	Assume that in a neighborhood of $\zeta_c$ the Taylor branch has a unique dominant $s$--orbit of simple branch points,
	with dominant modulus $\rho_*(\zeta)=|x_*(\zeta)|$, and that the scalar nondegeneracy condition
	\[
		\frac{d}{dT}\rho_*(\zeta(T))\Big|_{T=T_c}\neq0
	\]
	holds.
	Then $\rho_*(\zeta(T))$ is $C^1$ near $T_c$ and
	\[
		\rho_*(\zeta(T))-1=\sigma\,(T_c-T)+O\!\left((T_c-T)^2\right)
	\]
	for some $\sigma>0$.
\end{lemma}

\begin{proof}
	By Lemma~\ref{lem:local-branch-data}, together with the assumption that the continued point remains a representative of the unique dominant orbit, one may shrink the neighborhood of
	$\zeta_c$ so that the representative dominant branch point
	$x_*(\zeta)$ depends $C^1$ on $\zeta$, and therefore so does its
	modulus $\rho_*(\zeta)=|x_*(\zeta)|$.
	Along the trajectory, Taylor expansion gives
	\[
		\rho_*(\zeta(T))-1 = (T-T_c)\,\frac{d}{dT}\rho_*(\zeta(T))\big|_{T=T_c}+O((T-T_c)^2).
	\]
	By the scalar nondegeneracy hypothesis of the lemma, the derivative is nonzero. By the subcriticality assumption $\rho_*(\zeta(T))>1$ for $T<T_c$ close to $T_c$, this derivative must in fact be negative, hence $\sigma:=-\frac{d}{dT}\rho_*(\zeta(T))|_{T=T_c}>0$ and the claim follows.
\end{proof}

\begin{theorem}[Spectral instability along Laplacian growth]
	\label{thm:lg-spectral}
	Let $\zeta(T)$ be the reduced-parameter curve of a Laplacian-growth reduction
	on a fixed polynomial leaf, as in Section~\ref{subsec:lg-trajectories}, and
	assume that it is real analytic on a left neighborhood $(T_c-\delta_0,T_c]$
	for some $\delta_0>0$.
	Assume that $T_c<\infty$.
	Suppose that $\zeta(T_c)=\zeta_c\in\mathcal C$
	is a simple critical point and that the scalar nondegeneracy condition
	\[
		\frac{d}{dT}\rho_*(\zeta(T))\Big|_{T=T_c}\neq0
	\]
	holds. Assume further that for subcritical times $T<T_c$ sufficiently close to
	$T_c$ the Taylor branch lies in the dominant regime with a unique dominant
	$s$-orbit of simple branch points, and that the uniform square-root
	continuation condition of Definition~\ref{def:uniform-delta} holds along the
	trajectory near $T_c$.

	Fix $q\in\{1,\dots,s\}$ and, for subcritical $T<T_c$ sufficiently close to
	$T_c$, define
	\[
	\mu_k^{(q)}(T):=\mu_k\bigl(\widetilde G^{(q)}(\zeta(T))\bigr).
	\]
	By Remark~\ref{rem:HG-spectrum}, these levels agree on that interval with
	the nonzero variational levels of
	\(\widetilde H^{(q)}(\zeta(T))\).
	Then there exists $\delta>0$ such that, on the left neighborhood
	$T\in(T_c-\delta,T_c)$, the top variational eigenvalue
	$\mu_1^{(q)}(T)$ diverges as $T\uparrow T_c$, while the higher
	variational eigenvalues remain uniformly bounded:
	\[
		\sup_{T_c-\delta<T<T_c}\mu_k^{(q)}(T)<\infty,
		\qquad k\ge2.
	\]
	More precisely,
	\begin{equation}\label{eq:mu-lg}
		\mu_1^{(q)}(T)
		=
		\Gamma^{(q)}\,\log\frac{1}{T_c-T}+O(1),
		\qquad T\uparrow T_c,
	\end{equation}
	where $\Gamma^{(q)}=\|\widetilde d_c^{(q)}\|^2>0$ is the constant from
	Theorem~\ref{thm:universality}.
\end{theorem}

\begin{proof}
	This is an immediate application of Theorem~\ref{thm:universality} to the
	trajectory \(T\mapsto\zeta(T)\), reparametrized by the distance-to-criticality
	coordinate \(\varepsilon(T):=\rho_*(\zeta(T))-1\).
	Set $\varepsilon(T):=\rho_*(\zeta(T))-1>0$.
	By Lemma~\ref{lem:lg-transversal}, there exists $\sigma>0$ such that
	\begin{equation}\label{eq:eps-transversal}
		\varepsilon(T)=\sigma(T_c-T)+O((T_c-T)^2).
	\end{equation}
	In particular, after shrinking the left neighborhood of \(T_c\) if
	necessary, \(\varepsilon(T)\) is strictly monotone and may be used as a
	local parameter. The path \(T\mapsto\zeta(T)\), reparametrized by
	\(\varepsilon(T)\downarrow0\), therefore satisfies the assumptions of
	Theorem~\ref{thm:universality} on that neighborhood.
	For $T<T_c$ sufficiently close to $T_c$, Remark~\ref{rem:HG-spectrum}
	identifies the nonzero min--max levels of
	$\widetilde H^{(q)}(\zeta(T))$ and $\widetilde G^{(q)}(\zeta(T))$.
	Applying Theorem~\ref{thm:universality} to the Gram block along the
	trajectory therefore gives
	\[
		\mu_1^{(q)}(T)
		=
		\Gamma^{(q)}\,L(\varepsilon(T))+O(1).
	\]
	By Lemma~\ref{lem:Lambda-asymp},
	\[
		L(\varepsilon(T))
		=\log\!\bigl(1/\varepsilon(T)\bigr)+O(1)
		=\log\!\bigl(1/(T_c-T)\bigr)+O(1),
	\]
	which yields \eqref{eq:mu-lg}.
	The boundedness of the remaining variational eigenvalues follows from
	Theorem~\ref{thm:universality} applied on a sufficiently small left
	neighborhood of $T_c$.
\end{proof}

	Theorem~\ref{thm:universality} leads to a simple numerical verification for the
	onset of spectral criticality.  Along any path approaching a simple critical
	point transversally from the subcritical side, one expects in each symmetry
	block \(q\) that the leading variational eigenvalue satisfies
	\[
		\mu_1^{(q)}(\varepsilon)=\Gamma^{(q)}L(\varepsilon)+O(1),
		\qquad
		L(\varepsilon)=\log(1/\varepsilon),
	\]
	while the higher eigenvalues remain bounded.  Equivalently,
	\[
		\frac{\mu_1^{(q)}(\varepsilon)}{L(\varepsilon)}\to \Gamma^{(q)},
		\qquad
		\frac{\mu_k^{(q)}(\varepsilon)}{L(\varepsilon)}\to 0
		\quad (k\ge2).
	\]

	In practice, one therefore expects a linear tail fit for the stiff mode and
	vanishing normalized soft modes.  For finite truncations of the renormalized
	blocks \(\widetilde H^{(q)}\) (or, equivalently, \(\widetilde G^{(q)}\)),
	this provides a robust numerical signature of the rank--one logarithmic
	instability predicted by Theorem~\ref{thm:universality}.

\subsection{Spectral criticality precedes geometric breakdown}
\label{subsec:spectral-vs-geom}

We now compare the spectral threshold~$T_c$ with the classical geometric
threshold $T_{\mathrm{univ}}$ at which univalence is lost. The conclusion of this subsection is conditional: it applies only when the reduced map is still univalent at the spectral crossing.

\begin{definition}[Geometric critical time]\label{def:Tuniv}
	The \emph{geometric critical time} is
	\[
		T_{\mathrm{univ}}
		:=
		\sup\bigl\{T>0:\;f(\cdot;T)
		\text{ is univalent on }|w|>1\bigr\}.
	\]
	For polynomial maps on a fixed leaf, loss of univalence typically
	occurs through cusp formation, i.e.\ when $f'(w;T)=0$ for some
	$|w|=1$.
\end{definition}

\begin{definition}[Univalence-breakdown locus]\label{def:calG}
	For $\zeta\in\CC^N$ set
	\[
		g_\zeta(w):=w+\sum_{n=1}^N \zeta_n\,w^{1-s_n},
		\qquad |w|>1,
	\]
	so that $f(w;\zeta)=r\,g_\zeta(w)$ on the polynomial leaf.
	The \emph{univalence-breakdown locus} in parameter space is
	\[
		\mathcal U
		:=
		\bigl\{\zeta\in\CC^N:\;
		g_\zeta \text{ is not univalent on }|w|>1\bigr\}.
	\]
\end{definition}

 \begin{proposition}[Elementary separation under univalence at the spectral threshold]
	\label{prop:zc-before-zuniv}
	Let $\zeta:[0,T]\to\CC^N$ be a continuous curve and let $g_{\zeta(t)}$ be as in Definition~\ref{def:calG}.
	Assume there exists $t_c\in(0,T)$ such that $\rho_*(\zeta(t))>1$ for all $t<t_c$ and
	$\zeta(t_c)\in\mathcal C$ (Definition~\ref{def:alg-critical-locus}).
	Let $t_{\univ}$ be the first time (possibly $+\infty$) at which $g_{\zeta(t)}$ fails to be univalent on $|w|>1$.
	If $\zeta(t_c)\notin\mathcal U$ (Definition~\ref{def:calG}), i.e.\ $g_{\zeta(t_c)}$ \emph{is} univalent on $|w|>1$,
	then
	\[
		t_c < t_{\univ}.
	\]
	When $\zeta(\,\cdot\,)$ is the Laplacian--growth trajectory parameterized by the growth time $T$,
	we write $T_c:=t_c$ and $T_{\univ}:=t_{\univ}$. Hence, if $\zeta(T_c)\notin\mathcal U$ then
	$T_c<T_{\univ}$.
\end{proposition}

\begin{proof}
	Let
	\[
		\Sigma:=\bigl\{g(w)=w+O(w^{-1})\, :\, g \text{ is holomorphic and univalent on } |w|>1\bigr\}
	\]
	be the standard normalized schlicht class on the exterior disk. It is a classical fact that $\Sigma$ is open in the topology
	of locally uniform convergence on compact subsets of $\{|w|>1\}$ together with preservation of the principal part at infinity;
	see, for example, \cite[Ch.~I]{Pommerenke1992} or \cite[Ch.~2]{Duren1983}. Since every map in our family has the same principal part
	$g_\zeta(w)=w+O(w^{-1})$, this openness applies directly to the finite-dimensional family
	$\zeta\mapsto g_\zeta$.

	Because $g_{\zeta(t_c)}\in\Sigma$, there exists a neighborhood $U$ of $\zeta(t_c)$ in $\CC^N$ such that
	$g_{\zeta'}\in\Sigma$ for all $\zeta'\in U$. Equivalently, every such $g_{\zeta'}$ is univalent on $|w|>1$.
	By continuity of $t\mapsto \zeta(t)$, there exists $\delta>0$ such that
	$\zeta(t)\in U$ for all $t\in[t_c,t_c+\delta]\cap[0,T]$. Hence $g_{\zeta(t)}$ remains univalent on $|w|>1$
	throughout that interval, so the first univalence-breakdown time satisfies
	\[
		t_{\univ}\ge t_c+\delta>t_c.
	\]
	This proves the claim.
\end{proof}

	The hypothesis that the trajectory meets the critical locus \(\mathcal C\) is essential.
	On some explicit \(N=\infty\) leaves, a chosen characteristic level may fail to be
	attained in the interior of the parameter range. For example, on the
	principal-sheet single-log diagram one numerically observes that for
	\(\gamma\le\gamma_c\) the level \(\rho_{\mathrm{char}}=1\) is only approached as
	\(b\uparrow1\). What follows in such cases is merely that no finite threshold is
	detected on that chosen sheet from the explicit characteristic formula.
	A conclusion about the actual spectral critical time \(T_c\) then requires an
	identification of the relevant sheet-dependent characteristic value with the true
	analytic radius \(\rho_*\). On the single-pole leaf this identification is exact by
	Proposition~\ref{prop:single-pole-rho-star}. On the single-log leaf,
	Proposition~\ref{prop:single-log-rho-star} shows only that any finite boundary
	singularity of the continued Taylor branch must come from a critical-value
	obstruction on the same logarithmic sheet. Turning the explicit formula into a closed
	radius theorem still requires identifying the active critical value on that sheet.

% =======================

\section{Explicit \texorpdfstring{\(N=\infty\)}{N=infinity} leaves: poles, logarithms, and characteristic boundaries}
\label{sec:infty-explicit-leaves}

% =======================

We conclude with two explicit infinite-dimensional reductions arising naturally in
Laplacian growth. First, we state two
rigorous explicit results: in the single-pole case a closed radius theorem, and
in the single-log case a description of finite boundary singularities as
critical-value obstructions on the continuation sheet reached from the origin.
The final subsection is diagnostic rather than theorem-level: it visualizes the
sheet-dependent quantity \(\rho_{\mathrm{char}}(\zeta)\) on selected sheets. In the single-pole case we prove that the characteristic modulus coincides with
the true analyticity radius,
\[
\rho_*(\zeta)=\rho_{\mathrm{char}}(\zeta).
\]
In the single-log case the conclusion is weaker: every finite boundary
singularity of the continued Taylor branch is shown to arise as a critical-value
obstruction on the logarithmic sheet reached from the origin. The remaining
discussion concerns \(\rho_{\mathrm{char}}\) on the chosen sheet. Except in the
single-pole case, we do not identify \(\rho_{\mathrm{char}}\) with the full
analyticity radius \(\rho_*\).

\begin{definition}[Characteristic system]\label{def:char-infty-system}
Consider the infinite-dimensional analogue of the inverse-map equation,
\begin{equation}\label{eq:inv-eq-infty}
U(x;\zeta)=1+\sum_{k\ge2}\zeta_k\,x^k U(x;\zeta)^k,
\qquad U(0;\zeta)=1,
\end{equation}
and assume that the series
\begin{equation}\label{eq:Psi-infty}
\Psi_\zeta(u):=\sum_{k\ge2}\zeta_k u^k,
\qquad
u:=x\,U(x;\zeta),
\end{equation}
converges in a neighborhood of the relevant values of \(u\) on the chosen
continuation sheet. Then the branch-point condition takes the same form as in the finite-mode case:
\begin{equation}\label{eq:char-infty}
1+\Psi_\zeta(u_*)=u_*\,\Psi_\zeta'(u_*),
\qquad
x_* =\frac{u_*}{1+\Psi_\zeta(u_*)}
=\frac{1}{\Psi_\zeta'(u_*)}.
\end{equation}
Accordingly, on any chosen continuation sheet we define the characteristic modulus by
\begin{equation}\label{eq:rho-char-infty}
\rho_{\mathrm{char}}(\zeta)
:=
\inf\Bigl\{|x_*|:\ (u_*,x_*) \text{ satisfies \eqref{eq:char-infty} on the chosen sheet}\Bigr\}.
\end{equation}
\end{definition}

The quantity \(\rho_{\mathrm{char}}\) denotes only the smallest modulus of a
sheet-dependent critical value. The regimes
\(\rho_{\mathrm{char}}>1\), \(\rho_{\mathrm{char}}=1\), and
\(\rho_{\mathrm{char}}<1\) will be referred to as subcritical, critical, and
supercritical at the characteristic level.

\begin{remark}[Logarithmic sheets]\label{rem:log-sheet}
On logarithmic leaves, the equation determining \(u_*\) becomes rational, but the value
of \(x_*=1/\Psi_\zeta'(u_*)\) still depends on the branch of the logarithm selected by
analytic continuation. Thus \(\rho_{\mathrm{char}}\) is intrinsically
sheet-dependent in the logarithmic examples below. In the single-log figures, the
\emph{principal sheet} means the branch of \(\log(1-bu)\) obtained by analytic
continuation of the germ \(\log(1-bu)\sim -bu\) at \(u=0\) without crossing the
standard branch cut of the logarithm.
\end{remark}

\subsection{Pole leaves and a single pole example}

For the exterior rational map with simple poles
\begin{equation}\label{eq:multi-pole-map}
f(w)=r\,w+\sum_{k=1}^K \frac{A_k}{w-b_k},
\qquad |w|>1,\quad |b_k|<1,
\end{equation}
one obtains
\begin{equation}\label{eq:Psi-multi-pole}
\Psi_\zeta(u)=\frac{u^2}{r}\sum_{k=1}^K \frac{A_k}{1-b_k u}.
\end{equation}
The characteristic equation reduces to a rational condition.

\begin{lemma}[Multi-pole characteristic equation]\label{lem:char-multi-pole}
For \(\Psi_\zeta\) given by \eqref{eq:Psi-multi-pole}, \eqref{eq:char-infty} is equivalent to
\begin{equation}\label{eq:char-multi-pole-simplified}
r=u^2\sum_{k=1}^K \frac{A_k}{(1-b_k u)^2}.
\end{equation}
\end{lemma}

\begin{proof}
Direct substitution of \eqref{eq:Psi-multi-pole} into \eqref{eq:char-infty}.
\end{proof}

For \(K=1\), write \(c:=A/r>0\). Then the characteristic values are explicit:
\begin{equation}\label{eq:xstar-one-pole-explicit}
x_*^\pm=\frac{1}{b\pm2\sqrt c},
\qquad
\rho_{\mathrm{char}}(b,c)=\min\{|x_*^+|,|x_*^-|\}.
\end{equation}
On the real slice \(b\in(-1,1)\), the characteristic level \(\rho_{\mathrm{char}}=1\)
consists of the two curves
\begin{equation}\label{eq:pole-critical-curves}
b+2\sqrt c=1,
\qquad
b-2\sqrt c=-1,
\end{equation}
meeting at \((b,c)=(0,1/4)\); see
Figure~\ref{fig:infty-phase-diagrams}\textup{(a)}.
In this case the characteristic description coincides with the true analyticity boundary.

\begin{proposition}[Single pole: \(\rho_*=\rho_{\mathrm{char}}\)]
\label{prop:single-pole-rho-star}
Fix \(c>0\) and \(b\in\mathbb C\) with \(|b|<1\), and let \(u(x)\) be the germ at \(x=0\)
defined by
\[
u=x\Bigl(1+\frac{c u^2}{1-bu}\Bigr).
\]
Equivalently, \(u\) is the local inverse at \(u=0\) of
\[
X(u)=\frac{u(1-bu)}{1-bu+c u^2}.
\]
If \(\rho_*\) denotes the radius of analyticity of this germ, then
\[
\rho_*=\rho_{\mathrm{char}}(b,c)=\min\{|x_*^+|,|x_*^-|\},
\qquad
x_*^\pm=\frac{1}{b\pm2\sqrt c}.
\]
\end{proposition}

\begin{proof}
Since \(X(0)=0\) and \(X'(0)=1\), the germ \(u(x)\) is the local inverse of \(X\) near
\(0\). Any finite boundary singularity of the inverse must therefore occur at a finite
critical value of the rational map \(X\). The critical-point equation is equivalent to
\((1-bu)^2=c u^2\), which yields the two critical points \(u_\pm=(b\pm\sqrt c)^{-1}\) and
hence the two critical values \(x_*^\pm=(b\pm2\sqrt c)^{-1}\).

Conversely, the inverse relation may be written as the quadratic equation
\[
(c x+b)u^2-(b x+1)u+x=0.
\]
Its discriminant is
\[
\Delta(x)=(1+b x)^2-4x(b+c x)=1-2b x+(b^2-4c)x^2,
\]
which vanishes precisely at \(x=x_*^\pm\). Thus \(x_*^\pm\) are genuine branch values of
this algebraic inverse and hence genuine singularities of the continued germ. It remains only to exclude an additional singularity at the point where the
quadratic coefficient vanishes, namely \(x=-b/c\). At that value the equation becomes
linear unless \(1-b^2/c=0\). In the exceptional case \(c=b^2\), one has
\(\Delta(-b/c)=0\), so the point is already one of the discriminant values. Thus no new
finite singularity arises from \(cx+b=0\). Since the inverse branch is algebraic of
degree two, the discriminant points are therefore its only finite singularities, and
\(\rho_*\) is the distance from the origin to the nearer of \(x_*^\pm\).
\end{proof}

\subsection{Logarithmic leaves and a single log example}

For the exterior multi-logarithmic map, on a chosen logarithmic sheet,
\begin{equation}\label{eq:multi-log-map}
f(w)=r\,w+\sum_{k=1}^K c_k \log\Bigl(1-\frac{b_k}{w}\Bigr)+\mathrm{const},
\qquad |w|>1,\quad |b_k|<1,
\end{equation}
one obtains
\begin{equation}\label{eq:Psi-multi-log}
\Psi_\zeta(u)=\frac{u}{r}\sum_{k=1}^K c_k\,\log(1-b_k u).
\end{equation}
Here the logarithms cancel from the equation determining \(u_*\), but they re-enter in the
corresponding value \(x_*\), so the characteristic boundary remains sheet-dependent.

\begin{lemma}[Multi-log characteristic equation]\label{lem:char-multi-log}
For \(\Psi_\zeta\) given by \eqref{eq:Psi-multi-log}, \eqref{eq:char-infty} is equivalent to
\begin{equation}\label{eq:char-multi-log-rational}
r=-u^2\sum_{k=1}^K \frac{c_k b_k}{1-b_k u}.
\end{equation}
\end{lemma}

\begin{proof}
Direct substitution of \eqref{eq:Psi-multi-log} into \eqref{eq:char-infty}.
\end{proof}

For \(K=1\), with \(b\in(0,1)\) and \(\gamma:=c/r>0\), the characteristic equation becomes
\begin{equation}\label{eq:char-one-log-quadratic}
\gamma b\,u^2-bu+1=0.
\end{equation}
Thus the characteristic points \(u_*^\pm\) are explicit. Writing
\begin{equation}\label{eq:X-one-log}
X(u):=\frac{u}{1+\gamma u\log(1-bu)},
\end{equation}
the corresponding characteristic values are
\begin{equation}\label{eq:xstar-one-log}
x_*^\pm=X(u_*^\pm),
\qquad
\rho_{\mathrm{char}}(b,\gamma):=\min\{|x_*^+|,|x_*^-|\}.
\end{equation}
Because of the logarithm, this formula still depends on the continuation sheet. The next result therefore characterizes the boundary obstruction without giving an explicit expression for $\rho_*$.

\begin{proposition}[Single log: finite boundary singularities are critical-value obstructions]
\label{prop:single-log-rho-star}
Fix \(\gamma>0\) and \(b\in(0,1)\), and let \(u(x)\) be the germ at \(x=0\) determined by
\[
u=x\bigl(1+\gamma u\log(1-bu)\bigr),
\qquad
\log(1-bu)\sim -bu
\quad (u\to0).
\]
Let
\[
X(u):=\frac{u}{1+\gamma u\log(1-bu)},
\]
and let \(\rho_*\) be the radius of analyticity of the germ. Then \(1-bu(x)\neq0\) for
\(|x|<\rho_*\), so the logarithm continued from \(0\) defines a holomorphic branch on the
disk \(|x|<\rho_*\). If \(\rho_*<\infty\), every finite boundary singularity \(x_0\) of the
continued germ is realized as a nonzero critical value of \(X\) on the logarithmic sheet
reached by continuation from the origin.
\end{proposition}

\begin{proof}
By the local inverse theorem, \(u(x)\) is the local inverse of
\(X(u)=u\bigl(1+\gamma u\log(1-bu)\bigr)^{-1}\) near \(x=0\).
We first claim that \(1-bu(x)\neq0\) for \(|x|<\rho_*\).
Indeed, if \(1-bu(x_0)=0\) at some point of analyticity, continue the logarithm from its germ at
\(0\) along a path from \(0\) to \(x_0\). As one approaches the first zero of \(1-bu\),
the continued logarithm \(\ell(x):=\log(1-bu(x))\) satisfies \(|\ell(x)|\to\infty\), while
\(u(x)\) stays finite and nonzero. Then
\(u=x(1+\gamma u\ell)\) forces the right-hand side to diverge, a contradiction.
Hence \(1-bu\) is nonvanishing on \(|x|<\rho_*\), so on the chosen sheet
\[
x=\frac{u(x)}{1+\gamma u(x)\ell(x)},\qquad \ell(x):=\log(1-bu(x)).
\]

Assume now that \(\rho_*<\infty\), and let \(x_0\) be a boundary point across which this
continuation does not extend holomorphically. Choose \(x_n\to x_0\) with \(|x_n|<\rho_*\), and set
\(u_n:=u(x_n)\), \(\ell_n:=\ell(x_n)\). Since \(x_0\neq0\), the identity
\(x_n=u_n(1+\gamma u_n\ell_n)^{-1}\) excludes the three possibilities
\(u_n\to0\), \(u_n\to1/b\), and \(|u_n|\to\infty\): if \(u_n\to0\), then \(x_n\to0\), and if
\(u_n\to1/b\), then \(\Re \ell_n=\log|1-bu_n|\to-\infty\), hence
\(|1+\gamma u_n\ell_n|\to\infty\) and again \(x_n\to0\). Besides, if \(|u_n|\to\infty\), then
\[
\frac1{x_n}=\frac1{u_n}+\gamma\ell_n
\]
implies \(\ell_n\to(\gamma x_0)^{-1}\), whereas on the sheet continued from the origin one has
\(\ell_n=\log(1-bu_n)\) and therefore \(|\ell_n|\to\infty\), a contradiction.
Thus, after passing to a subsequence,
\[
u_n\to u_0\in\mathbb C\setminus\{0,1/b\}.
\]

Since \(u_n\to u_0\) and \(x_n\to x_0\neq0\), the quotients \(u_n/x_n\) are bounded. Hence
\(u_n\ell_n=(u_n/x_n-1)/\gamma\) is bounded, and because \(u_0\neq0\), so is \(\ell_n\).
Passing to a further subsequence, let \(\ell_n\to\ell_0\). Then
\[
e^{\ell_0}=\lim_{n\to\infty}e^{\ell_n}=\lim_{n\to\infty}(1-bu_n)=1-bu_0,
\qquad
x_0=\lim_{n\to\infty}\frac{u_n}{1+\gamma u_n\ell_n}
=\frac{u_0}{1+\gamma u_0\ell_0}.
\]
Thus \(\ell_0\) determines a local branch \(\ell_u\) of \(\log(1-bu)\) near \(u_0\), and
\[
X_{\ell_u}(u):=\frac{u}{1+\gamma u\,\ell_u(u)}
\]
satisfies \(X_{\ell_u}(u_0)=x_0\). If \(X'_{\ell_u}(u_0)\neq0\), the inverse function theorem
would produce a holomorphic inverse branch near \(x_0\). By uniqueness, it would coincide with
the given continuation of \(u(x)\) for all large \(n\), yielding an analytic continuation across
\(x_0\), contrary to the choice of \(x_0\). Therefore \(X'_{\ell_u}(u_0)=0\), so \(x_0\) is a
nonzero critical value on the continuation sheet.
\end{proof}

\begin{remark}[Critical points on a fixed logarithmic sheet]
On any local branch \(\ell(u)=\log(1-bu)\), one has
\[
X_\ell'(u)
=
\frac{1+\gamma b\,u^2/(1-bu)}{(1+\gamma u\,\ell(u))^2}
=
\frac{\gamma b\,u^2-bu+1}
{(1-bu)\bigl(1+\gamma u\,\ell(u)\bigr)^2}.
\]
Accordingly, every finite critical point of a holomorphic branch of \(X\) is a root of
\[
\gamma b\,u^2-bu+1=0.
\]
Conversely, any root of this quadratic that lies on the chosen holomorphic sheet is a
critical point of the corresponding branch \(X_\ell\).
\end{remark}

\begin{remark}[Principal-sheet interpretation and scope]
\label{rem:single-log-principal-sheet}
Proposition~\ref{prop:single-log-rho-star} is intrinsic: it identifies the boundary
obstruction on the actual continuation sheet of the Taylor branch. By contrast, the explicit
formula \eqref{eq:xstar-one-log} becomes a closed radius formula only after one identifies
the active logarithmic sheet. In the figures below we use the principal sheet in the sense
of Remark~\ref{rem:log-sheet}. Direct numerical continuation of the Taylor branch along rays
in the \(x\)-plane indicates that, in the plotted regime, the first singularity agrees with
the active principal-sheet characteristic value. The figures should therefore be read only
as principal-sheet characteristic diagrams supported by numerical continuation.
\end{remark}

\subsection{Characteristic verification on selected sheets}
The remainder of this section is diagnostic rather than theorem-level. It presents the characteristic geometry on selected sheets and uses the explicit
formulas above only to visualize the behavior of
\(\rho_{\mathrm{char}}\). In the single-pole case this still describes the
true analyticity boundary by Proposition~\ref{prop:single-pole-rho-star}. In
the single-log case it describes only the principal-sheet characteristic
boundary, in the sense of Remark~\ref{rem:single-log-principal-sheet}.

Figure~\ref{fig:infty-phase-diagrams} summarizes the contrast between the two
one-parameter cases discussed above. Plot~\textup{(a)} shows the single-pole phase
diagram on the real slice \(b\in(-1,1)\): the two explicit curves
\eqref{eq:pole-critical-curves} meet at the apex \((b,c)=(0,1/4)\) and bound the region
\(\rho_{\mathrm{char}}=\rho_*>1\). Thus the figure gives the exact analyticity boundary in
the single-pole case. Plot~\textup{(b)} shows instead a principal-sheet
characteristic diagram for the single-log leaf.
The level curves of \(\rho_{\mathrm{char}}(b,\gamma)\) organize the characteristic
geometry, while the thick curve \(\rho_{\mathrm{char}}=1\) marks the principal-sheet
characteristic boundary discussed in
Remark~\ref{rem:single-log-principal-sheet}. The juxtaposition is intended to make the
logical difference visible: plot~\textup{(a)} is a theorem for \(\rho_*\), whereas
plot~\textup{(b)} is a principal-sheet characteristic picture supported by numerical
continuation.

\begin{figure}[t]
\centering
\includegraphics[width=0.98\textwidth]{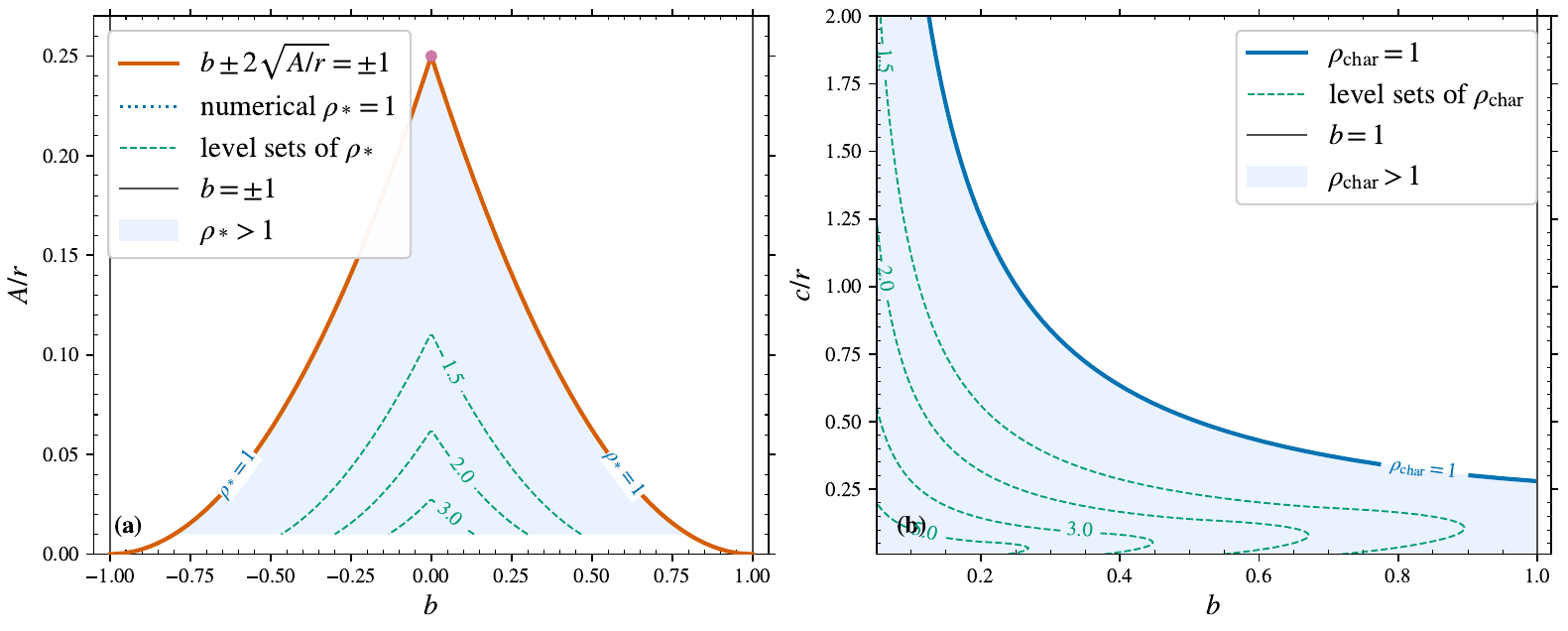}
\caption{\textbf{Characteristic boundaries for two explicit \(N=\infty\) leaves.}
\emph{(a) Single pole.} On the real slice \(b\in(-1,1)\), the level
\(\rho_{\mathrm{char}}=1\) is given by the two explicit curves
\(b+2\sqrt c=1\) and \(b-2\sqrt c=-1\), meeting at \((b,c)=(0,1/4)\). By
Proposition~\ref{prop:single-pole-rho-star}, this is also the exact phase boundary for the
true analyticity radius \(\rho_*\).
\emph{(b) Single log.} Principal-sheet level curves of the characteristic modulus
\(\rho_{\mathrm{char}}(b,\gamma)\). The thick curve is the level
\(\rho_{\mathrm{char}}=1\). By
Remark~\ref{rem:single-log-principal-sheet}, this should be read only as a
principal-sheet characteristic boundary supported by numerical continuation.}
\label{fig:infty-phase-diagrams}
\end{figure}

For the single-log leaf one should distinguish the characteristic boundary
\(\rho_{\mathrm{char}}(b,\gamma)=1\) from the discriminant transition of
\eqref{eq:char-one-log-quadratic}. Using the principal-sheet formula
\eqref{eq:xstar-one-log} and a direct numerical scan of the level condition
\(\rho_{\mathrm{char}}(b,\gamma)=1\) on a fine \((b,\gamma)\)-grid, one finds
an apparent principal-sheet threshold
\begin{equation}\label{eq:alpha-c-log}
\gamma_c\approx0.27997,
\end{equation}
with the following behavior in the plotted regime: for \(\gamma>\gamma_c\), the level
\(\rho_{\mathrm{char}}=1\) intersects \(b\in(0,1)\) at a unique point \(b=b_c(\gamma)<1\),
whereas for \(\gamma\le\gamma_c\) it is reached only in the limit \(b\uparrow1\). This value
is included only as an empirical guide to the phase diagram. The separate
discriminant transition discussed below is visualized in
Figure~\ref{fig:single-log-discriminant}.

A distinct phenomenon occurs at the discriminant line
\begin{equation}\label{eq:bdisc-one-log}
b_{\mathrm{disc}}=4\gamma.
\end{equation}
There the two characteristic points coalesce. Below it they form a complex-conjugate pair,
and above it they are real. The resulting branch splitting changes the visible shape of the
active modulus \(b\mapsto \rho_{\mathrm{char}}(b)\), but it does not by itself impose the
criticality condition \(\rho_{\mathrm{char}}=1\); this is illustrated in
Figure~\ref{fig:single-log-discriminant}. In
plot~\textup{(a)}, for the representative slice \(\gamma=0.05\), the two characteristic
moduli coincide to the left of the discriminant point
\(b_{\mathrm{disc}}=4\gamma\) and split to the right. The solid curve is the active
envelope \(\rho_{\mathrm{char}}(b)=\min\{|x_*^\pm(b)|\}\), the dotted curve is the
complementary branch, and the horizontal level \(\rho_{\mathrm{char}}=1\) shows that this
branch splitting need not coincide with criticality. Plot~\textup{(b)} shows the same
effect for several values of \(\gamma\): the visible change in the active envelope occurs
at the discriminant line. The figure therefore
separates two different notions that would otherwise be easy to confuse, namely branch
splitting and the condition \(\rho_{\mathrm{char}}=1\).

\begin{figure}[t]
\centering
\includegraphics[width=0.98\textwidth]{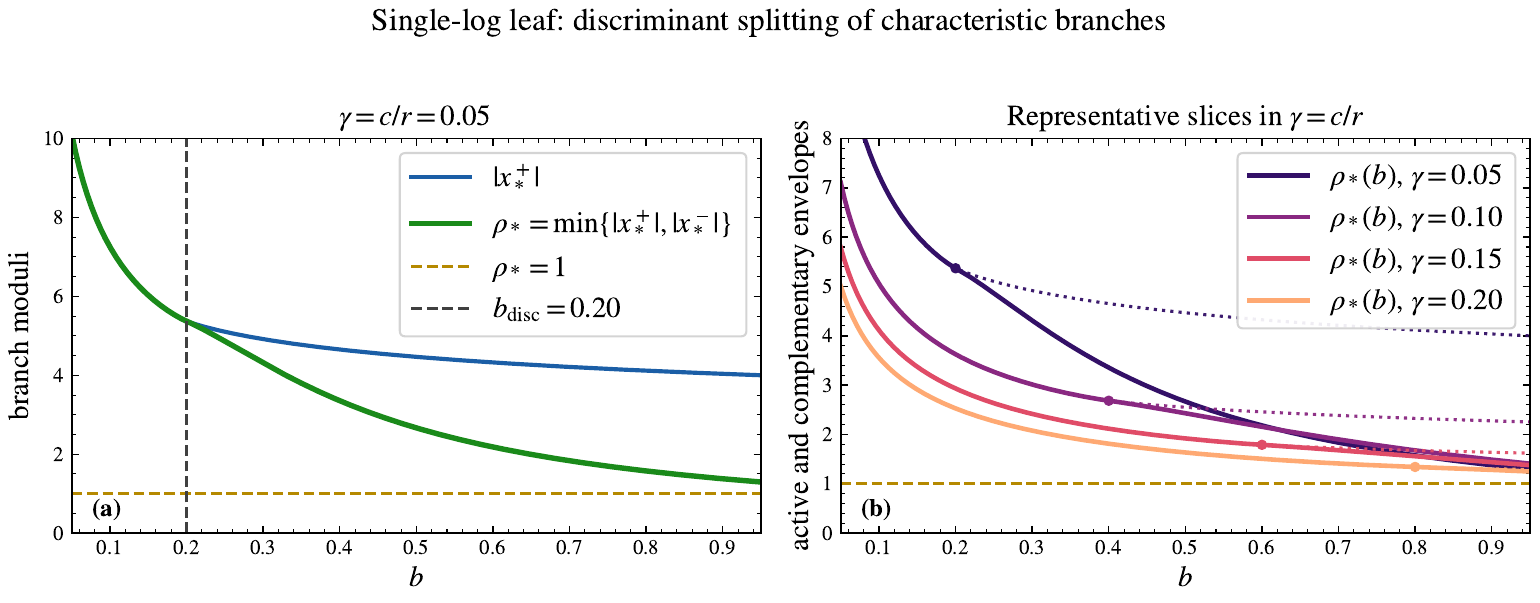}
\caption{\textbf{Single-log leaf: discriminant transition on the principal sheet.}
\emph{(a)} Representative slice \(\gamma=0.05\). The vertical line marks the
discriminant point \(b_{\mathrm{disc}}=4\gamma\). The solid curve is the active
characteristic modulus \(\rho_{\mathrm{char}}=\min\{|x_*^\pm|\}\), the dotted curve is the
complementary branch, and the horizontal line is the level \(\rho_{\mathrm{char}}=1\).
\emph{(b)} Representative slices for several values of \(\gamma\). The change of shape at
the marked points is caused by branch splitting at the discriminant line.}
\label{fig:single-log-discriminant}
\end{figure}

For the explicit leaves considered here, the infimum in \eqref{eq:rho-char-infty} is
attained. If two distinct characteristic solutions have the same minimal modulus, the
dominant orbit is no longer unique. This is the natural case in which one expects the
rank-one mechanism to be replaced by a higher-rank instability. By contrast, the
discriminant line \eqref{eq:bdisc-one-log} corresponds to a double characteristic point and
thus to a degenerate, rather than multiple, obstruction.

\subsection{Interpretation: stiffness.}\label{subsec:lg-stiffness}
We close with the corresponding physical interpretation. Because the higher harmonic moments are conserved in deterministic Laplacian growth, the dynamics does not directly involve the mixed Hessian. The Hessian enters when one studies fluctuations around the deterministic
trajectory, for example under stochastic forcing or in the \(1/n\)-expansion of the
associated normal-matrix model. In that framework, large eigenvalues of \(H\) correspond to
stiff fluctuation directions, whereas bounded eigenvalues correspond to soft ones.

Theorem~\ref{thm:lg-spectral} therefore has a clear interpretation: as \(T\uparrow T_c\),
one variational direction in each symmetry block becomes asymptotically infinitely stiff,
while the higher variational levels remain bounded. Analytic criticality is thus a
rank-one stiffness transition in the space of times, and by
Proposition~\ref{prop:zc-before-zuniv} it may occur before geometric loss of univalence.

At the level of physical interpretation, this suggests that the Gaussian fluctuation
approximation breaks down first along the stiff mode, while no analogous singular response
is forced in the complementary directions at \(T_c\).

% =======================

\section{Conclusion}
\label{sec:conclusion}

% =======================

We have identified a local mechanism of spectral instability for the mixed
Hessian of the dispersionless Toda $\tau$--function on polynomial leaves,
expressed in the scale--free parameters $\zeta_n=a_n/r$. The argument combines
an exact blockwise Gram representation of the Hessian, the algebraic equation
for the inverse-map generating function $U(x;\zeta)$, and uniform analytic
control up to the first dominant obstruction.

Theorem~\ref{thm:universality} gives a conditional local universality
statement: near a simple critical point with a unique dominant
$s$--orbit of simple square--root branch points, and assuming the uniform
continuation hypothesis of Definition~\ref{def:uniform-delta} along the chosen
subcritical approach, each symmetry block develops a single logarithmically
diverging direction after fixed weighted renormalization. More precisely, one
variational eigenvalue grows like \(\Gamma^{(q)}\log(1/\varepsilon)\), while the
higher variational levels remain bounded. After subtraction of the singular
rank--one term, the remainder stays uniformly bounded and converges to a
compact limit.

A key consequence is a conditional separation between \emph{spectral} and \emph{geometric}
criticality. Along Laplacian growth trajectories, if the algebraic critical
crossing does not lie on the nonunivalence locus, then the spectral critical
time satisfies $T_c<T_{\mathrm{univ}}$. Thus the Hessian detects an analytic
instability that can precede geometric breakdown of the interface when the map remains univalent at the spectral crossing.

The explicit one-harmonic family studied in~\cite{Alekseev2025-1} exhibits the
same phenomenon in a globally solvable case. The present paper isolates the
local argument that survives beyond that special case, once the continuation
hypothesis has been verified on the polynomial leaf under consideration.
Natural extensions include an intrinsic verification theory for
Definition~\ref{def:uniform-delta}, nongeneric critical points, competing
dominant orbits leading to genuine rank--$r$ instabilities, and
infinite-source limits beyond the isolated-orbit regime.

\section*{Acknowledgments}
The study was implemented in the framework of the Basic Research Program at HSE
University (HSE-BR-2025-84).

\appendix
\section{Algebraic branching and coefficient transfer}
\label{app:coeff-asymp}

This appendix presents the auxiliary analytic statements from
Section~\ref{sec:alg-criticality} used later in the universality argument.
Section~\ref{sec:alg-criticality} provides the branching analysis and coefficient asymptotics used in the operator-theoretic argument of Appendix~\ref{app:gram-bounds}.

\subsection{Characteristic points and square--root branching}
\label{app:xi-star-system}

\begin{proof}[Proof of Proposition~\ref{prop:xi-star-system}]
Fix $\zeta$ and abbreviate $F(y,x):=F(y,x;\zeta)$. Let $x_*\neq0$ be a singular point of the analytic continuation of the Taylor sheet, and let $\lambda$ be the corresponding finite branch value attained on the chosen continuation at $x_*$. Then $F(\lambda,x_*)=0$. If $\partial_yF(\lambda,x_*)\neq0$, the analytic implicit function theorem would produce an analytic solution $y=\Lambda(x)$ to $F(y,x)=0$ near $x_*$. Since this coincides with the continued Taylor sheet on a punctured neighborhood of $x_*$, the sheet would extend analytically across $x_*$, a contradiction. Hence $\partial_yF(\lambda,x_*)=0$, which is \eqref{eq:char-system-F}. Expanding \eqref{eq:char-system-F} gives \eqref{eq:char-system}.

	Assume now \eqref{eq:char-system-F}, \eqref{eq:Fx-nondeg}, and \eqref{eq:Fyy-nondeg}. For the present polynomial $F$ one has

	\begin{equation}\label{eq:Fx-nonzero-cmp}
		\partial_xF(\lambda,x_*)=-\lambda/x_*\neq0.
	\end{equation}
	Indeed,
	\[
		\partial_xF(y,x)=-\sum_{n=1}^N s_n\zeta_n\,x^{s_n-1}y^{s_n}
		=-\frac{y}{x}\sum_{n=1}^N s_n\zeta_n\,x^{s_n}y^{s_n-1},
	\]
	and the second equation in \eqref{eq:char-system} gives $\sum_n s_n\zeta_n x_*^{s_n}\lambda^{s_n-1}=1$. Thus \eqref{eq:Fx-nonzero-cmp} holds, so \eqref{eq:Fx-nondeg} is automatic whenever $0<|x_*|<\infty$.

	Since $\partial_xF(\lambda,x_*)\neq0$, the analytic implicit function theorem solves $F(y,x)=0$ locally for $x$ as an analytic function of $y$: there exists analytic $x=x(y)$ near $y=\lambda$ with $x(\lambda)=x_*$ and $F(y,x(y))\equiv0$. Differentiating gives $x'(\lambda)=0$ because $F_y(\lambda,x_*)=0$, and a second differentiation yields
	\[
		x''(\lambda)=-\frac{F_{yy}(\lambda,x_*)}{F_x(\lambda,x_*)}\neq0
	\]
	by \eqref{eq:Fyy-nondeg} and \eqref{eq:Fx-nonzero-cmp}. Hence
	\[
		x(y)=x_* + c\,(y-\lambda)^2+O\bigl((y-\lambda)^3\bigr),\qquad c\neq0,
	\]
	and inversion gives two local branches
	\[
		y-\lambda=\pm c'\,\sqrt{x-x_*}\,\bigl(1+O(x-x_*)\bigr),\qquad c'\neq0.
	\]
	The continued Taylor branch coincides with one of them on a punctured neighborhood of $x_*$, which gives the claimed square-root form.
\end{proof}

\subsection{Transfer from a dominant orbit}
\label{app:transfer-dominant}

\begin{lemma}[Transfer from a dominant simple $s$--orbit]\label{lem:coeff-asymp-dominant}
	Fix $\zeta$ and assume Definition~\ref{def:dominant-orbit}. Let $x_*(\zeta)$ be a representative
	dominant singularity so $|x_*(\zeta)|=\rho_*(\zeta)$, and assume that $U(\cdot;\zeta)$ is
	$\Delta$--analytic at each orbit point $\{e^{2\pi i j/s}x_*(\zeta)\}_{j=0}^{s-1}$ and that the
	representative point is \emph{simple} on the Taylor sheet in the sense of Proposition~\ref{prop:xi-star-system} (so the nondegeneracy conditions \eqref{eq:Fx-nondeg} and \eqref{eq:Fyy-nondeg} hold for the corresponding branch value $\lambda$ at $x_*$).
	Then for every $p\ge1$, writing
	\[
		U(x;\zeta)^p=\sum_{m\ge0} R_p(m;\zeta)\,x^{ms},
	\]
	one has the asymptotics
	\begin{equation}\label{eq:Rp-asymp}
		R_p(m;\zeta)=A_p(\zeta)\,x_*(\zeta)^{-ms}\,m^{-3/2}\Bigl(1+O(m^{-1})\Bigr),
		\qquad m\to\infty,
	\end{equation}
	and there exists $C>0$ (locally in $\zeta$) such that
	\[
		|A_p(\zeta)|\le C\,p\,|\lambda|^{p-1}.
	\]
\end{lemma}

\begin{proof}[Proof of Lemma~\ref{lem:coeff-asymp-dominant}]
	Work in the variable $z=x^s$. By Lemma~\ref{lem:s-symmetry} there exists an analytic germ
	$\widehat U(z;\zeta)$ near $z=0$ such that $U(x;\zeta)=\widehat U(z;\zeta)$, and therefore
	\[
	R_p(m;\zeta)=[z^m]\,\widehat U(z;\zeta)^p.
	\]
	Because the dominant singularities of $U(\cdot;\zeta)$ form exactly one $s$--orbit
	$\{e^{2\pi i j/s}x_*(\zeta)\}_{j=0}^{s-1}$, the function $\widehat U(\cdot;\zeta)$ has a unique dominant
	singularity at
	\[
		z_*(\zeta):=x_*(\zeta)^s,
		\qquad |z_*(\zeta)|=\rho_*(\zeta)^s.
	\]
	By Proposition~\ref{prop:xi-star-system}, the Taylor sheet has a local expansion at the representative
	point of the form
	\[
		U(x;\zeta)=\lambda+\kappa\sqrt{1-x/x_*(\zeta)}+O\!\left(1-x/x_*(\zeta)\right),
		\qquad x\to x_*(\zeta),
	\]
	with $\kappa\neq0$. Since $z/z_*(\zeta)=(x/x_*(\zeta))^s$ and
	\[
		1-\frac{z}{z_*(\zeta)}=1-\Bigl(\frac{x}{x_*(\zeta)}\Bigr)^s
		=s\Bigl(1-\frac{x}{x_*(\zeta)}\Bigr)+O\!\left(\Bigl(1-\frac{x}{x_*(\zeta)}\Bigr)^2\right),
	\]
	the same branch may be rewritten in the $z$-variable as
	\[
		\widehat U(z;\zeta)
		=
		\lambda+s^{-1/2}\kappa\,\sqrt{1-z/z_*(\zeta)}+O\!\left(1-z/z_*(\zeta)\right),
		\qquad z\to z_*(\zeta),
	\]
	Hence
	\[
	\widehat U(z;\zeta)^p
		=
	\lambda^p+s^{-1/2}p\kappa\lambda^{p-1}\sqrt{1-z/z_*(\zeta)}+O\!\left(1-z/z_*(\zeta)\right).
	\]
	The function $\widehat U(\cdot;\zeta)^p$ is $\Delta$--analytic at $z_*(\zeta)$ and has no other dominant
	singularity on $|z|=|z_*(\zeta)|$. Therefore the standard transfer theorem
	(\cite[Ch.~VI, Thm.~VI.3]{FlajoletSedgewick2009}) yields
	\[
		[z^m]\,\widehat U(z;\zeta)^p
		=
		-\frac{s^{-1/2}}{2\sqrt{\pi}}\,p\kappa\lambda^{p-1}
		z_*(\zeta)^{-m}m^{-3/2}\Bigl(1+O(m^{-1})\Bigr).
	\]
	Since $z_*(\zeta)=x_*(\zeta)^s$, this is \eqref{eq:Rp-asymp} with
	\[
	A_p(\zeta):=-\frac{s^{-1/2}}{2\sqrt{\pi}}\,p\kappa\lambda^{p-1},
	\]
	up to the conventional choice of the square-root branch. The amplitude bound follows immediately:
	$|A_p(\zeta)|\le C\,p\,|\lambda|^{p-1}$ locally in $\zeta$.
\end{proof}

\subsection{Critical limit of the transfer amplitudes}
\label{app:transfer-critical}

\begin{lemma}[Critical limit of the transfer amplitudes]\label{lem:critical-coeff-asymp}
	Let $\zeta(\varepsilon)$ be a $C^1$ path with $\zeta(0)=\zeta_c\in\mathcal C$ and $\varepsilon\downarrow0$.
	Assume Definition~\ref{def:dominant-orbit} along $\zeta(\varepsilon)$ and that the dominant orbit consists
	of simple branch points on the Taylor sheet. Choose a representative $x_*(\varepsilon)$ and let
	$\lambda(\varepsilon)$ be the corresponding branch value defined in \eqref{eq:lambda-def}.
	Then there exist limits
	\[
		x_*(\varepsilon)\to x_{*,c},\qquad
		\lambda(\varepsilon)\to \lambda_c,\qquad
		\kappa(\varepsilon)\to\kappa_c\neq0,
	\]
	where $\kappa(\varepsilon)$ is the square--root coefficient in the Puiseux expansion of the Taylor sheet
	at $x=x_*(\varepsilon)$. In the variable \(z=x^s\) from
	Definition~\ref{def:uniform-delta}, the coefficient of
	\(\sqrt{1-z/z_*(\varepsilon)}\) is therefore
	\(s^{-1/2}\kappa(\varepsilon)\), up to the fixed choice of square-root branch.
	Moreover, for each fixed $p\ge1$ the amplitudes $A_p(\varepsilon)$ from Lemma~\ref{lem:coeff-asymp-dominant}
	satisfy
	\[
		A_p(\varepsilon)\to A_p^{(c)}\qquad (\varepsilon\downarrow0),
	\]
	with the explicit limit
	\begin{equation}\label{eq:Ap-critical-explicit}
	A_p^{(c)}
	=
	-\frac{s^{-1/2}}{2\sqrt{\pi}}\;p\,\kappa_c\,\lambda_c^{\,p-1},
	\end{equation}
	again up to the fixed choice of square-root branch. Finally, there exist $\varepsilon_0>0$ and
	$C>0$ such that for all $\varepsilon\in(0,\varepsilon_0)$ and all $p\ge1$,
	\[
		|A_p(\varepsilon)|\le C\,p\,|\lambda(\varepsilon)|^{p-1}.
	\]
\end{lemma}

\begin{proof}[Proof of Lemma~\ref{lem:critical-coeff-asymp}]
	Abbreviate $\zeta=\zeta(\varepsilon)$. By
	Lemma~\ref{lem:local-branch-data}, after shrinking the neighborhood if
	necessary there is a unique $C^1$ continuation
	$(\lambda(\varepsilon),x_*(\varepsilon))$ of the chosen representative
	branch-point parameters for $\varepsilon$ small.
	In particular,
	\[
		x_*(\varepsilon)\to x_{*,c}:=x_*(0),\qquad \lambda(\varepsilon)\to\lambda_c:=\lambda(0).
	\]
	Moreover, by Proposition~\ref{prop:xi-star-system} each $x_*(\varepsilon)$ is a simple branch point of the Taylor sheet,
	hence the Taylor branch has a Puiseux expansion at $x=x_*(\varepsilon)$ of the form
	\begin{equation}\label{eq:U-local-eps-proof}
		U(x;\zeta(\varepsilon))
		=
		\lambda(\varepsilon)+\kappa(\varepsilon)\sqrt{1-x/x_*(\varepsilon)}
		+O\!\left(1-x/x_*(\varepsilon)\right),
	\end{equation}
	For a simple branch point one may choose the square-root branch so that
	\[
		\kappa(\varepsilon)^2=
		\frac{2x_*(\varepsilon)\,\partial_xF(\lambda(\varepsilon),x_*(\varepsilon);\zeta(\varepsilon))}
		{\partial_y^2F(\lambda(\varepsilon),x_*(\varepsilon);\zeta(\varepsilon))}.
	\]
	The right-hand side depends continuously on $\varepsilon$ and stays nonzero, so after fixing one branch sign
	consistently we obtain $\kappa(\varepsilon)\to\kappa_c\neq0$.
	By the change of variable $z=x^s$ one has
	\[
		1-\frac{z}{z_*(\varepsilon)}
		=
		1-\Bigl(\frac{x}{x_*(\varepsilon)}\Bigr)^s
		=
		s\Bigl(1-\frac{x}{x_*(\varepsilon)}\Bigr)
		+
		O\!\left(\Bigl(1-\frac{x}{x_*(\varepsilon)}\Bigr)^2\right),
	\]
	so the coefficient of $\sqrt{1-z/z_*(\varepsilon)}$ is
	$s^{-1/2}\kappa(\varepsilon)$, up to the fixed choice of square-root branch.

	Fix $p\ge1$. Lemma~\ref{lem:coeff-asymp-dominant} applied at $\zeta(\varepsilon)$ yields
	\eqref{eq:Rp-asymp} with amplitude $A_p(\varepsilon)$, and the uniform estimate
	$|A_p(\varepsilon)|\le C\,p\,|\lambda(\varepsilon)|^{p-1}$.

	Using the explicit amplitude formula from the proof of Lemma~\ref{lem:coeff-asymp-dominant}, we obtain
	\[
		A_p(\varepsilon)
		=
		-\frac{s^{-1/2}}{2\sqrt{\pi}}\;
		p\,\kappa(\varepsilon)\,\lambda(\varepsilon)^{p-1}.
	\]
	Passing to the limit gives $A_p(\varepsilon)\to A_p^{(c)}$ with
	\[
		A_p^{(c)}
		=
		-\frac{s^{-1/2}}{2\sqrt{\pi}}\;
		p\,\kappa_c\,\lambda_c^{\,p-1},
	\]
	up to the conventional choice of the square--root branch.
\end{proof}

\section{Uniform transfer and operator-theoretic tail extraction}
\label{app:gram-bounds}

This appendix contains the deferred proofs from
\S\ref{sec:universal-instability}. More precisely, it proves the following
results from the main text: Lemma~\ref{lem:cauchy-Rp},
Lemma~\ref{lem:uniform-transfer}, Proposition~\ref{prop:rank-one-extraction},
and Lemma~\ref{lem:compact-remainder}. Throughout, constants \(C,C',\dots\) may
change from line to line but are uniform for \(\varepsilon\) sufficiently
small.

\subsection{Uniform transfer along the critical approach}

\begin{proof}[Proof of Lemma~\ref{lem:uniform-transfer}]
Work in the variable \(z=x^s\). Since \(U(x;\zeta(\varepsilon))\) depends on
\(x\) only through \(x^s\), write
\[
U(x;\zeta(\varepsilon))=\widehat U(z;\varepsilon),
\qquad
R_p(m;\varepsilon)=[z^m]\widehat U(z;\varepsilon)^p .
\]
Let \(z_*(\varepsilon)=x_*(\varepsilon)^s\) be the unique dominant singularity
of \(\widehat U(\cdot;\varepsilon)\). By
Definition~\ref{def:uniform-delta}, in a uniform \(\Delta\)-domain at
\(z_*(\varepsilon)\) one has
\[
\widehat U(z;\varepsilon)
=
\lambda(\varepsilon)
+
s^{-1/2}\kappa(\varepsilon)\Delta^{1/2}
+
c(\varepsilon)\Delta
+
d(\varepsilon)\Delta^{3/2}
+
O(\Delta^2),
\qquad
\Delta:=1-\frac{z}{z_*(\varepsilon)},
\]
with coefficients uniformly bounded for \(0<\varepsilon<\varepsilon_0\), and
with $|\lambda(\varepsilon)|\ge c_0$, $|\kappa(\varepsilon)|\ge c_0$. Raising to the \(p\)-th power gives
\[
\widehat U(z;\varepsilon)^p
=
\lambda(\varepsilon)^p
+
b_p(\varepsilon)\Delta^{1/2}
+
d_p(\varepsilon)\Delta^{3/2}
+
\mathcal R_p(z;\varepsilon),
\]
where \(\mathcal R_p(\cdot;\varepsilon)\) is analytic at
\(z=z_*(\varepsilon)\). Writing $a(\varepsilon):=s^{-1/2}\kappa(\varepsilon)$,
the coefficient of \(\Delta^{1/2}\) is
\[
b_p(\varepsilon)
=
p\,\lambda(\varepsilon)^{p-1}a(\varepsilon)
=
s^{-1/2}p\,\lambda(\varepsilon)^{p-1}\kappa(\varepsilon).
\]
The coefficient of \(\Delta^{3/2}\) comes only from the monomials
\(\lambda^{p-1}d\Delta^{3/2}\), \(\lambda^{p-2}ac\Delta^{3/2}\), and
\(\lambda^{p-3}a^3\Delta^{3/2}\), hence
\[
d_p(\varepsilon)
=
p\,\lambda(\varepsilon)^{p-1}d(\varepsilon)
+
p(p-1)\,\lambda(\varepsilon)^{p-2}a(\varepsilon)c(\varepsilon)
+
\binom{p}{3}\lambda(\varepsilon)^{p-3}a(\varepsilon)^3.
\]
Since \(|\lambda(\varepsilon)|\) is bounded above and below away from zero,
and \(a(\varepsilon),c(\varepsilon),d(\varepsilon)\) are uniformly bounded,
it follows that
\[
|d_p(\varepsilon)|
\le
C\,p^3|\lambda(\varepsilon)|^{p-3}.
\]
Accordingly,
\[
A_p(\varepsilon)
:=
-\frac{b_p(\varepsilon)}{2\sqrt{\pi}}
=
-\frac{s^{-1/2}}{2\sqrt{\pi}}\,p\,\lambda(\varepsilon)^{p-1}\kappa(\varepsilon),
\]
so that
\begin{equation}\label{eq:Ap-growth-app}
|A_p(\varepsilon)|
\le
C\,p\,|\lambda(\varepsilon)|^{p-1}.
\end{equation}
Moreover, since $|\kappa(\varepsilon)|\ge c_0$ and $|\lambda(\varepsilon)|\ge c_0$,
there is also a uniform lower bound
\begin{equation}\label{eq:Ap-growth-lower-app}
|A_p(\varepsilon)|
\ge
c\,p\,|\lambda(\varepsilon)|^{p-1}.
\end{equation}

The transfer theorem for algebraic singularities in a uniform \(\Delta\)-domain
yields
\[
[z^m]\Delta^{1/2}
=
-\frac{1}{2\sqrt{\pi}}\,m^{-3/2}z_*(\varepsilon)^{-m}
\bigl(1+O(m^{-1})\bigr),
\]
and
\[
[z^m]\Delta^{3/2}
=
O\bigl(m^{-5/2}z_*(\varepsilon)^{-m}\bigr),
\]
uniformly in \(\varepsilon\). Hence
\[
[z^m]\widehat U(z;\varepsilon)^p
=
A_p(\varepsilon)m^{-3/2}z_*(\varepsilon)^{-m}
+
O\left(
p^3|\lambda(\varepsilon)|^{p-3}m^{-5/2}z_*(\varepsilon)^{-m}
\right)
+
[z^m]\mathcal R_p(\cdot;\varepsilon).
\]

It remains to estimate the analytic part. By the regular-part clause in
Definition~\ref{def:uniform-delta}\textup{(U2)}, \(\mathcal R_p(\cdot;\varepsilon)\) extends
holomorphically to the full disk
\[
R(\varepsilon):=|z_*(\varepsilon)|+\eta_0/2,
\qquad |z|<R(\varepsilon),
\]
and satisfies
\[
\sup_{|z|=R(\varepsilon)}|\mathcal R_p(z;\varepsilon)|\le B_0^{\,p}.
\]
Therefore Cauchy's formula on the circle \(|z|=R(\varepsilon)\) gives
\[
[z^m]\mathcal R_p(\cdot;\varepsilon)
=
O\bigl(B_0^{\,p}R(\varepsilon)^{-m}\bigr).
\]
Since \(|\lambda(\varepsilon)|\ge c_0\) and
\[
q_0:=\sup_{0<\varepsilon<\varepsilon_0}\frac{|z_*(\varepsilon)|}{R(\varepsilon)}<1,
\]
one has
\[
B_0^{\,p}R(\varepsilon)^{-m}
\le
C\Bigl(\frac{B_0}{c_0}\Bigr)^p q_0^m
|\lambda(\varepsilon)|^{p-1}|z_*(\varepsilon)|^{-m}.
\]
Set \(B:=B_0/c_0\) and \(\gamma:=-\log q_0>0\). Choose
\(M_{\mathrm{tr}}\ge 2\log(B)/\gamma\). Then, for
\(m\ge M_{\mathrm{tr}}(1+p^2)\),
\[
B^p q_0^m
=
\exp\bigl(p\log B-\gamma m\bigr)
\le
\exp\bigl(p\log B-\gamma M_{\mathrm{tr}}(1+p^2)\bigr)
\le
C e^{-\gamma m/2}.
\]
Since \(\sup_{m\ge1} e^{-\gamma m/2}m^{5/2}<\infty\), it follows that $B^p q_0^m\le C\,m^{-5/2}$ for $m\ge M_{\mathrm{tr}}(1+p^2)$. Therefore
\[
B_0^{\,p}R(\varepsilon)^{-m}
\le
C\,|\lambda(\varepsilon)|^{p-1}m^{-5/2}|z_*(\varepsilon)|^{-m},
\qquad
m\ge M_{\mathrm{tr}}(1+p^2),
\]
uniformly in \(p\), \(m\), and \(\varepsilon\). Combining the singular and analytic contributions, and using
\eqref{eq:Ap-growth-app} together with the lower bound \eqref{eq:Ap-growth-lower-app}, we obtain
\[
R_p(m;\varepsilon)
=
A_p(\varepsilon)\,m^{-3/2}z_*(\varepsilon)^{-m}
\left(
1+O\!\left(\frac{1+p^2}{m}\right)
\right)
\]
uniformly for \(m\ge M_{\mathrm{tr}}(1+p^2)\). Indeed, the absolute error term
\(O\!\left(p^3|\lambda(\varepsilon)|^{p-3}m^{-5/2}|z_*(\varepsilon)|^{-m}\right)\)
is of relative size \(O((1+p^2)/m)\) compared with the leading term because
\(|A_p(\varepsilon)|\asymp p|\lambda(\varepsilon)|^{p-1}\) uniformly in
\(\varepsilon\). Since \(z_*(\varepsilon)=x_*(\varepsilon)^s\), this is exactly
\eqref{eq:uniform-transfer}.
\end{proof}

\subsection{Auxiliary midpoint-circle bounds and rank--one extraction}
\label{app:rank-one-proof}

\begin{proof}[Proof of Lemma~\ref{lem:cauchy-Rp}]
By Cauchy's coefficient formula on the circle \(|x|=\rho(\varepsilon)\),
\[
|R_p(m;\varepsilon)|
=
\bigl|[x^{ms}]\,U(x;\zeta(\varepsilon))^p\bigr|
\le
\rho(\varepsilon)^{-ms}
\sup_{|x|=\rho(\varepsilon)}|U(x;\zeta(\varepsilon))|^p .
\]
It therefore suffices to bound
\[
M(\varepsilon):=\sup_{|x|=\rho(\varepsilon)}|U(x;\zeta(\varepsilon))|
\]
uniformly as \(\varepsilon\downarrow0\).

Write again \(z=x^s\) and \(U(x;\zeta(\varepsilon))=\widehat U(z;\varepsilon)\).
By Definition~\ref{def:uniform-delta}\textup{(U1)}, for
\(0<\varepsilon<\varepsilon_0\) the function \(\widehat U(\cdot;\varepsilon)\)
is analytic on the indented disk \(\Omega_\varepsilon\), whose positively
oriented boundary is \(\Gamma_\varepsilon\), and
\[
\overline{D\bigl(0,\rho(\varepsilon)^s\bigr)}\subset \Omega_\varepsilon,
\qquad
\sup_{z\in\Gamma_\varepsilon}|\widehat U(z;\varepsilon)|\le M
\]
with \(M\) independent of \(\varepsilon\). The maximum-modulus principle on
\(\Omega_\varepsilon\) therefore gives
\[
\sup_{|z|\le \rho(\varepsilon)^s}|\widehat U(z;\varepsilon)|
\le
\sup_{z\in\Gamma_\varepsilon}|\widehat U(z;\varepsilon)|
\le M.
\]
Restricting to \(z=x^s\) with \(|x|=\rho(\varepsilon)\) yields
\(M(\varepsilon)\le M\). Thus
\[
|R_p(m;\varepsilon)|\le M_0^p\,\rho(\varepsilon)^{-ms},
\qquad
M_0:=M,
\]
which is the claimed bound.
\end{proof}

Throughout this subsection fix a symmetry block \(q\in\{1,\dots,s\}\) and set
\[
p_j:=q+js,\qquad j\ge0.
\]
Write
\[
z_*(\varepsilon):=x_*(\varepsilon)^s
=
\rho_*(\varepsilon)^s e^{i\phi(\varepsilon)},
\qquad
\theta(\varepsilon):=\rho_*(\varepsilon)^{-s},
\qquad
\eta_*(\varepsilon):=|z_*(\varepsilon)|^{-2}=\rho_*(\varepsilon)^{-2s}.
\]
Define
\begin{equation}\label{eq:spike-def-r1}
\widetilde d^{(q)}(\varepsilon)(j)
:=
e^{-ij\phi(\varepsilon)}
\frac{s}{\sqrt{p_j}}\,
\frac{\overline{A_{p_j}(\varepsilon)}}{w_j},
\qquad j\ge0.
\end{equation}
By \eqref{eq:Ap-growth} and the choice
\(w_j=p_j^{3/2+\beta}\alpha^{p_j}\), there exists \(r\in(0,1)\) such that
\begin{equation}\label{eq:spike-envelope-r1}
|\widetilde d^{(q)}(\varepsilon)(j)|
\le
C\,p_j^{-1-\beta}r^{p_j},
\end{equation}
uniformly for small \(\varepsilon\). In particular,
\begin{equation}\label{eq:spike-moment-r1}
(1+p_j^2)\,\widetilde d^{(q)}(\varepsilon)(j)\in\ell^2(\mathbb N_0)
\end{equation}
uniformly in \(\varepsilon\).
\begin{proof}[Proof of Proposition~\ref{prop:rank-one-extraction}]
By Hermitian symmetry it is enough to treat the case \(j_2\ge j_1\). Set
\[
\Delta:=j_2-j_1,
\qquad
P:=\max\{p_{j_1},p_{j_2}\},
\qquad
M:=\left\lceil M_{\mathrm{tr}}(1+P^2)\right\rceil .
\]
Proposition~\ref{prop:analytic-gram} gives
\begin{equation}\label{eq:Gtilde-entry-r1}
\widetilde G^{(q)}_{j_1j_2}(\varepsilon)
=
\frac{1}{w_{j_1}w_{j_2}}
\sum_{m\ge0}
\frac{(p_{j_2}+sm)^2}{\sqrt{p_{j_1}p_{j_2}}}\,
\overline{R_{p_{j_1}}(m+\Delta;\varepsilon)}\,
R_{p_{j_2}}(m;\varepsilon).
\end{equation}
We decompose the sum into the tail \(m\ge M\) and the finite segment
\(0\le m<M\).

\smallskip
\noindent
\textit{Step 1: the tail \(m\ge M\).}
For \(m\ge M\), Lemma~\ref{lem:uniform-transfer} yields, uniformly in \(p\) and
\(\varepsilon\),
\[
R_p(m;\varepsilon)
=
A_p(\varepsilon)m^{-3/2}z_*(\varepsilon)^{-m}
+
O\!\Bigl(
|A_p(\varepsilon)|(1+p^2)m^{-5/2}|z_*(\varepsilon)|^{-m}
\Bigr).
\]
Substituting this into the tail of \eqref{eq:Gtilde-entry-r1}, the product of
the principal parts is
\[
\widetilde d^{(q)}_{j_1}(\varepsilon)\,
\overline{\widetilde d^{(q)}_{j_2}(\varepsilon)}\,
\theta(\varepsilon)^\Delta
\frac{(p_{j_2}+sm)^2}{s^2\,m^{3/2}(m+\Delta)^{3/2}}\,
\eta_*(\varepsilon)^m .
\]
Indeed,
\[
\overline{z_*(\varepsilon)}^{-\Delta}
=
\rho_*(\varepsilon)^{-s\Delta}e^{i\Delta\phi(\varepsilon)}
=
\theta(\varepsilon)^\Delta e^{i\Delta\phi(\varepsilon)},
\]
and the remaining phase is exactly the one already incorporated into
\(
\widetilde d^{(q)}_{j_1}(\varepsilon)
\overline{\widetilde d^{(q)}_{j_2}(\varepsilon)}
\).

Since \(m\ge M_{\mathrm{tr}}(1+P^2)\) and \(\Delta\le P\), we have
\[
u:=\frac{p_{j_2}}{sm}
=
O\!\left(\frac{1+P^2}{m+1}\right),
\qquad
v:=\frac{\Delta}{m}
=
O\!\left(\frac{1+P^2}{m+1}\right),
\]
uniformly in \(j_1,j_2\). Hence
\[
\frac{(p_{j_2}+sm)^2}{s^2\,m^{3/2}(m+\Delta)^{3/2}}
=
\frac1m(1+u)^2(1+v)^{-3/2}
=
\frac1m\Bigl(1+O\!\left(\frac{1+P^2}{m+1}\right)\Bigr),
\]
and therefore
\[
\frac{(p_{j_2}+sm)^2}{s^2\,m^{3/2}(m+\Delta)^{3/2}}
=
\frac{1}{m+1}
+
O\!\left(\frac{1+P^2}{(m+1)^2}\right).
\]
It follows that
\begin{equation}\label{eq:tail-structure-r1}
(\widetilde G^{(q)}_{j_1j_2}(\varepsilon))_{\mathrm{tail}}
=
\widetilde d^{(q)}_{j_1}(\varepsilon)\,
\overline{\widetilde d^{(q)}_{j_2}(\varepsilon)}\,
\theta(\varepsilon)^\Delta
\sum_{m\ge M}\frac{\eta_*(\varepsilon)^m}{m+1}
+
\mathcal R_{j_1j_2}(\varepsilon),
\end{equation}
where \(\mathcal R_{j_1j_2}(\varepsilon)\) contains both the prefactor reduction
error and all terms in which at least one coefficient is replaced by its
transfer remainder. Each such term gains one additional factor
\((m+\delta)^{-1}\), and therefore is bounded by
\[
C\,(1+p_{j_1}^2)(1+p_{j_2}^2)\,
|\widetilde d^{(q)}_{j_1}(\varepsilon)|\,
|\widetilde d^{(q)}_{j_2}(\varepsilon)|\,
\frac{\eta_*(\varepsilon)^m}{(m+1)^2}.
\]
After summation over \(m\ge M\) and use of \eqref{eq:spike-moment-r1}, we obtain
\begin{equation}\label{eq:tail-remainder-envelope-r1}
|\mathcal R_{j_1j_2}(\varepsilon)|
\le
C\,(1+p_{j_1}^2)(1+p_{j_2}^2)\,
|\widetilde d^{(q)}_{j_1}(\varepsilon)|\,
|\widetilde d^{(q)}_{j_2}(\varepsilon)|.
\end{equation}
Thus, with
\[
a_j(\varepsilon):=C(1+p_j^2)|\widetilde d^{(q)}_j(\varepsilon)|,
\]
the sequence \(a(\varepsilon)\) belongs to \(\ell^2(\mathbb N_0)\) uniformly in
\(\varepsilon\), by \eqref{eq:spike-moment-r1}, and
\(
|\mathcal R_{j_1j_2}(\varepsilon)|
\le a_{j_1}(\varepsilon)a_{j_2}(\varepsilon)
\).
Hence the tail remainder defines a uniformly Hilbert--Schmidt, and therefore
uniformly bounded, operator on \(\ell^2(\mathbb N_0)\).

\smallskip
\noindent
\textit{Step 2: the finite segment \(0\le m<M\).}
By Lemma~\ref{lem:cauchy-Rp},
\[
|R_p(m;\varepsilon)|
\le
M_0^p\rho(\varepsilon)^{-ms},
\qquad
M_0:=\sup_{0<\varepsilon<\varepsilon_0}M(\varepsilon)<\infty.
\]
Since \(\rho(\varepsilon)>1\), this gives
\[
|R_{p_{j_1}}(m+\Delta;\varepsilon)R_{p_{j_2}}(m;\varepsilon)|
\le
M_0^{p_{j_1}+p_{j_2}}.
\]
Moreover \(m<M\asymp 1+P^2\), so there are \(O(1+P^2)\) terms, and
\(p_{j_2}+sm=O(1+P^2)\). Returning to \eqref{eq:Gtilde-entry-r1}, we find
\[
\bigl|(\widetilde G^{(q)}_{j_1j_2}(\varepsilon))_{\mathrm{fin}}\bigr|
\le
C\,(1+P^2)^3\,p_{j_1}^{-2-\beta}p_{j_2}^{-2-\beta}
\Bigl(\frac{M_0}{\alpha}\Bigr)^{p_{j_1}+p_{j_2}}.
\]
Since \(P\le p_{j_1}+p_{j_2}\), the right-hand side is bounded by
\(h(j_1)h(j_2)\) for some \(h\in\ell^2(\mathbb N_0)\), independent of
\(\varepsilon\). Hence the finite segment also defines a uniformly
Hilbert--Schmidt operator.

\smallskip
\noindent
\textit{Step 3: the omitted initial part of the logarithmic sum.}
In \eqref{eq:tail-structure-r1} the logarithmic factor starts at \(m=M\). To
replace it by the full sum \(L_*(\varepsilon)\), it remains to estimate
\[
\widetilde d^{(q)}_{j_1}(\varepsilon)\,
\overline{\widetilde d^{(q)}_{j_2}(\varepsilon)}\,
\theta(\varepsilon)^\Delta
\sum_{m=0}^{M-1}\frac{\eta_*(\varepsilon)^m}{m+1}.
\]
Since \(M=\lceil M_{\mathrm{tr}}(1+P^2)\rceil\),
\[
\sum_{m=0}^{M-1}\frac{\eta_*(\varepsilon)^m}{m+1}
\le
\sum_{m=0}^{M-1}\frac1{m+1}
\le
C\bigl(1+\log(1+P)\bigr),
\]
uniformly in \(\varepsilon\). On the other hand,
\eqref{eq:spike-envelope-r1} gives
\[
|\widetilde d^{(q)}_j(\varepsilon)|
\le
C\,p_j^{-1-\beta}\Bigl(\frac{L_0}{\alpha}\Bigr)^{p_j},
\qquad j\ge0.
\]
Therefore
\[
h(j):=
C\,(1+\log(1+p_j))\,p_j^{-1-\beta}
\Bigl(\frac{L_0}{\alpha}\Bigr)^{p_j}
\]
belongs to \(\ell^2(\mathbb N_0)\), uniformly in \(\varepsilon\), and since
\(P\le p_{j_1}+p_{j_2}\),
\[
\bigl(1+\log(1+P)\bigr)
|\widetilde d^{(q)}_{j_1}(\varepsilon)|
|\widetilde d^{(q)}_{j_2}(\varepsilon)|
\le
h(j_1)h(j_2).
\]
Thus the omitted initial part also defines a uniformly Hilbert--Schmidt
operator.

Combining Steps 1--3, we obtain
\begin{equation}\label{eq:tail-toeplitz-r1}
\widetilde G^{(q)}(\varepsilon)
=
L_*(\varepsilon)\,
D(\varepsilon)T_{\theta(\varepsilon)}D(\varepsilon)^*
+
\widetilde B^{(q)}(\varepsilon),
\qquad
\sup_{\varepsilon}\|\widetilde B^{(q)}(\varepsilon)\|<\infty,
\end{equation}
where \(D(\varepsilon)\) is the diagonal operator with diagonal
\(\widetilde d^{(q)}(\varepsilon)(j)\),
\[
(T_{\theta})_{j_1j_2}:=\theta^{|j_1-j_2|},
\qquad
L_*(\varepsilon):=\sum_{m\ge0}\frac{\eta_*(\varepsilon)^m}{m+1}.
\]

\smallskip
\noindent
\textit{Step 4: removal of the Toeplitz factor.}
The matrix of
\(D(\varepsilon)(T_{\theta(\varepsilon)}-I)D(\varepsilon)^*\) has entries
\[
\widetilde d^{(q)}_{j_1}(\varepsilon)\,
\bigl(\theta(\varepsilon)^{|j_1-j_2|}-1\bigr)\,
\overline{\widetilde d^{(q)}_{j_2}(\varepsilon)},
\]
so
\[
\|D(\varepsilon)(T_{\theta(\varepsilon)}-I)D(\varepsilon)^*\|_{\mathrm{HS}}^2
=
\sum_{j_1,j_2\ge0}
|\widetilde d^{(q)}_{j_1}(\varepsilon)|^2
|\widetilde d^{(q)}_{j_2}(\varepsilon)|^2
\bigl|\theta(\varepsilon)^{|j_1-j_2|}-1\bigr|^2.
\]
Since \(|\theta^n-1|\le n(1-\theta)\) for \(0<\theta<1\),
\[
\|D(\varepsilon)(T_{\theta(\varepsilon)}-I)D(\varepsilon)^*\|_{\mathrm{HS}}^2
\le
(1-\theta(\varepsilon))^2
\sum_{j_1,j_2\ge0}
|\widetilde d^{(q)}_{j_1}(\varepsilon)|^2
|\widetilde d^{(q)}_{j_2}(\varepsilon)|^2
|j_1-j_2|^2.
\]
Because \(p_j=q+js\), one has \(j\asymp p_j\), hence
\[
|j_1-j_2|^2
\le
2(j_1^2+j_2^2)
\le
C\,(p_{j_1}^2+p_{j_2}^2).
\]
The double sum is therefore controlled by \eqref{eq:spike-moment-r1}, and is
uniformly finite. We conclude that
\[
\|D(\varepsilon)(T_{\theta(\varepsilon)}-I)D(\varepsilon)^*\|_{\mathrm{HS}}
\le
C\,(1-\theta(\varepsilon)).
\]
Now \(\theta(\varepsilon)=\rho_*(\varepsilon)^{-s}\uparrow1\), so
\[
1-\theta(\varepsilon)=O(\rho_*(\varepsilon)-1),
\]
whereas
\[
L_*(\varepsilon)
=
\sum_{m\ge0}\frac{\eta_*(\varepsilon)^m}{m+1}
=
\log\frac{1}{\rho_*(\varepsilon)-1}+O(1).
\]
Hence
\[
L_*(\varepsilon)\,
\|D(\varepsilon)(T_{\theta(\varepsilon)}-I)D(\varepsilon)^*\|_{\mathrm{HS}}
\longrightarrow 0.
\]
It follows that
\[
L_*(\varepsilon)\,
D(\varepsilon)T_{\theta(\varepsilon)}D(\varepsilon)^*
=
L_*(\varepsilon)\,
\widetilde d^{(q)}(\varepsilon)\otimes \widetilde d^{(q)}(\varepsilon)
+
o(1)
\]
in Hilbert--Schmidt norm.

\smallskip
\noindent
\textit{Step 5: replacement of \(L_*\) by \(L\), and convergence of the spike.}
Since
\[
\eta(\varepsilon)=\rho(\varepsilon)^{-2s}\eta_*(\varepsilon),
\]
one has \(L_*(\varepsilon)-L(\varepsilon)=O(1)\). This contributes only a
bounded rank-one term and can therefore be absorbed into the remainder. Thus
\[
\widetilde G^{(q)}(\varepsilon)
=
L(\varepsilon)\,
\widetilde d^{(q)}(\varepsilon)\otimes \widetilde d^{(q)}(\varepsilon)
+
\widetilde C^{(q)}(\varepsilon),
\qquad
\sup_{\varepsilon}\|\widetilde C^{(q)}(\varepsilon)\|<\infty.
\]

Finally, Lemma~\ref{lem:critical-coeff-asymp} gives
\(A_{p_j}(\varepsilon)\to A_{p_j}^{(c)}\) for each fixed \(j\), while
\(z_*(\varepsilon)\to z_{*,c}\neq0\). After choosing a continuous local
argument, \(\phi(\varepsilon)\to\phi_c\). The uniform envelope
\eqref{eq:spike-envelope-r1} therefore yields, by dominated convergence,
\[
\widetilde d^{(q)}(\varepsilon)\to \widetilde d_c^{(q)}
\qquad\text{in }\ell^2(\mathbb N_0),
\]
with
\[
\widetilde d_c^{(q)}(j)
=
e^{-ij\phi_c}\frac{s}{\sqrt{p_j}}\,
\frac{\overline{A_{p_j}^{(c)}}}{w_j}.
\]
Since \(A_{p_j}^{(c)}\neq0\) for every \(j\), the limit vector
\(\widetilde d_c^{(q)}\) is nonzero.
\end{proof}

\subsection{Compact remainder and norm convergence}
\label{app:compact-remainder}

\begin{lemma}[Analytic dependence of fixed Taylor coefficients]
\label{lem:fixed-coeff-analytic}
For each fixed pair of integers $p\ge 1$ and $m\ge 0$, the coefficient
$R_p(m;\zeta)=[x^{ms}]U(x;\zeta)^p$ depends analytically on \(\zeta\) in a
neighborhood of any subcritical parameter value for which the Taylor branch of
\eqref{eq:inv-eq} is defined.
\end{lemma}

\begin{proof}
By Proposition~\ref{prop:inv-eq}, the Taylor branch is the unique analytic germ
solving
\[
F(U,x;\zeta)=U-1-\sum_{n=1}^N \zeta_n x^{s_n}U^{s_n}=0
\]
with $U(0;\zeta)=1$. Since $\partial_U F(1,0;\zeta)=1$, the analytic implicit
function theorem yields a jointly analytic germ \((x,\zeta)\mapsto U(x;\zeta)\)
near \((0,\zeta_0)\) for every fixed subcritical \(\zeta_0\). Therefore
$U(x;\zeta)^p$ is jointly analytic as well, and Cauchy's coefficient formula on
any sufficiently small circle around $x=0$ shows that each fixed coefficient
$R_p(m;\zeta)$ is analytic in \(\zeta\).
\end{proof}

\begin{proof}[Proof of Lemma~\ref{lem:compact-remainder}]
We retain the notation of
Appendix~\ref{app:rank-one-proof}. By the proof of
Proposition~\ref{prop:rank-one-extraction}, both the tail remainder
\(\mathcal R(\varepsilon)\) from \eqref{eq:tail-structure-r1} and the finite
segment of \eqref{eq:Gtilde-entry-r1} admit \(\varepsilon\)-independent
rank-one envelopes \(a(j_1)a(j_2)\) and \(b(j_1)b(j_2)\) with
\(a,b\in\ell^2(\mathbb N_0)\). Consequently, both families are uniformly
Hilbert--Schmidt.

For each fixed \((j_1,j_2)\), the corresponding matrix entries converge as
\(\varepsilon\downarrow0\). For the tail remainder, write the entry as the
infinite \(m\)-sum given by \eqref{eq:tail-structure-r1} after subtraction of the
rank-one model term. For each fixed \(m\), the summand converges by
Lemma~\ref{lem:critical-coeff-asymp}, together with the continuity of the
finitely many fixed Taylor coefficients \(R_p(m;\varepsilon)\) appearing in the
cutoff contribution. Moreover, the proof of
Proposition~\ref{prop:rank-one-extraction} already provides an
\(\varepsilon\)-independent summable envelope for this summand, so dominated
convergence applies termwise in \(m\). For the finite segment, the same follows
from Lemma~\ref{lem:fixed-coeff-analytic}, which gives analytic dependence of each fixed coefficient \(R_p(m;\zeta)\) on the
parameter, together with the \(\varepsilon\)-independent Hilbert--Schmidt
envelope established there. Dominated convergence with these Hilbert--Schmidt
envelopes therefore gives Hilbert--Schmidt convergence of both pieces.

Next, the proof of
Proposition~\ref{prop:rank-one-extraction} shows that
\[
L_*(\varepsilon)\,
D(\varepsilon)T_{\theta(\varepsilon)}D(\varepsilon)^*
=
L_*(\varepsilon)\,
\widetilde d^{(q)}(\varepsilon)\otimes \widetilde d^{(q)}(\varepsilon)
+
o(1)
\]
in Hilbert--Schmidt norm. Combining this with the Hilbert--Schmidt convergence
of the bounded remainder from that proof, we conclude that
\[
\widetilde G^{(q)}(\varepsilon)
-
L_*(\varepsilon)\,
\widetilde d^{(q)}(\varepsilon)\otimes \widetilde d^{(q)}(\varepsilon)
\]
converges in Hilbert--Schmidt norm to a Hilbert--Schmidt limit.

Finally, \(L_*(\varepsilon)-L(\varepsilon)\) has a finite limit
(\(= \log(3/2)\) with the present choice \(\rho(\varepsilon)=(1+\rho_*(\varepsilon))/2\)),
and \(\widetilde d^{(q)}(\varepsilon)\to \widetilde d_c^{(q)}\) in \(\ell^2\).
Hence the rank-one correction
\[
(L_*(\varepsilon)-L(\varepsilon))\,
\widetilde d^{(q)}(\varepsilon)\otimes \widetilde d^{(q)}(\varepsilon)
\]
also converges in Hilbert--Schmidt norm. Therefore
\[
\widetilde C^{(q)}(\varepsilon)
=
\widetilde G^{(q)}(\varepsilon)
-
L(\varepsilon)\,
\widetilde d^{(q)}(\varepsilon)\otimes \widetilde d^{(q)}(\varepsilon)
\]
converges in Hilbert--Schmidt norm, and in particular in operator norm, to a
compact limit. This proves the lemma.
\end{proof}

% \input{declaration}
%=============================================================================
% Bibliography
%=============================================================================

\end{document}